\documentclass[11pt]{article}
\usepackage[margin=1in]{geometry}

\usepackage{algorithm}
\usepackage{algpseudocode}
\usepackage{mathrsfs}
\usepackage{booktabs}
\usepackage{indentfirst}

\usepackage{amssymb}
\usepackage{amsthm}
\usepackage{amsmath,amsfonts}
\usepackage{graphicx}
\usepackage{CJK}
\usepackage{amsmath}
\usepackage{cases}
\usepackage{arydshln}
\usepackage{subcaption}           
\usepackage{diagbox}
\usepackage{bbm}
\usepackage{multirow}
\usepackage{float}
\usepackage{url}
\usepackage{makecell}

\usepackage[dvipsnames,svgnames]{xcolor}
\usepackage{hyperref}
\usepackage{cleveref}
\hypersetup{
    colorlinks=true,
    linkcolor=magenta,
    filecolor=magenta,
    citecolor=DarkRed,
    urlcolor=cyan,
}

\usepackage{authblk}

\algnewcommand{\Meta}[1]{\State \enspace\;\textbf{Meta parameter:} #1}

\newtheorem{definition}{Definition}
\newtheorem{example}{Example}[section]
\newtheorem{theorem}{Theorem}
\newtheorem{remark}{Remark}

\newtheorem{property}{Property}

\newtheorem{lemma}{Lemma}

\newcommand{\remove}[1]{}

\begin{document}

\title{Decomposition-Based Optimal Bounds for Privacy Amplification via Shuffling}

\author[a]{Pengcheng Su }
\author[b]{Haibo Cheng* }
\author[b]{Ping Wang* }
\affil[a]{School of Computer Science, Peking University }
\affil[b]{National Engineering Research Center for Software Engineering, Peking University  \newline \texttt{\textnormal{ \{pcs, hbcheng, pwang\}@pku.edu.cn } } } 

\renewcommand*{\Affilfont}{\small\it} 
\renewcommand\Authands{ and } 

\date{}
\maketitle

\begin{abstract}
Shuffling has been shown to amplify differential privacy guarantees, enabling a more favorable privacy-utility trade-off. To characterize and compute this amplification, two fundamental analytical frameworks have been proposed: the \emph{privacy blanket} by Balle et al.\ (CRYPTO 2019) and the \emph{clone}—including both the standard and stronger variants—by Feldman et al.\ (FOCS 2021, SODA 2023). These frameworks share a common foundation: decomposing local randomizers into structured components for analysis.

In this work, we introduce a unified analytical framework—the general clone paradigm—which subsumes all possible decompositions, with the clone and blanket decompositions arising as special cases. Within this framework, we identify the optimal decomposition, which is precisely the one used by the privacy blanket. Moreover, we develop a simple and efficient algorithm based on the Fast Fourier Transform (FFT) to compute optimal privacy amplification bounds. Experimental results show that our computed upper bounds nearly match the lower bounds, demonstrating the tightness of our method.
Building on this method, we also derive optimal amplification bounds for both \emph{joint} and \emph{parallel} compositions of LDP mechanisms in the shuffle model.
\end{abstract}

\section{Introduction}

Differential Privacy (DP) has become a foundational framework for safeguarding individual privacy while enabling meaningful data analysis~\cite{Dwork2006}. In real-world applications, Local Differential Privacy (LDP) is widely adopted as it eliminates the need for a trusted curator by applying noise to each user's data before aggregation~\cite{7,8,10,Apple17}. However, this decentralized approach often results in significant utility loss due to excessive noise.

To address this trade-off, the \textit{shuffle model} introduces a trusted shuffler between users and the aggregator~\cite{cheu19,Erlingsson19,Prochlo}. This work considers the \emph{single-message shuffle model}, where each client
sends a report of their data and these reports are then anonymized and randomly shuffled before being sent to the server \cite{feldman2023soda}.
Shuffle DP improves the privacy-utility trade-off while minimizing trust assumptions to a single shuffler, making it a promising model for real-world deployment~\cite{Wang2020,triangle,cheu22}. For instance, privacy amplification by shuffling was used in Apple and Google’s Exposure Notification Privacy-preserving Analytics~\cite{application}.

The ``amplification-by-shuffling" theorem in the shuffle model implies that when each of the \(n\) users randomizes their data using an \(\varepsilon_0\)-LDP mechanism, the collection of shuffled reports satisfies \((\varepsilon(\varepsilon_0, \delta, n), \delta)\)-DP, where \(\varepsilon(\varepsilon_0, \delta, n) \ll \varepsilon_0\) for sufficiently large \(n\) and not too small $\delta$~\cite{feldman2023soda}. A central theoretical challenge is to characterize and compute this privacy amplification effect. A tighter bound allows for a larger value of \(\varepsilon_0\) while still achieving \((\varepsilon, \delta)\)-DP after shuffling, thereby improving the utility of the mechanism. It also enables a more accurate analysis of the overall \((\varepsilon, \delta)\)-guarantees resulting from shuffling \(n\) independent \(\varepsilon_0\)-LDP reports.

Researchers are interested in two types of upper bounds on privacy amplification. The first is a \emph{generic bound}, which provides a uniform guarantee for all $\varepsilon_0$-DP local randomizers. The second is a \emph{specific bound}, which is tailored for a given $\varepsilon_0$-DP local randomizer.
Generic bounds are primarily studied from a theoretical perspective, as they allow the derivation of asymptotic expressions for $\varepsilon(\varepsilon_0, n, \delta)$ \cite{Erlingsson19,Balle2019,Feldman2021}. In contrast, specific bounds are usually much tighter and are of significant practical importance, as they provide precise guarantees for concrete mechanisms used in real-world applications.

Among the various strategies proposed to analyze this effect~\cite{cheu19,Erlingsson19,checkin}, decomposition-based approaches are the most popular \cite{Balle2019,Feldman2021,feldman2023soda}. Two prominent frameworks are the \textbf{privacy blanket} by Balle et al.~\cite{Balle2019} and the \textbf{clone paradigm} by Feldman et al., which includes the \textit{standard clone} \cite{Feldman2021} and the \textit{stronger clone} \cite{feldman2023soda}. These approaches rely on decomposing the probability distributions induced by a local randomizer under different inputs. A decomposition naturally leads to a reduction, which yields an upper bound on the privacy amplification via shuffling.

The privacy blanket framework designs a tailored decomposition to obtain a bound for each specific local randomizer \cite{Balle2019}. However, this bound cannot be computed precisely in its original form. To address this, the framework resorts to further approximations—such as Hoeffding's and Bennett's inequalities—which yield looser but computable bounds. A computable generic bound is also derived in a similar manner.

The standard clone~\cite{Feldman2021} adopts a simple, unified decomposition, resulting in a generic bound that applies to all \(\varepsilon_0\)-DP local randomizers. This bound can be computed precisely using a dedicated numerical algorithm and empirically outperforms the computable generic bound derived from the privacy blanket framework. However, the standard clone does not provide specific bounds, except in the special case of $k$-ary Randomized Response.

The stronger clone was introduced with a more refined decomposition and was expected to yield improved bounds for both generic and specific local randomizers~\cite{feldman2023soda}. Unfortunately, a critical flaw was later identified in the proof's core lemma. A corrected version was released on arXiv~\cite{feldman2023arxiv}, which showed that the original generic bounds hold only for a restricted class of local randomizers. The specific bounds were also revised and replaced with a version that is much weaker and lacks an efficient computation method.

\subsection{Our contributions}\label{sec:contributions}
In this paper, we focus on computing \emph{specific bounds}. We provide the \textbf{optimal bound} achievable via any decomposition-based method for a given randomizer. Furthermore, we introduce an efficient numerical algorithm that can compute this optimal bound. In addition, we present methods for computing optimal amplification bounds for both joint composition and parallel composition in the shuffle model, achieving substantially tighter results than existing approaches.

\paragraph{General Clone and Optimal Bound.}
In prior work on Shuffle DP, the specific bounds of the blanket framework has seen limited use; instead, its associated analysis has mostly been applied in specific cases when the blanket distribution is uniform, as in $k$-ary randomized response ($k$-RR) \cite{Balle2019}, optimal local hash \cite{Wang2020} or in multi-message shuffle models where uniform blanket noise is directly injected \cite{Luo22}. Although the clone and blanket appear somewhat similar, \emph{no formal connection between them has been established.}

In this work, we present a unified analytical framework—termed the \emph{general clone}—which encompasses all possible decompositions, including both the clone and the blanket as special cases. We further demonstrate that, among all such decompositions, \emph{the optimal one corresponds precisely to the privacy blanket}.

\paragraph{Efficient Computation of Optimal Bounds.}
The decomposition bound derived from the privacy blanket, in its original form, is not directly computable. To address this, we reinterpret it through the lens of the general clone framework and derive a simplified representation that admits efficient numerical evaluation via the Fast Fourier Transform (FFT).

Specifically, we extend the Privacy Amplification Random Variable (PARV) introduced in the privacy blanket literature to define the \emph{Generalized Privacy Amplification Random Variable} (GPARV). In this formulation, the privacy loss corresponds to the expected positive part of the sum of \( n \) independent and identically distributed GPARVs $G_i$:
\[
\frac{1}{n}\mathbb{E}\left[\max\left\{0, \sum_{i=1}^n G_i \right\}\right].
\]
This expression naturally takes the form of a convolution, enabling efficient evaluation via FFT.

Our method yields the best-known bounds achievable through decomposition-based techniques. Empirical results show that the computed upper bounds closely match the corresponding lower bounds, demonstrating both tightness and reliability. Notably, the new bounds improve the resulting \(\varepsilon\) by at least 10\% across all settings, and by up to 50\% when \(\varepsilon_0\) is large and \(n\) is small.

\paragraph{Optimal bounds for joint composition.}
We conduct the first systematic analysis of \textit{joint composition} in the shuffle model using our algorithm. In classical DP, $k$-fold composition refers to applying $k$ independent mechanisms to the \emph{same} data:
\[
\mathcal{M}_{kFold}(D) = (\mathcal{M}_1(D), \mathcal{M}_2(D), \dots, \mathcal{M}_k(D)).
\]
In contrast, our notion of joint composition applies independent mechanisms to \emph{different} data:
\[
\mathcal{M}_{joint}(D_1, D_2, \dots, D_k) = (\mathcal{M}_1(D_1), \mathcal{M}_2(D_2), \dots, \mathcal{M}_k(D_k)).
\]
These two definitions are interconvertible: treating $(D_1, D_2, \ldots, D_k)$ as a single data $D$ yields joint composition as a special case of $k$-fold composition; conversely, setting $D_1 = D_2 = \cdots = D_k = D$ makes $k$-fold composition a special case of joint composition.
Our analysis of joint composition also applies to the $k$-fold composition setting. However, the joint composition framework captures a strictly more general scenario—namely, when $D_1, D_2, \dots, D_k$ are not necessarily identical and may even belong to different domains. It is important to note that our setting concerns the composition of LDP mechanisms in the single-message shuffle model, in which \( D \) or the tuple \( (D_1, D_2, \dots, D_k) \) corresponds to the data held by an individual user. Moreover, during the shuffling process, these outputs are treated as a single tuple and are shuffled together as a unit.

Joint composition is widely used in LDP applications such as joint distribution estimation and heavy hitter detection~\cite{50,51,5,kikuchi2022castellscalablejointprobability}. For example, when each user's data contains $d$ attributes, applying an \( \frac{\varepsilon_0}{d} \)-LDP mechanism to each attribute ensures overall \( \varepsilon_0 \)-LDP while preserving inter-attribute correlations.
Our experimental results show that existing methods yield relatively loose bounds in this setting, whereas our algorithm computes significantly tighter results by leveraging the optimal bounds.

\paragraph{Optimal bounds for parallel composition.}
Additionally, we show how to compute the optimal bounds for the \emph{parallel composition} of LDP mechanisms. Parallel composition refers to the setting where each user randomly selects a differentially private mechanism $\mathcal{R}_i$ with probability $p_i$ and applies it to their own data:
\begin{align*}
\mathcal{R}_{[1,m]}^{\parallel,\boldsymbol{p}}(x)=(i,\mathcal{R}_i(x)), w.p.\ p_i,\quad i=1,2,\dots,m
\end{align*}
where \(\boldsymbol{p} = (p_1, \dots, p_m)\) is a probability vector with \(\sum_{i=1}^m p_i = 1\).

It is particularly suited for scenarios in which the analyst aims to compute \emph{multiple} statistical queries—each corresponding to a distinct local mechanism $\mathcal{R}_1, \dots, \mathcal{R}_m$. Parallel composition partitions the user population into $m$ non-overlapping subsets, assigning each subset to one of the $m$ estimation tasks. As a result, each subset can utilize the full privacy budget $\varepsilon_0$, rather than splitting the budget across all tasks. This technique has been widely adopted in the literature~\cite{4,8731512,10.1145/3299869.3319891,8758350,10.1145/3183713.3196906} to achieve better utility compared to naive budget division.

It is worth noting that parallel composition also characterizes \textit{Poisson subsampling} in the shuffle model \cite{subsample18,steinke2022compositiondifferentialprivacy,Balle_Barthe_Gaboardi_2020}: each user participates in a mechanism $\mathcal{R}^*$ with probability $p^*$, which is equivalent to the parallel composition $\mathcal{R}_{[1,2]}^{\parallel,\boldsymbol{p}}(x)$, where $\mathcal{R}_1(x) = \mathcal{R}^*(x)$, $\mathcal{R}_2(x) = \bot$, and $p_1 = p^*$, $p_2 = 1 - p^*$. 

To compute the privacy amplification of parallel composition in the shuffle model was previously studied in~\cite{tv_shuffle}, but their analysis inherited a technical flaw from~\cite{feldman2023soda}, resulting in incorrect conclusions. Our method provides the first correct computation of the optimal bound in this setting.

\subsection{Related Work}\label{related_work}

We now discuss prior work most relevant to the technical contributions of our paper. Koskela et al.~\cite{koskela2023numerical} investigate \emph{$k$-fold composition} of Shuffle DP protocols. Their approach computes the Privacy Loss Random Variable (PLRV) of a dominating pair for a single-round shuffle protocol, and then applies the FFT to analyze \emph{multiple rounds} (i.e., composition over $k$ executions)~\cite{koskela20b}. In contrast, our work focuses on \emph{joint composition}, where a \emph{single-round} shuffle is applied to a \emph{composite local randomizer} constructed from multiple sub-randomizers.

Some subtle aspects warrant further clarification, particularly regarding how coupling among sub-randomizers affects the composition analysis.
In the single-message shuffle model, when the local randomizer is a joint composition, the outputs of the sub-local randomizers for a single user are inherently coupled. From a theoretical standpoint, privacy analysis in this setting ($\mathcal{S}\circ (\mathcal{M}_1\times \mathcal{M}_2\times \dots \times M_k)^n$) is not equivalent to applying composition to the individual mechanisms $\mathcal{R}_i$ after shuffling $(\mathcal{S}_1\circ \mathcal{M}_1^n,\mathcal{S}_2\circ \mathcal{M}_2^n,\dots,\mathcal{S}_k\circ \mathcal{M}_k^n)$. Furthermore, the standard compositional properties of the PLRV do not apply.

In~\cite{koskela2023numerical}, computing the PLRV requires a dominating pair~\cite{Zhu22}, obtained either from the \emph{stronger clone} (in the general case) or from the \emph{privacy blanket} (in the case of $k$-RR). In the general case, \cite{koskela2023numerical} inherits a flaw that originates from the original version of the stronger clone construction (see Section~\ref{sec:stronger_clone}). For any specific local randomizer, our work demonstrates that the privacy bound derived from the blanket framework is strictly tighter than that obtained via the clone.

Theoretically, computing the PLRV of $(P_0^\mathsf{B}, P_1^\mathsf{B})$ (as defined in Section~\ref{sec_blanket}) and applying the Hockey-Stick divergence $D_{e^\varepsilon}(P_0^\mathsf{B}, P_1^\mathsf{B})$ yields the same amplification bound as computing it directly from the GPARV; both approaches evaluate $D_{e^\varepsilon}(P_0^\mathsf{B}, P_1^\mathsf{B})$. However, for general local randomizers, the PLRV of $(P_0^\mathsf{B}, P_1^\mathsf{B})$ is analytically intractable due to the complex structure of the blanket distribution. GPARV overcomes this limitation by expressing the privacy loss in a convolution-friendly form, enabling efficient FFT-based evaluation even in the absence of a closed-form expression for the PLRV.

\paragraph{Subsequent work.}
A recent paper to appear at CSF 2026 extends our decomposition framework to the analysis of re-identification attacks in the shuffle model, showing that the same hierarchy of decompositions—where the blanket decomposition remains optimal—also holds in this context \cite{Su26CSF}.

\subsection{Roadmap}

Section~\ref{sec:prelim} introduces the definitions of differential privacy and the shuffle model. In Section~\ref{sec_review}, we briefly review the clone and blanket frameworks. Section~\ref{sec_general_clone} presents our general clone analysis framework, demonstrating that both the clone and blanket decompositions are special cases within it. We prove the optimality of the blanket decomposition and, due to the redundancy of the original construction, propose a method for computing the \emph{Simplified Optimal Decomposition}. For both joint and parallel compositions, we show that their optimal decompositions are given by the Cartesian product and the weighted combination of the optimal decompositions of their respective sub-local randomizers.

Section~\ref{sec_our_algo} introduces the Generalized Privacy Amplification Random Variable (GPARV) and an FFT-based algorithm for computing the optimal privacy amplification bound. We also clarify the relationship between GPARV and the original PARV introduced by Balle et al.~\cite{Balle2019}. Section~\ref{sec_experiment} presents our experimental results. 

Sections~\ref{sec_discussion} and~\ref{sec_limitation} discuss potential directions beyond decomposition-based techniques, as well as limitations of the current work. Finally, Appendix~\ref{appenix_decomposition} provides detailed derivations for computing the optimal decompositions of several commonly used local randomizers.

\section{Preliminaries} \label{sec:prelim}

Differential privacy is a privacy-preserving framework for randomized algorithms. Intuitively, an algorithm is differentially private if the output distribution does not change significantly when a single individual's data is modified. This ensures that the output does not reveal substantial information about any individual in the dataset. The hockey-stick divergence is commonly used to define $(\varepsilon,\delta)$-DP.

\begin{definition}[Hockey-Stick Divergence]
The \textit{hockey-stick} divergence between two random variables \( P \) and \( Q \) is defined as:
\[
D_{\alpha}(P \parallel Q) = \int \max\{0, P(x) - \alpha Q(x)\}\,\mathrm{d}x,
\]
where we use the notation $P$ and $Q$ to refer to both the random
variables and their probability density functions.
\end{definition}

We say that \( P \) and \( Q \) are $(\varepsilon, \delta)$-indistinguishable if:
\[
\max\{D_{e^\varepsilon}(P \parallel Q), D_{e^\varepsilon}(Q \parallel P)\} \leq \delta.
\]
If two datasets \( X^0 \) and \( X^1 \) have the same size and differ only by the data of a single individual, they are referred to as neighboring datasets (denoted by \( X^0 \simeq X^1 \)).

\begin{definition}[Differential Privacy]
An algorithm \( \mathcal{R} : \mathbb{X}^n \to \mathbb{Z} \) satisfies $(\varepsilon, \delta)$-differential privacy if for all neighboring datasets \( X, X' \in \mathbb{X}^n \), \( \mathcal{R}(X) \) and \( \mathcal{R}(X') \) are $(\varepsilon, \delta)$-indistinguishable.    
\end{definition}

\begin{definition}[Local Differential Privacy]
An algorithm \( \mathcal{R} : \mathbb{X} \to \mathbb{Y} \) satisfies local $(\varepsilon, \delta)$-differential privacy if for all \( x, x' \in \mathbb{X} \), \( \mathcal{R}(x) \) and \( \mathcal{R}(x') \) are $(\varepsilon, \delta)$-indistinguishable.    
\end{definition}

Here, \( \varepsilon \) is referred to as the privacy budget, which controls the privacy loss, while \( \delta \) allows for a small probability of failure. When \( \delta = 0 \), the mechanism is also called \( \varepsilon \)-DP.

Following conventions in the shuffle model based on randomize-then-shuffle~\cite{Balle2019,cheu19}, we define a single-message protocol \( \mathcal{P} \) in the shuffle model as a pair of algorithms \( \mathcal{P} = (\mathcal{R}, \mathcal{A}) \), where \( \mathcal{R} : \mathbb{X} \to \mathbb{Y} \), and \( \mathcal{A} : \mathbb{Y}^n \to \mathbb{O} \). We call \( \mathcal{R} \) the \textit{local randomizer}, \( \mathbb{Y} \) the \textit{message space} of the protocol, \( \mathcal{A} \) the \textit{analyzer}, and \( \mathbb{O} \) the \textit{output space}. 

The overall protocol implements a mechanism \( \mathcal{P} : \mathbb{X}^n \to \mathbb{O} \) as follows: Each user \( i \) holds a data record \( x_i \), to which they apply the local randomizer to obtain a message \( y_i = \mathcal{R}(x_i) \). The messages \( y_i \) are then shuffled and submitted to the analyzer. Let \( \mathcal{S}(y_1, \dots, y_n) \) denote the random shuffling step, where \( \mathcal{S} : \mathbb{Y}^n \to \mathbb{Y}^n \) is a \textit{shuffler} that applies a random permutation to its inputs. 

In summary, the output of \( \mathcal{P}(x_1, \dots, x_n) \) is given by
\[
\mathcal{A} \circ \mathcal{S} \circ \mathcal{R}^n (\boldsymbol{x}) = \mathcal{A}(\mathcal{S}(\mathcal{R}(x_1), \dots, \mathcal{R}(x_n))).
\]

\begin{definition}[Differential Privacy in the Shuffle Model]
A protocol \( \mathcal{P} = (\mathcal{R}, \mathcal{A}) \) satisfies $(\varepsilon, \delta)$-differential privacy in the shuffle model if for all neighboring datasets \( X, X' \in \mathbb{X}^n \), the distributions \( \mathcal{S} \circ \mathcal{R}^n(X) \) and \( \mathcal{S} \circ \mathcal{R}^n(X') \) are $(\varepsilon, \delta)$-indistinguishable.    
\end{definition}

\section{Review of Existing Analysis Techniques}\label{sec_review}

In this section, we review existing analysis techniques for studying privacy amplification in the shuffle model. We begin by introducing the standard clone due to its simplicity~\cite{Feldman2021}. Next, we restate the privacy blanket framework, using consistent terminology with the former~\cite{Balle2019}. 
It aids in identifying the intrinsic connection between the two approaches, which will be explored in Section \ref{sec_general_clone}.

We then discuss subsequent attempts to extend the clone paradigm, specifically the stronger clone. We highlight both the vision and the shortcomings of its original and corrected versions. The lessons learned from these failures are incorporated into the design of the general clone framework (see Section~\ref{sec:definition_gc}).

\subsection{Standard clone}\label{sec_clone}
The intuition behind the standard clone paradigm is as follows \cite{Feldman2021}: Suppose that $X^0$ and $X^1$ are neighbouring databases that differ on the first datapoint, $x_1^0 \neq x_1^1$. A key observation is that for any $\varepsilon_0$-DP local randomizer $\mathcal{R}$ and data point $x$, $\mathcal{R}(x)$ can be seen as sampling from the same distribution as $\mathcal{R}(x_1^0)$ with probability at least $e^{-\varepsilon_0}/2$ and sampling 
from the same distribution as $\mathcal{R}(x_1^1)$ with probability at least $e^{-\varepsilon_0}/2$. That is, with probability $e^{-\varepsilon_0}$ each data point can 
create a clone of the output of $\mathcal{R}(x_1^0)$ or a clone of $\mathcal{R}(x_1^1)$ with equal probability. Thus $n - 1$ data elements effectively 
produce a random number of clones of both $x_1^0$ and $x_1^1$, making it more challenging to distinguish whether the original dataset contains $x_1^0$ or $x_1^1$ as its first element.

Due to the $\varepsilon_0$-DP property of the local randomizer $\mathcal{R}$, we have the following inequality:
\begin{align*}
\forall x_i \in \mathbb{X},\forall y \in \mathbb{Y}:\Pr[\mathcal{R}(x_i)=y]\ge \frac{1}{e^{\varepsilon_0}}\Pr[\mathcal{R}(x_1^0)=y]\\
\wedge \Pr[\mathcal{R}(x_i)=y]\ge \frac{1}{e^{\varepsilon_0}}\Pr[\mathcal{R}(x_1^1)=y].
\end{align*} Therefore, the local randomizer $\mathcal{R}$ on any input $x_i$ can be decomposed into a mixture of $\mathcal{R}(x_1^0), \mathcal{R}(x_1^1)$ and some ``left-over'' distribution $\text{LO}(x_i)$ such that
$$\mathcal{R}(x_i)=\frac{1}{2e^{\varepsilon_0}}\mathcal{R}(x_1^0)+\frac{1}{2e^{\varepsilon_0}}\mathcal{R}(x_1^1)+(1-\frac{1}{e^{\varepsilon_0}})\text{LO}(x_i).$$

Let \( \mathcal{M}_S = \mathcal{S} \circ \mathcal{R}^n \) denote the shuffling of \( \mathcal{R} \).
To compute the privacy amplification provided by the shuffle model, we need to compute $D_{e^\varepsilon}(\mathcal{M}_S(X^0)\parallel \mathcal{M}_S(X^1))$ for a given $\varepsilon$. The exact computation is computationally complex, so the researchers seek an upper bound for it. A key property is that hockey-stick divergence satisfies the data processing inequality.
\begin{property}[Data Processing Inequality]
For all distributions $P$ and $Q$ defined on a set $S$ and (possibly randomized) functions $f: S \to S'$, $$D_\alpha\big(f(P)||f(Q)\big) \le D_\alpha(P||Q).$$
\end{property}

If we can find two probability distributions \( P_0 \) and \( P_1 \) along with a post-processing function \( f \) such that \( f(P_0) = \mathcal{M}_S(X^0) \) and \( f(P_1) = \mathcal{M}_S(X^1) \), then it follows that \( D_{e^\varepsilon}(P_0, P_1) \) is an upper bound for \( D_{e^\varepsilon}(\mathcal{M}_S(X^0), \mathcal{M}_S(X^1)) \). 
We refer to \( (P_0, P_1) \) as a \textit{reduction pair}. Different analysis techniques construct different reduction pairs. We first present an intuitive construction of the reduction pair within the standard clone framework, followed by the formal construction.

\begin{definition}[Standard Clone Reduction Pair (Intuitive) \cite{Feldman2021}]
Define random variables \( A_0, A_1 \) and \( A_2 \) as follows:
\[
A_0 =
\begin{cases}
0 & \text{w.p. } 1 \\
1 & \text{w.p. } 0 \\
2 & \text{w.p. } 0
\end{cases}
,\ 
A_1 =
\begin{cases}
0 & \text{w.p. } 0 \\
1 & \text{w.p. } 1 \\
2 & \text{w.p. } 0
\end{cases}
, \text{ and }
A_2 =
\begin{cases}
0 & \text{w.p. } \frac{1}{2e^{\varepsilon_0}} \\
1 & \text{w.p. } \frac{1}{2e^{\varepsilon_0}} \\
2 & \text{w.p. } 1 - \frac{1}{e^{\varepsilon_0}}
\end{cases}
\]
To obtain a sample from $P_0$ (or $P_1$), sample one copy from \( A_0 \) (or $A_1$) and \( n - 1 \) copies of \( A_2 \), the output \( (n_0, n_1) \) where \( n_0 \) is the total number of 0s and \( n_1 \) is the total number of 1s. Equivalently,
$$C \sim \text{Bin}(n-1,\frac{1}{e^{\varepsilon_0}}),\ A \sim \text{Bin}(C,\frac{1}{2}).$$
$$P_0=(A+1,C-A),\ P_1=(A,C-A+1).$$
The corresponding post-processing function $f_1$ is shown in the Algorithm \ref{alg:post-processing1}.    
\end{definition}

\begin{algorithm}[t]
\caption{Post-processing function of standard clone \cite{Feldman2021}, $f_1$ }
\label{alg:post-processing1}
\begin{algorithmic}[0]
\Meta $x_1^0, x_1^1, x_2, \dots, x_n$
\Require  $y \in \{0,1,2\}^n$
\State $J \gets \emptyset$
\For{$i = 1, \dots, n$}
    \If{$y_i = 2$}
        \State Let $j$ be a randomly and uniformly chosen element of $[2 : n] \setminus J$
        \State $J \gets J \cup \{j\}$
    \EndIf
    \State Sample $z_i$ from
    \Statex \hspace{1.5em} $\begin{cases}
        \mathcal{R}(x_1^0) & \text{if } y_i = 0; \\
        \mathcal{R}(x_1^1) & \text{if } y_i = 1; \\
        \text{LO}(x_{j}) & \text{if } y_i = 2.
    \end{cases}$
\EndFor
\State \Return $z_1, \dots, z_n$
\end{algorithmic}
\end{algorithm}

An additional observation is that if $\mathcal{R}$ is $\varepsilon_0$-DP, then $\mathcal{R}(x_1^0)$ and $\mathcal{R}(x_1^1)$ are similar, hence privacy is further amplified \cite{Feldman2021}. The similarity is characterized by the following lemma:
\begin{lemma}[\cite{KOV15}]\label{lemma1}
Let \( \mathcal{R}: \mathbb{X} \to \mathbb{Y} \) be an \( \varepsilon_0 \)-DP local randomizer and \( x_0, x_1 \in \mathbb{X} \). Then there exists two probability distributions \( \mathcal{Q}_0,\mathcal{Q}_1 \) such that
\[
\mathcal{R}(x_0) = \frac{e^{\varepsilon_0}}{e^{\varepsilon_0}+1} \mathcal{Q}_0 + \frac{1}{e^{\varepsilon_0}+1} \mathcal{Q}_1
\]
and
\[
\mathcal{R}(x_1) = \frac{1}{e^{\varepsilon_0}+1} \mathcal{Q}_0 + \frac{e^{\varepsilon_0}}{e^{\varepsilon_0}+1} \mathcal{Q}_1.
\]
\end{lemma}
With the help of Lemma \ref{lemma1}, \cite{Feldman2021} gives the following decomposition for generic local randomizers:
\begin{align*}
\mathcal{R}(x_1^0)&=\frac{e^{\varepsilon_0}}{e^{\varepsilon_0}+1}\mathcal{Q}_1^0+\frac{1}{e^{\varepsilon_0}+1}\mathcal{Q}_1^1,\\
\mathcal{R}(x_1^1)&=\frac{1}{e^{\varepsilon_0}+1}\mathcal{Q}_1^0+\frac{e^{\varepsilon_0}}{e^{\varepsilon_0}+1}\mathcal{Q}_1^1,\\
\forall i \in [2,n]:\mathcal{R}(x_i)&=\frac{1}{2e^{\varepsilon_0}}\mathcal{Q}_1^0+\frac{1}{2e^{\varepsilon_0}}\mathcal{Q}_1^1+(1-\frac{1}{e^{\varepsilon_0}})\text{LO}(x_i).
\end{align*}
This decomposition leads to the formal reduction of the standard clone:
\begin{theorem}[Standard Clone Reduction \cite{Feldman2021}]
Let \( \mathcal{R} : \mathbb{X} \to \mathbb{Y} \) be a $\varepsilon_0$-DP local randomizer and let \( \mathcal{M}_S = \mathcal{S} \circ \mathcal{R}^n \) be the shuffling of \( \mathcal{R} \). For \( \varepsilon \geq 0 \) and inputs \( X^0 \simeq X^1 \) with \( x_1^0 \neq x_1^1 \), we have
$$D_{e^\varepsilon}\big(\mathcal{M}_S(X^0),\mathcal{M}_S(X^1)\big) \le D_{e^\varepsilon}\big(P_0^{\mathsf{C}},P_1^{\mathsf{C}} \big)$$
where $P_0^{\mathsf{C}},\ P_1^{\mathsf{C}}$ are defined as below (with ``$\mathsf{C}$'' denoting ``standard Clone'') :
$$C \sim \text{Bin}(n-1,\frac{1}{e^{\varepsilon_0}}), \quad A \sim \text{Bin}(C,\frac{1}{2}),\quad \text{and } \quad\Delta \sim \text{Bern}(\frac{e^{\varepsilon_0}}{e^{\varepsilon_0}+1}).$$ $$P_0^\mathsf{C}=(A+\Delta,C-A+1-\Delta), \quad P_1^\mathsf{C}=(A+1-\Delta,C-A+\Delta).$$
Bern($p$) represents a Bernoulli random variable with bias $p$.
\end{theorem}
\begin{proof}
We can construct a post-processing function from \( (P_0^{\mathsf{C}}, P_1^{\mathsf{C}}) \) to \( \big( \mathcal{M}_S(X^0), \mathcal{M}_S(X^1) \big) \), which is similar to Algorithm \ref{alg:post-processing1}. The only difference is that \( \mathcal{R}(x_1^0) \) and \( \mathcal{R}(x_1^1) \) are replaced by \( \mathcal{Q}_1^0 \) and \( \mathcal{Q}_1^1 \), respectively. 

Since the standard clone is a special case of the general clone in Section \ref{sec_general_clone}, the correctness of this post-processing function is implied by the proof of Theorem~\ref{theorem_generalclone}.
\end{proof}

\subsection{Privacy blanket framework}\label{sec_blanket}
The decomposition of the standard clone is constructed by treating \( \mathcal{R}(x_1^0) \) and \( \mathcal{R}(x_1^1) \) as reference points and projecting each \( \mathcal{R}(x_i) \) onto the basis defined by these two distributions~\cite{Feldman2021}.
In contrast, the decomposition provided by the privacy blanket framework takes a different approach: it first identifies the ``common part'' shared across all \( \mathcal{R}(x) \) distributions~\cite{Balle2019}:
\[
\omega(y) = \inf_{x \in \mathbb{X}} \mathcal{R}(x)(y) / \gamma,
\]
where \( \mathcal{R}(x)(y) \) is the probability density of \( \mathcal{R}(x) \) at point \( y \), and $\gamma$ is a normalization factor:
\[
\gamma = \int \inf_{x} \mathcal{R}(x)(y) \, \mathrm{d}y.
\]

Here, \( \omega \) and \( \gamma \) are referred to as the privacy blanket distribution and the total variation similarity of the local randomizer \( \mathcal{R} \)~\cite{Balle2019}. Each \( \mathcal{R}(x) \) can then be decomposed as:
\[
\mathcal{R}(x) = \gamma \omega+(1 - \gamma) \text{LO}(x),
\]
where \( \text{LO}(x) \) represents the ``left-over" distribution. 

In other words, the execution of each \( \mathcal{R}(x_i) \) can be viewed as first sampling a random variable \( b_i \sim \text{Bern}(\gamma) \). If \( b_i = 1 \), a sample is drawn from \( \omega \) and returned; otherwise, a sample is drawn from \( \text{LO}(x_i) \).

The original proof is formulated using the terminology of the ``View'' of the server. We restate the privacy blanket technique using the following notation: probability distributions \( P_0^\mathsf{B} \) and \( P_1^\mathsf{B} \) (with ``$\mathsf{B}$'' denoting ``Blanket''), along with a post-processing function \( f^{\mathsf{B}} \).

\begin{definition}[Privacy Blanket Reduction Pair \cite{Balle2019} (Restated)]\label{definition_blanket}
Let $\boldsymbol{x}_{-1}=(x_2,x_3,\dots,x_n)$ with the inputs from the last $n-1$ users, $\boldsymbol{y}^a=(y_1^a,y_2,\dots,y_n)$ where $y_i\sim R(x_i)$ is the output of the i-th user, $a$ indicates that the input of the first user is $x_1^a$, $a\in \{0,1\}$. Let $\boldsymbol{b}=(b_2,b_3,\dots,b_n)$ be binary values indicating which users sample from the privacy blanket distribution. A multiset $Y_{\boldsymbol{b}}^a=\mathcal{S}(\{y_1^a\} \cup\{y_i|b_i=1\})$.

Observe that the distribution of $Y_{\boldsymbol{b}}^a$ depends only on $|\boldsymbol{b}|$ rather than $\boldsymbol{b}$, where $|\boldsymbol{b}|$ represents the number of 1 in $\boldsymbol{b}$. We can rewrite it as $Y_{|\boldsymbol{b}|}^a=\mathcal{S}(\{y_1^a\}\cup\{y_i|y_i \sim \omega, i=1,2,\dots,|\boldsymbol{b}|\})$.
Then $P_0^B$ and $P_1^B$ are defined below:
$$P_0^\mathsf{B}=(|\boldsymbol{b}|,Y^0_{|\boldsymbol{b}|}),$$
$$P_1^\mathsf{B}=(|\boldsymbol{b}|,Y^1_{|\boldsymbol{b}|}),$$
where $|\boldsymbol{b}| \sim \text{Bin}(n-1,\gamma)$.
\end{definition}

\begin{theorem}[Privacy Blanket Reduction \cite{Balle2019} (Restated)]
Let \( \mathcal{R} : \mathbb{X} \to \mathbb{Y} \) be a $\varepsilon_0$-DP local randomizer and let \( \mathcal{M}_S = \mathcal{S} \circ \mathcal{R}^n \) be the shuffling of \( \mathcal{R} \). For \( \varepsilon \geq 0 \) and inputs \( X^0 \simeq X^1 \) with \( x_1^0 \neq x_1^1 \), we have
$$D_{e^\varepsilon}\big(\mathcal{M}_S(X^0),\mathcal{M}_S(X^1)\big) \le D_{e^\varepsilon}\big(P_0^{\mathsf{B}},P_1^{\mathsf{B}} \big).$$
\end{theorem}
\begin{proof}
The corresponding post-processing function $f^{\mathsf{B}}$ is shown in Algorithm \ref{alg:post-processing2}. The core idea of the post-processing function \( f^{\mathsf{B}} \) is that, given \( Y_{|\boldsymbol{b}|}^a \), it suffices to sample from the left-over distributions of a randomly selected subset of \( n - 1 - |\boldsymbol{b}| \) users and mix the results accordingly. 

In Section~\ref{sec_general_clone}, we will show that the blanket is, in fact, also a special case of the general clone. The correctness of the corresponding post-processing function $f^{\mathsf{B}}$ follows from the proof of Theorem~\ref{theorem_generalclone}.
\end{proof}

\begin{algorithm}[t]
\caption{Post-processing function of privacy blanket, $f^{\mathsf{B}}$ }
\label{alg:post-processing2}
\begin{algorithmic}[0]
\Meta $x_2, \dots, x_n$
\Require $|\boldsymbol{b}| \in \{0,1,\dots,n-1 \},Y^a_{|\boldsymbol{b}|}$
\State $J \gets \emptyset$
\State $S \gets \emptyset$
\For{$loop = 1, \dots, n-1-|\boldsymbol{b}|$}
    \State Let $j$ be a randomly and uniformly chosen element of $[2 : n] \setminus J$
    \State $s \gets_{\text{LO}(x_j)} \mathbb{Y}$ \Comment{Sample from $\text{LO}(x_j)$}
    \State $J \gets J \cup \{j\}$
    \State $S \gets S \cup \{s\}$
\EndFor
\State \Return $Y^a_{|\boldsymbol{b}|} \cup S$
\end{algorithmic}
\end{algorithm}

\subsection{Vision and failure of stronger clone}\label{sec:stronger_clone}
The stronger clone is expected to improve the probability of producing a ``clone'' from \( \frac{1}{2e^{\varepsilon_0}} \) to \( \frac{1}{e^{\varepsilon_0} + 1} \). For \( \varepsilon_0 > 1 \), this results in approximately a factor of 2 improvement in the expected number of ``clones'' \cite{feldman2023soda}. 
This improvement is anticipated to be achieved through a more refined analysis that, instead of cloning the entire output distributions on differing elements, clones only the portions of those distributions where they actually differ.

Specifically, it leverages a lemma from \cite{Ye} to establish the existence of the following decomposition:
\begin{theorem}[Corollary 3.4 in \cite{feldman2023soda}]\label{theorem_sc_decomp}
Given any $\varepsilon_0$-DP local randomizer $\mathcal{R}: \mathbb{X} \to \mathbb{Y}$, and any $n + 1$ inputs $x_1^0, x_1^1, x_2, \dots, x_n \in \mathbb{X}$, if $\mathbb{Y}$ is finite then there exists $p \in [0, 1/(e^{\varepsilon_0} + 1)]$ and distributions $\mathcal{Q}_1^0, \mathcal{Q}_1^1, \mathcal{Q}_1, \mathcal{Q}_2, \dots, \mathcal{Q}_n$ such that
\begin{align*}
\mathcal{R}(x_1^0) &= e^{\varepsilon_0} p \mathcal{Q}_1^0 + p \mathcal{Q}_1^1 + (1 - p - e^{\varepsilon_0} p) \mathcal{Q}_1,\\
\mathcal{R}(x_1^1) &= p \mathcal{Q}_1^0 + e^{\varepsilon_0} p \mathcal{Q}_1^1 + (1 - p - e^{\varepsilon_0} p) \mathcal{Q}_1,\\
\forall i \in [2, n], \, \mathcal{R}(x_i) &= p \mathcal{Q}_1^0 + p \mathcal{Q}_1^1 + (1 - 2p) \mathcal{Q}_i.
\end{align*} 
\end{theorem} 

Such a decomposition is guaranteed to exist for any local randomizer. However, an error occurred in the construction of the reduction pair \( P_0 \) and \( P_1 \) based on this decomposition.
Similar to the standard clone in Section \ref{sec_clone}, they define the following distribution \( P_0(\varepsilon_0, p) \) and \( P_1(\varepsilon_0, p) \): For any \( p \in [0, 1/(e^{\varepsilon_0} + 1)] \),  let
$$
C \sim \text{Bin}(n - 1, 2p), \ A \sim \text{Bin}(C, 1/2) \ 
$$
and $$\Delta_1 \sim \text{Bern}(e^{\varepsilon_0} p), \quad \Delta_2 \sim \text{Bin}(1 - \Delta_1, p/(1 - e^{\varepsilon_0} p)).$$  Let
\[
P_0(\varepsilon_0, p) = (A + \Delta_1, C - A + \Delta_2) \quad \text{and} \quad P_1(\varepsilon_0, p) = (A + \Delta_2, C - A + \Delta_1).
\]
They intended to prove that
\[
D_{e^{\varepsilon}}\big(\mathcal{M}_S(X^0), \mathcal{M}_S(X^1)\big) \leq D_{e^\varepsilon}\big(P_0(\varepsilon_0, p), P_1(\varepsilon_0, p)\big),
\]
which serves as an upper bound for a specific local randomizer (different $\varepsilon_0$-DP randomizers may have different values of $p$).
Leveraging Lemma~\ref{lemma2}, they would then conclude the general upper bound for any \( \varepsilon_0 \)-DP local randomizer:
\begin{align}
D_{e^{\varepsilon}}\big(\mathcal{M}_S(X^0), \mathcal{M}_S(X^1)\big) \leq D_{e^\varepsilon}\big(P_0(\varepsilon_0, \frac{1}{e^{\varepsilon_0}+1}), P_1(\varepsilon_0, \frac{1}{e^{\varepsilon_0}+1})\big). \label{eq1}
\end{align}

\begin{lemma}[Lemma 5.1. in \cite{feldman2023soda}] \label{lemma2}
For any \( p, p' \in [0, 1] \) and \( \varepsilon > 0 \), if \( p < p' \), then
\[
D_{e^\varepsilon}\big(P_0(\varepsilon_0, p) \,\|\, P_1(\varepsilon_0, p)\big)
\le
D_{e^\varepsilon}\big(P_0(\varepsilon_0, p') \,\|\, P_1(\varepsilon_0, p')\big)
\]
\end{lemma}

Unfortunately, they encountered difficulty in constructing a post-processing function \( f \) for this construction of \( P_0 \) and \( P_1 \). While they provided a function in the original paper, it was proven to be incorrect in the corrected revision \cite{feldman2023arxiv}. The issue arises from the fact that the ``leftover'' distribution of \( x_1 \) (i.e., \( \mathcal{Q}_1 \)) is mixed with the ``leftover'' distribution of \( x_i \) (i.e., \( \mathcal{Q}_i \)) in the above construction.
In this case, the function \( f \) does not know which distribution to sample from. This technical problem is fundamental and remains unsolved.

Although the corrected version was published on arXiv in October 2023, this error has been propagated in subsequent works~\cite{tv_shuffle,ChenCG24,wang1,wang2,koskela2023numerical}.
For instance, the variation-ratio framework made significant efforts to design an algorithm to find the parameter \( p \) for various specific local randomizers in the decomposition of Theorem \ref{theorem_sc_decomp}. However, their work relies on the incorrect post-processing function \( f \) presented in~\cite{feldman2023soda}, which renders their results invalid.

Due to this fundamental difficulty, it is required that \( x_1 \) has no ``leftover'' distribution. In other words, each component of \( \mathcal{R}(x_1) \) must be distinguishable from the ``leftover'' distribution of \( x_i \) for \( i \geq 2 \). In the above example, this necessitates a four-point-based construction for \( P_0(\varepsilon_0, p, q) \) and \( P_1(\varepsilon_0, p, q) \)~\cite{feldman2023arxiv}:
\begin{align*}
\mathcal{R}(x_1^0) &= e^{\varepsilon_0} p \mathcal{Q}_1^0 + p \mathcal{Q}_1^1 + (1 - p - e^{\varepsilon_0} p) \mathcal{Q}_1,\\
\mathcal{R}(x_1^1) &= p \mathcal{Q}_1^0 + e^{\varepsilon_0} p \mathcal{Q}_1^1 + (1 - p - e^{\varepsilon_0} p) \mathcal{Q}_1,\\
\forall i \in [2, n], \quad \mathcal{R}(x_i) &= p \mathcal{Q}_1^0 + p \mathcal{Q}_1^1 + q \mathcal{Q}_1 + (1 - 2p - q) \mathcal{Q}_i.
\end{align*}

This new decomposition proposed in the corrected version introduces additional challenges. 
First, for specific randomizers, computing a tight value of \( q \) is nontrivial. 
Second, the monotonicity of 
\(
D_{e^\varepsilon}\big(P_0(\varepsilon_0, p, q), P_1(\varepsilon_0, p, q)\big)
\)
with respect to \( p \) is not known. Consequently, we are unable to derive the desired conclusion—namely, a general upper bound applicable to any \( \varepsilon_0 \)-DP local randomizer, as stated in (\ref{eq1}). 
For the same reason, it remains unclear whether this decomposition necessarily yields tighter bounds than the standard clone decomposition. Moreover, the new bound \( D_{e^\varepsilon}\big(P_0(\varepsilon_0, p, q), P_1(\varepsilon_0, p, q)\big) \) for a specific local randomizer lacks an efficient algorithm to compute.

\section{General Clone and the Optimal Bounds}\label{sec_general_clone}

In this section, we formalize the \textit{general clone paradigm}, which unifies and generalizes all decomposition methods for analyzing privacy amplification in the shuffle model. We then identify the optimal bounds achievable within this paradigm. The main results are summarized as follows:

\begin{itemize}
    \item \textbf{Upper bound limitation:} The general clone paradigm does not provide tighter bounds than the privacy blanket. In other words, its analytical capability is not inherently stronger than that of the privacy blanket framework.
    
    \item \textbf{Equivalence for specific randomizers:} For any \emph{specific} local randomizer, the optimal decomposition under the general clone paradigm is \emph{equivalent} to the decomposition used in the privacy blanket framework.
\end{itemize}

The hierarchy of the bounds provided by the decomposition-based methods is shown in Fig. \ref{hierarchy}.
\begin{figure}
    \centering
    \includegraphics[width=0.6\linewidth]{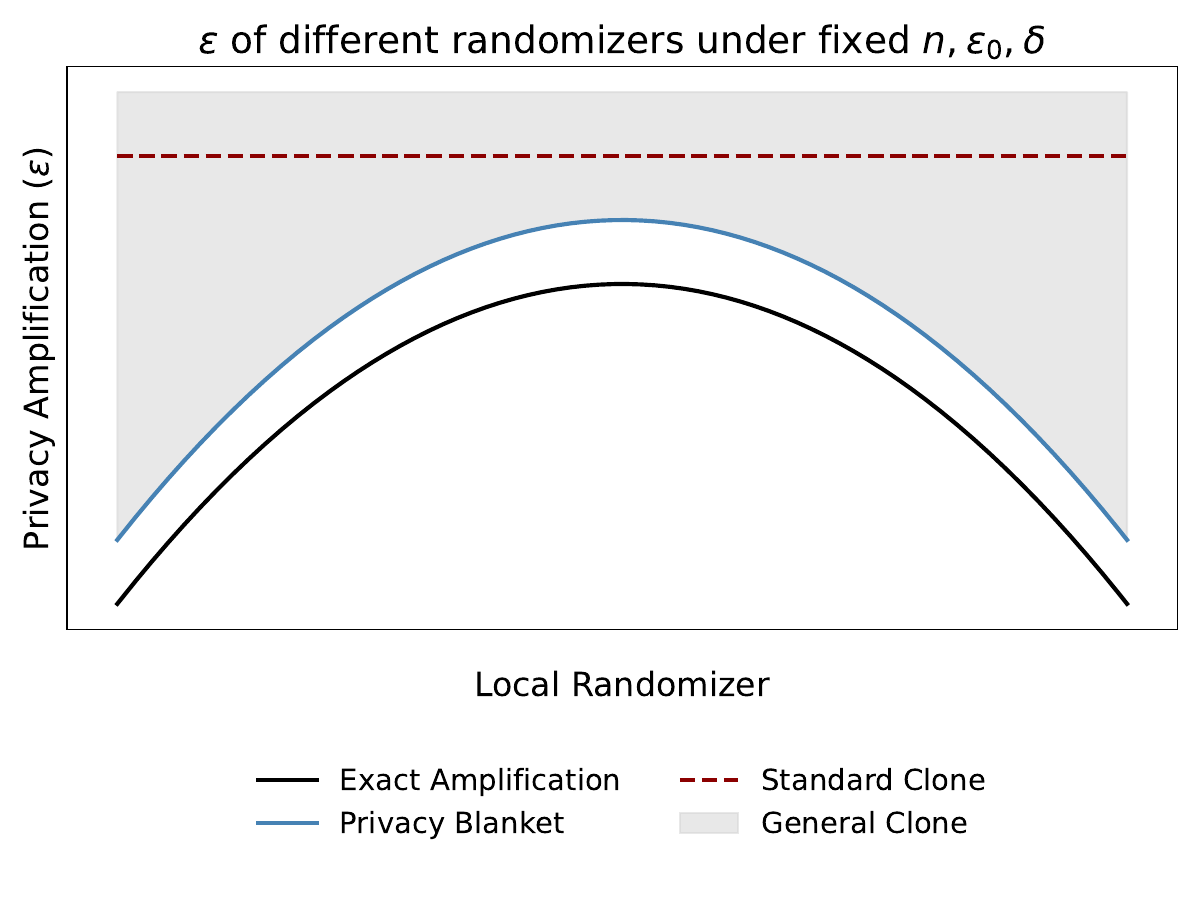}
    \caption{Hierarchy among decomposition-based methods. The Standard clone is horizontal (generic bound), while the Blanket curve varies with the randomizer (specific bound).}
    \label{hierarchy}
\end{figure}

Since the primary optimal decomposition derived from the blanket framework is generally highly redundant, we propose a general method for computing a simplified optimal decomposition. Furthermore, we observe that the optimal decomposition satisfies desirable compositional properties under both joint composition and parallel composition. Specifically, the optimal decomposition of a joint composition corresponds to the Cartesian product of the optimal decompositions of the sub-local randomizers, while that of a parallel composition corresponds to their weighted combination.

\subsection{Definition of general clone}\label{sec:definition_gc}
\begin{definition}[Decomposition in the General Clone]\label{definition:gc}
Let \( \mathcal{R}: \mathbb{X} \to \mathbb{Y} \) be a local randomizer. The general clone paradigm defines a decomposition of \( \mathcal{R} \) as follows:
\begin{align}
\mathcal{R}(x_1^0) &= \sum_{j=1}^k a_j Q_1^j, \nonumber \\
\mathcal{R}(x_1^1) &= \sum_{j=1}^k b_j Q_1^j, \nonumber \\
\forall i \in [2, n],\quad \mathcal{R}(x_i) &= \sum_{j=1}^k c_j Q_1^j + \beta Q_i, \label{f1}
\end{align}
where \( Q_1^j \) for \( j = 1, \dots, k \) and \( Q_i \) for \( i = 2, \dots, n \) are probability distributions over \( \mathbb{Y} \), and \( a_j, b_j, c_j, \beta \) are non-negative coefficients and satisfy:
\[
\sum_{j=1}^k a_j = 1,\quad \sum_{j=1}^k b_j = 1,\quad \beta+\sum_{j=1}^k c_j = 1.
\]
\end{definition}

\begin{remark}
Although the index \( i \) ranges over \( [2, n] \), each \( Q_i \) depends only on the corresponding input \( x_i \). Specifically, it is defined as the ``left-over'' distribution:
\[
\text{\emph{LO}}(x_i) := \frac{1}{\beta} \left( \mathcal{R}(x_i) - \sum_{j=1}^k c_j Q_1^j \right).
\]
For notational simplicity, we use \( Q_i \) to denote \( \text{\emph{LO}}(x_i) \) throughout the analysis.
\end{remark}

\begin{remark}
The general clone is defined with respect to three parameters: the local randomizer $\mathcal{R}$ and two inputs $x_1^0, x_1^1$. When $\mathcal{R}: \mathbb{X} \to \mathbb{Y}$ satisfies $\varepsilon_0$-differential privacy, a decomposition of the form given by the standard clone exists for any pair $x_1^0, x_1^1 \in \mathbb{X}$. However, in the general clone framework, the feasible decomposition may vary depending on the specific choice of $x_1^0$ and $x_1^1$.
\end{remark}

The general clone paradigm characterizes the general form of decompositions used in privacy amplification analysis. It directly subsumes the decomposition used in the standard clone framework. Although the decomposition defined by the privacy blanket appears structurally different, we show that it naturally corresponds to a valid decomposition under the general clone paradigm in Section \ref{sec:blanket_in_gc}.

When deriving the reduction pair from a general clone decomposition, an important constraint must be considered. Motivated by the failure of the stronger clone, the components \( Q_1^j \) (shared across users) should not be mixed with the left-over distributions \( Q_i \) (which are user-specific). This separation is essential to ensure the correctness and validity of the reduction.

\begin{definition}[Reduction Pair of the General Clone] \label{reduction_pair_gc}
Let \( A_0, A_1 \), and \( A_2 \) be random variables on \( \{1, 2, \dots, k+1\} \), defined by the probabilities:
\begin{align*}
&\Pr[A_0 = j] = \begin{cases}
a_j & \text{for } j \le k, \\
0 & \text{for } j = k+1.
\end{cases}   , \quad
\Pr[A_1 = j] = \begin{cases}
b_j & \text{for } j \le k, \\
0 & \text{for } j = k+1.
\end{cases}\\   
&\Pr[A_2 = j] = 
\begin{cases}
c_j & \text{for } j \le k, \\
\beta & \text{for } j = k+1.
\end{cases}   
\end{align*}
The reduction pair \( (P_0^{\mathsf{GC}}, P_1^{\mathsf{GC}}) \) (with ``$\mathsf{GC}$'' denoting ``General Clone'') is defined as the distributions of the histograms over \( \{1, \dots, k+1\} \) generated by sampling:
\begin{itemize}
    \item One sample from \( A_0 \) for \( P_0^{\mathsf{GC}} \) (or \( A_1 \) for \( P_1^{\mathsf{GC}} \));
    \item \( n-1 \) i.i.d. samples from \( A_2 \).
\end{itemize}
The output is the histogram vector \( (n_1, n_2, \dots, n_k, n_{k+1}) \) indicating the counts of each index.
\end{definition}

\begin{theorem}[General Clone Reduction]\label{theorem_generalclone}
Let \( \mathcal{R} : \mathbb{X} \to \mathbb{Y} \) be a $\varepsilon_0$-DP local randomizer and let \( \mathcal{M}_S = \mathcal{S} \circ \mathcal{R}^n \) be the shuffling of \( \mathcal{R} \). For \( \varepsilon \geq 0 \) and inputs \( X^0 \simeq X^1 \) with \( x_1^0 \neq x_1^1 \), we have
$$D_{e^\varepsilon}\big(\mathcal{M}_S(X^0),\mathcal{M}_S(X^1)\big) \le D_{e^\varepsilon}\big(P_0^{\mathsf{GC}},P_1^{\mathsf{GC}} \big).$$
\end{theorem}
\begin{proof}
The corresponding post-processing function $f^{\mathsf{GC}}$ is given in Algorithm~\ref{alg:post-processing_GC}. Its behavior is as follows.  
For each index $j \in \{1,2,\dots,k\}$, it draws $n_j$ samples from the distribution~$Q_1^j$.  
For index $k+1$, it uniformly samples an $n_{k+1}$-subset $I \subseteq \{2,3,\dots,n\}$ and, for each $i \in I$, draws a sample from the ``left-over'' distribution~$Q_i$.  
The algorithm finally outputs the multiset of all sampled values.

We show that \( f^{\mathsf{GC}}(P_0^{\mathsf{GC}}) \stackrel{d}{=} \mathcal{M}_S(X^0) \)\footnote{$X\stackrel{d}{=}Y$ means that $X$ and $Y$ have the same distribution.}; the argument for \( f^{\mathsf{GC}}(P_1^{\mathsf{GC}}) \stackrel{d}{=} \mathcal{M}_S(X^1) \) is symmetric.

Recall that \( P_0^{\mathsf{GC}} \) is the histogram of one sample drawn from \( A_0 \) and \( n-1 \) independent samples drawn from \( A_2 \), whereas \( \mathcal{M}_S(X^0) \) is obtained by shuffling the outputs of the local randomizer applied to \(x_1^0,x_2,\dots,x_n\).  
By the general clone decomposition, \( \mathcal{M}_S(X^0) \) admits the following equivalent generative description:  
user~1 draws $a_1 \sim A_0$ and then $y_1 \sim Q_1^{a_1}$; each other user $i \ge 2$ draws $a_i \sim A_2$ and then $y_i \sim Q_1^{a_i}$ if $a_i \in [k]$, or $y_i \sim Q_i$ otherwise.  
The shuffled output is \( \{y_i\}_{i=1}^n \).

Let $T^a = \{a_i\}$, $T^y = \{y_i\}$, and $T^{a,y} = \{(a_i,y_i)\}$ denote the latent, output, and joint multisets, respectively.  
Let $H^a = (n_1,n_2,\dots,n_{k+1})$ be the histogram of $T^a$, where $n_j = |\{i : a_i = j\}|$.  
Then the marginals of $H^a$ and $T^y$ are exactly $P_0^{\mathsf{GC}}$ and $\mathcal{M}_S(X^0)$.

Define the index set  
\(
I = \{ i : a_i = k+1 \}.
\)
Conditioned on $H^a$, symmetry implies that $I$ is uniformly distributed over all subsets of $\{2,3,\dots,n\}$ of size~$n_{k+1}$.  
Furthermore,
\begin{align*}
T^{a,y} \mid (H^a, I)
&= \cup_{j=1}^k\{(j, y_t) : \ y_t \sim Q_1^{j},t=1,2,\dots,n_j\}\\
  &\quad \cup \{(k+1, y_i) : i \in I,\ y_i \sim Q_i\}.\\
\text{Thus, }  &T^y\mid (H^a, I)= \Big( \bigcup_{j=1}^k V_j \Big) \cup V_{k+1},
\end{align*}
where $V_j$ consists of $n_j$ i.i.d.\ samples from~$Q_1^j$ for $1 \le j \le k$, and  
\(
V_{k+1} = \{ y_i : i \in I,\ y_i \sim Q_i \}.
\)
This conditional generative process matches exactly the sampling performed by $f^{\mathsf{GC}}$.  
Therefore, for every histogram $h$,  
\(
f^{\mathsf{GC}}(h) \stackrel{d}{=} T^y \mid (H^a = h),
\)
and hence
\(
f^{\mathsf{GC}}(P_0^{\mathsf{GC}}) \stackrel{d}{=} \mathcal{M}_S(X^0).
\)
\end{proof}

\begin{algorithm}[t]
\caption{Post-processing function of general clone, $f^{\mathsf{GC}}$ }
\label{alg:post-processing_GC}
\begin{algorithmic}[0]
\Require $(n_1,\dots,n_k,n_{k+1})\in \{0,1,\dots,n\}^{k+1}$
\State $O\gets\emptyset$
\For{$j = 1, \dots, k+1$}
    \If{$j = k+1$}
        \State Let $I$ be a randomly and uniformly chosen $n_{k+1}$-size subset of $\{2,3,\dots,n\}$
        \For{$i\in I$}
            \State $O\cup \{s\mid s\sim Q_i\}$ \Comment{Sample from the left-over distribution of $x_i$}
        \EndFor
    \Else
    \For{$loop=1,2,\dots,n_{j}$}
        \State $O\cup \{s\mid s\sim Q_1^j\}$ \Comment{Sample from $Q^j_1$}
    \EndFor
    \EndIf
\EndFor
\State \Return $O$
\end{algorithmic}
\end{algorithm}

\subsection{Blanket is ``in'' the general clone}\label{sec:blanket_in_gc}

\begin{algorithm}[t]
\caption{Bijection from privacy blanket to primary optimal decomposition in general clone, $f_2$ }
\label{alg:post-processing_B2OGC_}
\begin{algorithmic}[0]
\Meta $n$, $\mathbb{Y}=\{y_1,y_2, \dots, y_{|\mathbb{Y}|}\}$
\Require $(|\boldsymbol{b}|,Y_{|\boldsymbol{b}|})$
\State $H \gets \{\}$ \Comment{$H$ is a histogram}
\For{$y \in Y_{|\boldsymbol{b}|}$}
\State $j \gets \text{index}(y) \text{ such that } y_j = y$
\State $H[j]\gets H[j]+1$ \Comment{Map $y=y_j$ in Blanket to index $j$ in GC}
\EndFor
\State $H[|\mathbb{Y}|+1] \gets n-1-|\boldsymbol{b}|$ \Comment{Map $y=\bot$ in Blanket to index $|\mathbb{Y}|+1$ in GC}
\State \Return $H$
\end{algorithmic}
\end{algorithm}

For every local randomizer, the general clone paradigm always admits a decomposition that is equivalent to the decomposition in the privacy blanket framework, where the components are single-point distributions \( \mathbbm{1}_{y_j} \), with \( y_j \in \mathbb{Y} \).

\begin{theorem}\label{primary_optimal_decomposition}
For any local randomizer $\mathcal{R}: \mathbb{X}\to \mathbb{Y}$ and inputs $x_1^0, x_1^1 \in \mathbb{X}$, the following decomposition in the general clone paradigm is equivalent to the privacy blanket framework:
\begin{align}
\mathcal{R}(x_1^0) &= \sum_{j=1}^{|\mathbb{Y}|} a_j \mathbbm{1}_{y_j}, \nonumber \\ 
\mathcal{R}(x_1^1) &= \sum_{j=1}^{|\mathbb{Y}|} b_j \mathbbm{1}_{y_j}, \nonumber \\ 
\forall i \in [2, n], \quad \mathcal{R}(x_i) &= \sum_{j=1}^{|\mathbb{Y}|} c_j \mathbbm{1}_{y_j} + \beta Q_i, \label{c}
\end{align}
where $a_j = \mathcal{R}(x_1^0)(y_j)$, $b_j = \mathcal{R}(x_1^1)(y_j)$, and $c_j = \inf_{x \in \mathbb{X}} \mathcal{R}(x)(y_j)$.

We refer to this as the \textbf{primary optimal decomposition} with respect to $\mathcal{R}$ and $(x_1^0, x_1^1)$.
\end{theorem}

\begin{proof}
We construct a post-processing function $f_2$ (shown in Algorithm \ref{alg:post-processing_B2OGC_}).
The function $f_2$ initializes an empty set $H$, then iterates over its input. Upon encountering an element $y_j$, it adds the index $j$ to the set $H$. Finally, it appends $n - 1 - |\boldsymbol{b}|$ copies of $|\mathbb{Y}| + 1$ to $H$. 

We now prove that $f_2(P^\mathsf{B}_0) = P^{\mathsf{GC}}_0$.  
The input of User 1 is sampled from $\mathcal{R}(x_1^0)$ and always appears in the output of $P^\mathsf{B}_0$. By the definition of $f_2$, we have:
$$f_2(\mathcal{R}(x_1^0)) \sim A_0,$$
where $A_0$ is defined in Definition~\ref{reduction_pair_gc}. 

For the remaining users, each appears in the output of $P^\mathsf{B}_0$ with probability $\gamma$, and is omitted with probability $1 - \gamma$. In the latter case, $f_2$ outputs the special index $|\mathbb{Y}| + 1$. In the former case, the distribution of $f_2(\omega)$ matches the \textbf{conditional distribution} of $A_2$ given that $A_2 \ne |\mathbb{Y}| + 1$:
$$\Pr[f_2(\omega)=y_j]=\inf_{x \in \mathbb{X}} \mathcal{R}(x)(y_j)/\gamma=c_j/\sum_{m=1}^{|\mathbb{Y}|}c_m.$$
Together, this shows that applying $P^B_0$ followed by $f_2$ for a remaining user is equivalent to directly sampling a random variable from $A_2$. Therefore, the output distribution satisfies:
$$f_2(P^\mathsf{B}_0) = P^{\mathsf{GC}}_0.$$
The proof for $f_2(P^\mathsf{B}_1) = P^{\mathsf{GC}}_1$ follows by a similar argument.

As $f_2$ is a bijection, a post-processing function $f_2^{-1}$ exists such that $f_2^{-1}(P^{\mathsf{GC}}_0) = P^{\mathsf{B}}_0, f_2^{-1}(P^{\mathsf{GC}}_1) = P^{\mathsf{B}}_1$, thereby proving the equivalence between the two constructions.
\end{proof}

\subsection{General clone is not stronger than blanket}
Given a specific randomizer, we compare any decomposition within the general clone paradigm against the decomposition provided by the privacy blanket framework.

\begin{theorem}\label{theorem_clonevsblanket}
For every local randomizer, there is a post-processing function from $(P_0^{\mathsf{GC}},P_1^{\mathsf{GC}})$ to the $(P_0^\mathsf{B},P_1^\mathsf{B})$, where $(P_0^{\mathsf{GC}},P_1^{\mathsf{GC}})$ is the reduction pair given by the general clone paradigm equipped with any decomposition, and $(P_0^\mathsf{B},P_1^\mathsf{B})$ is the reduction pair given by the privacy blanket framework. Therefore,
$$D_{e^\varepsilon}(\mathcal{M}_S(X_0),\mathcal{M}_S(X_1))\le D_{e^\varepsilon}(P_0^\mathsf{B},P_1^\mathsf{B}) \le D_{e^\varepsilon}(P_0^{\mathsf{GC}},P_1^{\mathsf{GC}}).$$
\end{theorem}


\begin{proof}
The core idea is that the privacy blanket characterizes the maximal common component shared by all $\mathcal{R}(x_i)$, and no decomposition under the general clone paradigm can yield a larger common part.
Recall that the privacy blanket is defined via:
$$\omega(y) = \frac{1}{\gamma} \inf_{x \in \mathbb{X}} \mathcal{R}(x)(y),$$
where \( \gamma \) is a normalization constant.
An important observation is:
\begin{align}
\forall y \in \mathbb{Y}: \sum_{j=1}^k c_j Q_1^j(y) \le \gamma \omega(y). \label{f12}
\end{align}
This holds because the decomposition condition (\ref{f1}) in Definition \ref{definition:gc} must be satisfied for all inputs $x_2, x_3, \dots, x_n \in \mathbb{X}$, i.e.,
$$\forall x \in \mathbb{X}: \sum_{j=1}^k c_j Q_1^j(y) \le \mathcal{R}(x)(y).$$
Taking the infimum over $x \in \mathbb{X}$ on both sides of the above inequality yields (\ref{f12}).
Integrating both sides of (\ref{f12}) over $y$ yields $\sum_{j=1}^k c_j \le \gamma$.

Define a distribution $Q^{\text{com}}$ as the common part of all $Q_i$:
$$Q^{\text{com}}(y) = \frac{\gamma \omega(y) - \sum_{j=1}^k c_j Q_1^j(y)}{\mu},$$
where $\mu = \gamma - \sum_{j=1}^k c_j$ is a normalization factor. When $\mu = 0$, the distribution $Q^{\mathrm{com}}$ can be defined arbitrarily.
Then, each $Q_i$ can be decomposed by $$Q_i= \frac{\mu}{\beta} Q^{\text{com}}+(1-\frac{\mu}{\beta})Q_i'.$$
Here, \( Q_i' \) is a probability distribution, whose existence is established as follows:
From equation (\ref{f1}) and the definition of $\omega$, we have:
$$\forall i\in [2,n], \forall y\in \mathbb{Y}:\gamma \omega(y) \le \mathcal{R}(x_i)(y) = \sum_{j=1}^k c_j Q_1^j(y) + \beta Q_i(y).$$
This implies:
$$\forall i\in [2,n],\forall y\in \mathbb{Y}:Q_i(y) \ge \frac{\gamma \omega(y) -\sum_{j=1}^k c_j Q_1^j(y)}{\beta}=\frac{\mu}{\beta} Q^{\text{com}}(y).$$

Let \( f_3 \) be the post-processing function defined in Algorithm~\ref{alg:post-processing_GC2B}.
Its behavior is as follows: it processes each input sequentially. When the input is \( j \in [1, k] \), it outputs a sample from \( Q_1^j \); when \( j = k+1 \), it samples from \( Q^{\text{com}} \) with a certain probability. 

We now argue that \( f_3(P_0^{\mathsf{GC}}) = P_0^\mathsf{B} \). By definition, \( P_0^\mathsf{B} \) outputs one sample from \( \mathcal{R}(x_1^0) \), and \( |\boldsymbol{b}| \sim \text{Bin}(n-1, \gamma) \) samples from \( \omega \).
From Definition~\ref{reduction_pair_gc}, \( P_0^{\mathsf{GC}} \) draws one sample from \( A_0 \), and \( n-1 \) samples from \( A_2 \). Since \( \text{supp}(A_0) \subseteq \{1,\dots,k\} \), its output is always retained. The transformed output has distribution:
\begin{align*}
\Pr[f_3(A_0) = y] &= \sum_{j=1}^k \Pr[A_0 = j] \cdot \Pr[f_3(j) = y] \\
&= \sum_{j=1}^k a_j Q_1^j(y) = \mathcal{R}(x_1^0)(y).
\end{align*}

Now consider \( A_2 \). The event \( f_3(A_2) = \emptyset \) occurs with probability:
\begin{align*}
\Pr[f_3(A_2) = \emptyset] &= \Pr[A_2 = k+1] \cdot \Pr[f_3(k+1) = \emptyset] \\
&= \beta \cdot \left(1 - \frac{\mu}{\beta} \right) = \beta + \sum_{j=1}^k c_j - \gamma = 1 - \gamma.
\end{align*}

Moreover, the conditional distribution of \( f_3(A_2) \), given that it is non-empty, is exactly \( \omega \):
\begin{align*}
&\Pr[f_3(A_2) =y|f_3(A_2) \neq \emptyset]\\
&=\sum_{j=1}^{k+1} \Pr[A_2 =j|f_3(A_2) \neq \emptyset]\cdot \Pr[f_3(j)=y|f_3(A_2) \neq \emptyset \wedge A_2=j]\\
&=\sum_{j=1}^k c_jQ_1^j(y)/\gamma+ \frac{\gamma-\sum_{j=1}^k c_j}{\gamma} \cdot Q^{\text{com}}(y)\\
&=\sum_{j=1}^k c_jQ_1^j(y)/\gamma+  \frac{\gamma-\sum_{j=1}^k c_j}{\gamma}\cdot \frac{\gamma \omega(y) - \sum_{j=1}^k c_j Q_1^j(y)}{\gamma-\sum_{j=1}^k c_j}\\
&=\omega(y).
\end{align*}
Hence, \( f_3(P_0^{\mathsf{GC}}) = P_0^\mathsf{B} \), and the same argument applies for \( f_3(P_1^{\mathsf{GC}}) = P_1^\mathsf{B} \). This completes the proof.
\end{proof}

\begin{algorithm}[t]
\caption{Post-processing function from general clone to privacy blanket framework, $f_3$ }
\label{alg:post-processing_GC2B}
\begin{algorithmic}[0]
\Require $(n_1,\dots,n_k,n_{k+1})\in \{0,1,\dots,n\}^{k+1}$
\State $Y \gets \emptyset$
\For{$i = 1, \dots, k+1$}
\For{$count=1,2,\dots,n_i$}
    \If{$i = k+1$}
        \State $r \gets \text{Bern}(\frac{\gamma-\sum_{j=1}^k c_j}{\beta})$
        \If{$r=1$}
        \State $s \gets_{Q^{\text{com}}} \mathbb{Y}$ \Comment{Sample from $Q^{\text{com}}$}
            \State $Y \gets Y \cup \{s\}$
        \EndIf
    \Else
        \State $s \gets_{Q^i_1} \mathbb{Y}$ \Comment{Sample from $Q^i_1$}
            \State $Y \gets Y \cup \{s\}$
    \EndIf

\EndFor
\EndFor
\State \Return $(|Y|-1,Y)$
\end{algorithmic}
\end{algorithm}

As a result, although both \( D_{e^\varepsilon}(P_0^{\mathsf{GC}}|| P_1^{\mathsf{GC}}) \) and \( D_{e^\varepsilon}(P_0^{\mathsf{B}}|| P_1^{\mathsf{B}}) \) serve as upper bounds on the privacy amplification in the shuffle model, the bound provided by the privacy blanket is always at least as tight as that of the general clone paradigm under any decomposition.

\subsection{Simplified optimal decomposition}\label{sec_simplified_od}

We refer to the optimal decomposition of a local randomizer as the one that is equivalent to the decomposition in the privacy blanket framework. Since the primary optimal decomposition in Theorem~\ref{primary_optimal_decomposition} can be redundant, we propose a method for computing the \emph{simplified optimal decomposition}. Specifically, by merging components of the primary optimal decomposition that share the same ratio \( a_j : b_j : c_j \), we obtain a simplified yet equivalent form of the optimal decomposition.

\begin{theorem}\label{theorem_ogc}

For a local randomizer $\mathcal{R}: \mathbb{X}\to \mathbb{Y}$ and inputs $x_1^0, x_1^1 \in \mathbb{X}$,
define $$S_{r_0,r_1}=\{y| \frac{\mathcal{R}(x_1^0)(y)}{\inf_{x} \mathcal{R}(x)(y)}=r_0, \frac{\mathcal{R}(x_1^1)(y)}{\inf_{x} \mathcal{R}(x)(y)}=r_1\}.$$
The following decomposition under the general clone paradigm is equivalent to the privacy blanket framework:
\begin{align}
\mathcal{R}(x_1^0) &= \sum_{r_0,r_1} a_{r_0,r_1} Q_{r_0,r_1} \nonumber \\ 
\mathcal{R}(x_1^1) &= \sum_{r_0,r_1} b_{r_0,r_1} Q_{r_0,r_1} \nonumber \\ 
\forall i \in [2, n], \quad \mathcal{R}(x_i) &= \sum_{r_0,r_1} c_{r_0,r_1} Q_{r_0,r_1} + \beta Q_i. \label{gc_reduced}
\end{align}
where 
\begin{align*}
a_{r_0,r_1} &= \sum_{y \in S_{r_0,r_1}} \mathcal{R}(x_1^0)(y),\\
b_{r_0,r_1} &= \sum_{y \in S_{r_0,r_1}} \mathcal{R}(x_1^1)(y),\\
c_{r_0,r_1} &= \sum_{y \in S_{r_0,r_1}} \inf_{x} \mathcal{R}(x)(y)= \sum_{y \in S_{r_0,r_1}} \gamma \omega(y).
\end{align*}
The distribution $Q_{r_0,r_1}$ is defined as below:
\begin{align*}
Q_{r_0,r_1}(y)=
\begin{cases}
\frac{R(x_1^0)(y)}{a_{r_0,r_1}}, &y\in S_{r_0,r_1}\\
0,& y\notin S_{r_0,r_1}
\end{cases} \quad.
\end{align*}

\end{theorem}

\begin{proof}
We first verify that the decomposition defined in Equation~(\ref{gc_reduced}) is valid.
For any fixed \( y^* \), there exists a unique pair \( (p, q) \) such that \( y^* \in S_{p,q} \). Then,
$$\sum_{r_0,r_1} a_{r_0,r_1} Q_{r_0,r_1}(y^*)=a_{p,q}Q_{p,q}(y^*)=a_{p,q}\frac{R(x_1^0)(y^*)}{a_{p,q}}=R(x_1^0)(y^*).$$
This confirms the validity of the first equation.

By the definition of $S_{p, q}$, for every $y \in S_{p, q}$, the values $\mathcal{R}(x_1^0)(y)$, $\mathcal{R}(x_1^1)(y)$, and $\gamma \omega(y)$ are proportional. Hence,
$$
Q_{p, q}(y^*) = \frac{\mathcal{R}(x_1^0)(y^*)}{a_{p, q}} = \frac{\mathcal{R}(x_1^1)(y^*)}{b_{p, q}} = \frac{\gamma \omega(y^*)}{c_{p, q}},
$$
which allows us to verify the second equation:
$$
\sum_{r_0, r_1} b_{r_0, r_1} Q_{r_0, r_1}(y^*) = b_{p, q} Q_{p, q}(y^*) = \mathcal{R}(x_1^1)(y^*).
$$

To verify the third equation, observe:
$$
\sum_{r_0, r_1} c_{r_0, r_1} Q_{r_0, r_1}(y^*) = c_{p, q} \cdot \frac{\gamma \omega(y^*)}{c_{p, q}} = \gamma \omega(y^*).
$$
It guarantees the existence of the \( Q_i \). 

Now, we construct the post-processing function $f_4$, as given in Algorithm~\ref{alg:post-processing_B2OGC}. The function initializes an empty set $H$, then iterates over its input. For each element $y$, it computes:
$$
r_0 = \frac{\mathcal{R}(x_1^0)(y)}{\inf_{x \in \mathbb{X}} \mathcal{R}(x)(y)}, \quad
r_1 = \frac{\mathcal{R}(x_1^1)(y)}{\inf_{x \in \mathbb{X}} \mathcal{R}(x)(y)},
$$
and appends the pair $(r_0, r_1)$ to $H$. After processing all inputs, it appends $n - 1 - |\boldsymbol{b}|$ copies of a dummy symbol $\bot$, representing the residual components of the General Clone construction.

Following the same argument as in Theorem \ref{theorem_clonevsblanket}, it can be verified that
$$
f_4(\mathcal{R}(x_1^0)) \sim A_0, \quad f_4(\mathcal{R}(x_1^1)) \sim A_1,
$$
and for all other users, the output from applying $P_0^B$ followed by $f_4$ is distributed identically to $A_2$. Here, \( A_0 \), \( A_1 \), and \( A_2 \) are defined according to the general clone reduction induced by the above simplified decomposition (see Definition~\ref{reduction_pair_gc}).
\end{proof}

\begin{algorithm}[t]
\caption{Post-processing function from privacy blanket to simplified optimal decomposition in general clone, $f_4$ }
\label{alg:post-processing_B2OGC}
\begin{algorithmic}[0]
\Meta $n$
\Require $(|\boldsymbol{b}|,Y_{|\boldsymbol{b}|})$

\State $H \gets \{\}$ \Comment{$H$ is a histogram}
\For{$y \in Y_{|\boldsymbol{b}|}$}
\State $r_0 \gets\frac{\mathcal{R}(x_1^0)(y)}{\inf_{x \in \mathbb{X}} \mathcal{R}(x)(y)}$
\State $r_1 \gets \frac{\mathcal{R}(x_1^1)(y)}{\inf_{x \in \mathbb{X}} \mathcal{R}(x)(y)}$
\State $H[(r_0,r_1)] \gets H[(r_0,r_1)]+1$ 
\EndFor

\State $H[\bot] \gets n-1-|\boldsymbol{b}|$ \Comment{The residual part of GC}

\State \Return $H$
\end{algorithmic}
\end{algorithm}

\begin{remark}
Essentially, components with the same ratio can be merged because the GPARV values over merged components remain invariant under the transformation.
We provide a detailed explanation of this point in Section~\ref{sec_gparv}.
\end{remark}

According to Theorem~\ref{theorem_ogc}, the optimal decomposition of a local randomizer can be computed through the following procedure.
First, we identify all possible values that the pair $(r_0, r_1)$ can take. Many local randomizers satisfy the \textbf{extreme property}, namely:
$$
\forall x, x',\ \forall y:\quad \frac{\mathcal{R}(x)(y)}{\mathcal{R}(x')(y)} \in \{ e^{-\varepsilon_0},\ 1,\ e^{\varepsilon_0} \}.
$$
In such cases, there are only four distinct combinations for $(r_0, r_1)$:
$$
(1,1),\quad (e^{\varepsilon_0}, 1),\quad (1, e^{\varepsilon_0}),\quad (e^{\varepsilon_0}, e^{\varepsilon_0}).
$$
As a result, the simplified optimal decomposition contains only \textit{five} components (including the residual component).

To compute the coefficients $a_{r_0, r_1}$, we first determine the subsets $S_{r_0, r_1}$, and then evaluate:
$$
c_{r_0, r_1} = \sum_{y \in S_{r_0, r_1}} \inf_{x} \mathcal{R}(x)(y).
$$
The remaining coefficients are obtained by scaling:
$$
a_{r_0, r_1} = r_0 \cdot c_{r_0, r_1}, \quad b_{r_0, r_1} = r_1 \cdot c_{r_0, r_1}.
$$

As an example, consider the \( k \)-Random Response mechanism.

\begin{example}\label{example1}
Denote \( \{1,2,\dots,k\} \) by \( [k] \) and the uniform distribution on \( [k] \) by \( \mathcal{U}_{[k]} \). 
For any \( k \in \mathbb{N} \) and \( \varepsilon_0 > 0 \), the \( k \)-randomized response mechanism \( k\text{RR}: [k] \to [k] \) is defined as:
\[
k\text{RR}(x) =
\begin{cases}
x, & \text{with probability } \frac{e^{\varepsilon_0} - 1}{e^{\varepsilon_0} + k - 1}, \\
y \sim \mathcal{U}_{[k]}, & \text{with probability } \frac{k}{e^{\varepsilon_0} + k - 1}.
\end{cases}
\]

It can be computed that:
\[
S_{e^{\varepsilon_0}, 1} = \{ x_1^0 \}, \quad
S_{1, e^{\varepsilon_0}} = \{ x_1^1 \}, \quad
S_{1,1} = [k] \setminus \{ x_1^0, x_1^1 \},
\]
and for all \( y \), we have:
\[
\inf_{x} \mathcal{R}(x)(y) = \frac{1}{e^{\varepsilon_0} + k - 1}.
\]
Its optimal decomposition is as follows:
\begin{align*}
R(x_1^0) &= e^{\varepsilon_0} p \mathbbm{1}_{x_1^0} + p \mathbbm{1}_{x_1^1} + q U, \\
R(x_1^1) &= p \mathbbm{1}_{x_1^0} + e^{\varepsilon_0} p \mathbbm{1}_{x_1^1} + q U, \\
\forall i \in [2,n]: R(x_i) &= p \mathbbm{1}_{x_1^0} + p \mathbbm{1}_{x_1^1} + q U + (e^{\varepsilon_0} - 1) p \mathbbm{1}_{x_i}.
\end{align*}
where \( p = \frac{1}{e^{\varepsilon_0} + k - 1} \), \( q = (k - 2) p \), and \( U \) is the uniform distribution over \( [k] - \{x_1^0, x_1^1\} \). This decomposition is also considered in the corrected version of the stronger clone \cite{feldman2023arxiv}.    
\end{example}

The optimal decompositions for Binary Local Hash~\cite{3}, RAPPOR~\cite{6}, Optimized Unary Encoding (OUE)~\cite{3} and Hadamard Response~\cite{HR} are provided in Appendix~\ref{appenix_decomposition}. It is worth noting that the optimal decompositions of these randomizers depend on the domain size \( D \). However, as \( D \to \infty \), each decomposition converges to a limiting form that is independent of \( D \), with a convergence rate of \( O(2^{-D}) \). We provide both the exact, \( D \)-dependent solutions in Table~\ref{tab_decomposition} and the asymptotic, \( D \)-independent solutions in Table~\ref{tab_decomposition2}.

\subsection{Optimal decompositions of the joint composition}\label{sec:joint}

In this section, we analyze the joint composition of multiple LDP local randomizers \( \mathcal{R}_1, \mathcal{R}_2, \dots, \mathcal{R}_m \). Let \( \mathcal{R}_i : \mathbb{X}^i \to \mathbb{Y}^i \) denote the local randomizer for the \( i \)-th component. The joint composition of these randomizers is defined as follows:
\[
\mathcal{R}^{\times}_{[1:m]} : \mathbb{X}^1 \times \mathbb{X}^2 \times \dots \times \mathbb{X}^m \to \mathbb{Y}^1 \times \mathbb{Y}^2 \times \dots \times \mathbb{Y}^m,
\]
\[
\mathcal{R}^{\times}_{[1:m]}(x_1, x_2, \dots, x_m) = (\mathcal{R}_1(x_1), \mathcal{R}_2(x_2), \dots, \mathcal{R}_m(x_m)).
\]
While our default assumption is that all input pairs \( \boldsymbol{x} = (x_1, x_2, \dots, x_m) \) and \( \boldsymbol{x}' = (x_1', x_2', \dots, x_m') \) are considered neighboring, our framework can easily accommodate more restricted adjacency definitions depending on the application.
In practice, it is common to use the joint composition of \( m \) mechanisms, each satisfying \( \left( \frac{\varepsilon_0}{m} \right) \)-LDP, to achieve overall \( \varepsilon_0 \)-LDP.

\begin{theorem}[Joint Composition Theorem]
If \( \mathcal{R}_i \) satisfies \( \varepsilon_i \)-LDP, then \( \mathcal{R}^\times_{[1:m]} \) satisfies \( \sum_{i=1}^m \varepsilon_i \)-LDP.
\end{theorem}

When applied to LDP mechanisms resulting from joint composition, the clone paradigm provides particularly loose decompositions. This is due to the presence of many intermediate states in the joint probability distribution after the composition. For instance, many local randomizers have probability ratios between any two inputs that belong to the set \( \{e^{-\varepsilon_0}, 1, e^{\varepsilon_0}\} \). However, after a \( m \)-fold joint composition, the ratio of the joint probability distribution may belong to \( \{e^{i\varepsilon_0} | i \in [-m,m] \cap \mathbb{Z}\} \). Under these conditions, the common part in the clone paradigm deviates significantly from the actual privacy blanket.

Fortunately, the optimal decomposition of a joint composition mechanism is simply the Cartesian product of the optimal decompositions of each individual LDP component. 

\begin{theorem}
An optimal decomposition of \( \mathcal{R}^{\times}_{[1,m]} \) can be obtained as the Cartesian product of the primary (or simplified) optimal decompositions of each individual \( \mathcal{R}_i \).
\end{theorem}

\begin{proof}
For notational simplicity, we consider the case $m = 2$. The argument naturally generalizes to arbitrary $m$.
Unlike in other contexts, we explicitly denote the two possible inputs of the first user as $(x_1^1, x_1^2)$ and $(\bar{x}_1^1, \bar{x}_1^2)$ to emphasize the joint structure across randomizers, while the inputs of the remaining users are denoted by $(x_i^1, x_i^2),i=2,\dots,n$. Here, $x_1^1, \bar{x}_1^1, x_i^1 \in \mathbb{X}^1$ and $x_1^2, \bar{x}_1^2, x_i^2 \in \mathbb{X}^2$.

Define the Cartesian product of two probability distributions \( P, Q \) as \( (P \times Q)(a, b) = P(a) \cdot Q(b) \).
The Cartesian product of the primary optimal decompositions of $\mathcal{R}_1$ and $\mathcal{R}_2$ is given by:
\begin{align}
&\left( \sum_{y^1 \in \mathbb{Y}^1} a_{y^1} \mathbbm{1}_{y^1} \right) \times \left( \sum_{y^2 \in \mathbb{Y}^2} d_{y^2} \mathbbm{1}_{y^2} \right) \nonumber \\
&\left( \sum_{y^1 \in \mathbb{Y}^1} b_{y^1} \mathbbm{1}_{y^1} \right) \times \left( \sum_{y^2 \in \mathbb{Y}^2} e_{y^2} \mathbbm{1}_{y^2} \right) \nonumber \\
&\left( \sum_{y^1 \in \mathbb{Y}^1} c_{y^1} \mathbbm{1}_{y^1} + \beta_1 Q_i \right) \times \left( \sum_{y^2 \in \mathbb{Y}^2} f_{y^2} \mathbbm{1}_{y^2} + \beta_2 T_i \right), \label{product_od}
\end{align}
where
\begin{align}
a_{y^1} &= \mathcal{R}_1(x_1^1)(y^1), & d_{y^2} &= \mathcal{R}_2(x_1^2)(y^2)\nonumber  \\
b_{y^1} &= \mathcal{R}_1(\bar{x}_1^1)(y^1), & e_{y^2} &= \mathcal{R}_2(\bar{x}_1^2)(y^2)\nonumber  \\
c_{y^1} &= \inf_{x^1 \in \mathbb{X}^1} \mathcal{R}_1(x^1)(y^1), & f_{y^2} &= \inf_{x^2 \in \mathbb{X}^2} \mathcal{R}_2(x^2)(y^2)\label{product_od_eff}
\end{align}

By Theorem~\ref{primary_optimal_decomposition}, an optimal decomposition of $\mathcal{R}^{\times}$ takes the form:
\begin{align}
\mathcal{R}^{\times}(x_1^1, x_1^2) &= \sum_{y^1, y^2} a_{y^1, y^2}' \mathbbm{1}_{(y^1, y^2)} \nonumber \\
\mathcal{R}^{\times}(\bar{x}_1^1, \bar{x}_1^2) &= \sum_{y^1, y^2} b_{y^1, y^2}' \mathbbm{1}_{(y^1, y^2)} \nonumber \\
\forall i \in [2,n],\quad \mathcal{R}^{\times}(x_i^1, x_i^2) &= \sum_{y^1, y^2} c_{y^1, y^2}' \mathbbm{1}_{(y^1, y^2)} + \beta_3 V_i \label{joint_od}
\end{align}
where
\begin{alignat*}{4}
a_{y^1, y^2}' &=\ \mathbb{R}^{\times}(x_1^1, x_1^2)(y^1, y^2)  &\quad =\ &\ \mathcal{R}_1(x_1^1)(y^1) \cdot \mathcal{R}_2(x_1^2)(y^2)  &\quad =\ &\ a_{y^1} \cdot d_{y^2} \\
b_{y^1, y^2}' &=\ \mathbb{R}^{\times}(\bar{x}_1^1, \bar{x}_1^2)(y^1, y^2)  &\quad =\ &\ \mathcal{R}_1(\bar{x}_1^1)(y^1) \cdot \mathcal{R}_2(\bar{x}_1^2)(y^2)  &\quad =\ &\ b_{y^1} \cdot e_{y^2} \\
c_{y^1, y^2}' &=\ \inf_{x^1, x^2} \mathbb{R}^{\times}(x^1, x^2)(y^1, y^2)  &\quad =\ &\ \inf_{x^1} \mathbb{R}_1(x^1)(y^1) \cdot \inf_{x^2} \mathbb{R}_2(x^2)(y^2)  &\quad =\ &\ c_{y^1} \cdot f_{y^2}.
\end{alignat*}

In each line, the first equality is derived from the definition of the primary optimal decomposition of $\mathcal{R}^{\times}$. The second equality follows from the definition of joint composition. The third equality results from (\ref{product_od_eff}), which specifies the primary optimal decompositions of $\mathcal{R}_1$ and $\mathcal{R}_2$.

A comparison between Equations~(\ref{product_od}) and (\ref{joint_od}) reveals that the Cartesian product construction not only recovers the full primary optimal decomposition of $\mathcal{R}^{\times}$, but also exposes a more granular structure within the left-over distribution: $$V_i=\beta_1 Q_i \times \big(\sum_{y^2 \in \mathbb{Y}^2} f_{y^2} \mathbbm{1}_{y^2}\big)+\beta_2 \big(\sum_{y^1 \in \mathbb{Y}^1} c_{y^1} \mathbbm{1}_{y^1}\big)\times T_i+(\beta_1\beta_2) Q_i\times T_i.$$ For practical purposes, these components can be combined into a single left-over distribution.

Next, we show that the Cartesian product of the simplified Optimal Decompositions of $\mathcal{R}_1$ and $\mathcal{R}_2$ also forms an Optimal Decomposition of $\mathcal{R}^{\times}$.
Recall that a simplified Optimal Decomposition is formed by grouping according to the following sets:
\begin{align*}
S_{p_0, p_1}^1 &= \left\{ y \in \mathbb{Y}^1 \,\middle|\, \frac{\mathcal{R}_1(x_1^1)(y)}{\inf_x \mathcal{R}_1(x)(y)} = p_0,\ \frac{\mathcal{R}_1(\bar{x}_1^1)(y)}{\inf_x \mathcal{R}_1(x)(y)} = p_1 \right\} \\
S_{q_0, q_1}^2 &= \left\{ y \in \mathbb{Y}^2 \,\middle|\, \frac{\mathcal{R}_2(x_1^2)(y)}{\inf_x \mathcal{R}_2(x)(y)} = q_0,\ \frac{\mathcal{R}_2(\bar{x}_1^2)(y)}{\inf_x \mathcal{R}_2(x)(y)} = q_1 \right\}
\end{align*}
Each component in the product decomposition corresponds to the Cartesian product set $S_{p_0, p_1}^1 \times S_{q_0, q_1}^2$.
According to the rule of simplified Optimal Decomposition, all elements with the same ratio
$$
a_{y^1, y^2}': b_{y^1, y^2}': c_{y^1, y^2}'
$$
can be merged. For all $(y^1, y^2) \in S_{p_0, p_1}^1 \times S_{q_0, q_1}^2$, we have:
$$
a_{y^1, y^2}': b_{y^1, y^2}': c_{y^1, y^2}' = p_0 q_0 : p_1 q_1 : 1.
$$
This completes the proof.
\end{proof}

\begin{remark}
There is a minor gap in the proof: in the simplified optimal decomposition, all components of the primary optimal decomposition that share the same ratio \( a_j : b_j : c_j \) are merged into a single component, whereas the decomposition obtained via the Cartesian product yields a finer partition. 

On the one hand, the proof technique used for the simplified optimal decomposition in Theorem \ref{theorem_ogc} can be naturally extended to handle this more granular structure. More fundamentally, this discrepancy can be explained through the lens of the GPARV, as we elaborate in detail in Section~\ref{sec_gparv}.
\end{remark}

As an example, consider the joint composition of two \( k \)-RR mechanisms acting on \( [k] \times [k] \). The optimal decomposition is given below:

\begin{example}[Optimal Decomposition of Two-Joint \( k \)-RR]

\[
\begin{array}{rl}
\mathcal{R}^{\times}(a_1^0, b_1^0) 
&= \left( e^{\varepsilon_0/2} p\, \mathbbm{1}_{a_1^0} + p\, \mathbbm{1}_{a_1^1} + q\, U \right)
   \times \left( e^{\varepsilon_0/2} p\, \mathbbm{1}_{b_1^0} + p\, \mathbbm{1}_{b_1^1} + q\, V \right), \\[0.5ex]
\mathcal{R}^{\times}(a_1^1, b_1^1) 
&= \left( p\, \mathbbm{1}_{a_1^0} + e^{\varepsilon_0/2} p\, \mathbbm{1}_{a_1^1} + q\, U \right)
   \times \left( p\, \mathbbm{1}_{b_1^0} + e^{\varepsilon_0/2} p\, \mathbbm{1}_{b_1^1} + q\, V \right), \\[0.5ex]
\forall i \in [2,n]:\ \mathcal{R}^{\times}(a_i, b_i)
&= \left( p\, \mathbbm{1}_{a_1^0} + p\, \mathbbm{1}_{a_1^1} + q\, U + (1 - 2p - q)\, \mathbbm{1}_{a_i} \right) \\
&\quad \times \left( p\, \mathbbm{1}_{b_1^0} + p\, \mathbbm{1}_{b_1^1} + q\, V + (1 - 2p - q)\, \mathbbm{1}_{b_i} \right)
\end{array}
\]
where \( p = \frac{1}{e^{\varepsilon_0/2} + k - 1} \), \( q = \frac{k - 2}{e^{\varepsilon_0/2} + k - 1} \), and \( U \), \( V \) are the uniform distributions over \( [k] \setminus \{a_1^0, a_1^1\} \) and \( [k] \setminus \{b_1^0, b_1^1\} \), respectively.
\end{example}

\begin{remark}
The optimal decompositions of other jointly composed local randomizers can be computed in a similar manner. For simplicity, all examples and experiments in this paper consider the joint composition of multiple instances of the \emph{same} type of local randomizer, each with the \emph{same} privacy parameter \(\varepsilon_0\). However, the Cartesian product approach described above naturally generalizes to joint compositions involving \emph{different} types of local randomizers, each potentially with a \emph{different} value of \(\varepsilon_0\).
\end{remark}

\subsection{Optimal decompositions of the parallel composition}
The joint composition setting discussed earlier is well-suited for analyzing the joint distribution of multiple attributes associated with each user. In contrast, when the analyst is interested in the marginal distributions of individual attributes or aims to compute \emph{multiple} statistical queries (each corresponding to a different local mechanism \(\mathcal{R}_1, \dots, \mathcal{R}_m\)), a more appropriate strategy is \emph{parallel composition}. It is defined as:
\[
\mathcal{R}_{[1,m]}^{\parallel,\boldsymbol{p}}(x) = (i, \mathcal{R}_i(x)), \quad \text{with probability } p_i, \quad i = 1, 2, \dots, m,
\]
where \(\boldsymbol{p} = (p_1, \dots, p_m)\) is a probability vector such that \(\sum_{i=1}^m p_i = 1\).

\begin{theorem}
If \( \mathcal{R}_i \) satisfies \( \varepsilon \)-LDP, then \( \mathcal{R}_{[1,m]}^{\parallel,\boldsymbol{p}} \) satisfies \( \varepsilon \)-LDP.
\end{theorem}

The following theorem provides the method for computing the optimal decomposition of Parallel composition.
\begin{theorem}
An optimal decomposition of $\mathcal{R}^{\parallel,\boldsymbol{p}}_{[1:m]}$ is a weighted combination of the primary (or simplified) optimal decompositions of the individual local randomizers $\mathcal{R}_i$, with weights given by $p_i$.
\end{theorem}

\begin{proof}
For clarity, we consider the case $m = 2$. The argument naturally generalizes to arbitrary $m$.

Let $\mathcal{R}_1: \mathbb{X} \to \mathbb{Y}^1$ and $\mathcal{R}_2: \mathbb{X} \to \mathbb{Y}^2$ be two local randomizers considered under parallel composition.
Suppose the optimal decompositions of them are given by the coefficients \(a_j, b_j, c_j, \beta_1\) for \(j \in [|\mathbb{Y}^1|]\), and \(d_j, e_j, f_j, \beta_2\) for \(j \in [|\mathbb{Y}^2|]\), respectively:
\begin{align*}
\mathcal{R}_1(x_1^0) &= \sum_{y^1\in \mathbb{Y}^1} a_{y^1}\mathbbm{1}_{y^1} , &\mathcal{R}_2(x_1^0) &= \sum_{y^2\in \mathbb{Y}^2} d_{y^2}\mathbbm{1}_{y^2}, \\
\mathcal{R}_1(x_1^1) &= \sum_{y^1\in \mathbb{Y}^1} b_{y^1}\mathbbm{1}_{y^1}, &\mathcal{R}_2(x_1^1) &= \sum_{y^2\in \mathbb{Y}^2} e_{y^2}\mathbbm{1}_{y^2}, \\
\forall i \in [2, n], \mathcal{R}_1(x_i) &= \sum_{y^1\in \mathbb{Y}^1}c_j \mathbbm{1}_{y^1} + \beta_1 Q_i, &\mathcal{R}_2(x_i) &= \sum_{y^2\in \mathbb{Y}^2} f_{y^2}\mathbbm{1}_{y^2} + \beta_2 T_i,
\end{align*}
where
\begin{align*}
a_{y^1}&=\mathcal{R}_1(x_1^0)(y^1),& d_{y^2}&=\mathcal{R}_2(x_1^0)(y^2),\\
b_{y^1}&=\mathcal{R}_1(x_1^1)(y^1),& e_{y^2}&=\mathcal{R}_2(x_1^0)(y^2),\\
c_{y^1}&=\inf_{x\in \mathbb{X}}\mathcal{R}_1(x)(y^1),& f_{y^2}&=\inf_{x\in \mathbb{X}}\mathcal{R}_2(x)(y^2).
\end{align*}
Then, the weighted combination of them with \(\boldsymbol{p} = (p_1, p_2)\) is:
\begin{align}
\mathcal{R}^{\parallel}(x_1^0) &= p_1 \sum_{y^1\in \mathbb{Y}^1} a_{y^1}\mathbbm{1}_{y^1} + p_2 \sum_{y^2\in \mathbb{Y}^2} d_{y^2}\mathbbm{1}_{y^2}, \nonumber \\
\mathcal{R}^{\parallel}(x_1^1) &= p_1 \sum_{y^1\in \mathbb{Y}^1} b_{y^1}\mathbbm{1}_{y^1} + p_2 \sum_{y^2\in \mathbb{Y}^2} e_{y^2}\mathbbm{1}_{y^2}, \nonumber\\
\forall i \in [2, n],\ \mathcal{R}^{\parallel}(x_i) &= p_1 \left( \sum_{y^1\in \mathbb{Y}^1}c_j \mathbbm{1}_{y^1} + \beta_1 Q_i \right)
+ p_2 \left( \sum_{y^2\in \mathbb{Y}^2} f_{y^2}\mathbbm{1}_{y^2} + \beta_2 T_i \right). \label{weighted_od}
\end{align}

By Theorem~\ref{primary_optimal_decomposition}, the primary optimal decomposition of $\mathcal{R}^{\parallel,\boldsymbol{p}}$ takes the form:
\begin{align}
\mathcal{R}^{\parallel,\boldsymbol{p}}(x_1^0) &= \sum_{(1,y^1)} a_{1,y^1}' \mathbbm{1}_{(1,y^1)} +\sum_{(2,y^2)} a_{2,y^2}' \mathbbm{1}_{(2,y^2)} \nonumber \\
\mathcal{R}^{\parallel,\boldsymbol{p}}(x_1^1) &= \sum_{(1,y^1)} b_{1,y^1}' \mathbbm{1}_{(1,y^1)} +\sum_{(2,y^2)} b_{2,y^2}' \mathbbm{1}_{(2,y^2)} \nonumber \\
\forall i \in [2,n],\mathcal{R}^{\parallel,\boldsymbol{p}}(x_i) &= \sum_{(1,y^1)} c_{1,y^1}' \mathbbm{1}_{(1,y^1)} +\sum_{(2,y^2)} c_{2,y^2}' \mathbbm{1}_{(2,y^2)} + \beta_3 V_i \label{parallel_od}
\end{align}
where
\begin{align*}
a_{1, y^1}' &= p_1\mathcal{R}_1(x_1^0)(y^1), &a_{2, y^2}' &= p_2\mathcal{R}_2(x_1^0)(y^2),\\ 
b_{1, y^1}' &= p_1\mathcal{R}_1(x_1^1)(y^1), &a_{2, y^2}' &= p_2\mathcal{R}_2(x_1^1)(y^2),\\ 
c_{1, y^1}' &= p_1\inf_{x\in \mathbb{X}}\mathcal{R}_1(x)(y^1), &c_{2, y^2}' &= p_2\inf_{x\in \mathbb{X}}\mathcal{R}_2(x)(y^2).
\end{align*}

A comparison between Equations~(\ref{weighted_od}) and (\ref{parallel_od}) reveals that the weighted combination construction not only recovers the full primary optimal decomposition of $\mathcal{R}^{\parallel,\boldsymbol{p}}$, but also exposes a more granular structure within the left-over distribution: $$V_i=p_1\beta_1 Q_i+p_2\beta_2T_i.$$ For practical purposes, these components can be combined into a single left-over distribution.

Next, consider the simplified optimal decomposition. For each $\mathcal{R}_i$ ($i=1,2$), let the decomposition group outputs into regions defined by:
$$
S_{r_0^{(i)}, r_1^{(i)}} = \left\{ y \in \mathbb{Y}^i \,\middle|\, \frac{\mathcal{R}_i(x^0)(y)}{\inf_{x} \mathcal{R}_i(x)(y)} = r_0^{(i)},\ \frac{\mathcal{R}_i(x^1)(y)}{\inf_{x} \mathcal{R}_i(x)(y)} = r_1^{(i)} \right\}.
$$

Then the parallel mechanism's output space is partitioned by the sets $S_{r_0^{(1)}, r_1^{(1)}}$ and $S_{r_0^{(2)}, r_1^{(2)}}$, and each such subset preserves the same ratio of coefficients:
$$
\text{For all } y \in S_{r_0^{(i)}, r_1^{(i)}},\quad a_{i,y'} : b_{i,y'} : c_{i,y'} = r_0^{(i)} : r_1^{(i)} : 1.
$$

Thus, the simplified decomposition of the parallel mechanism is also obtained as a mixture of simplified decompositions of the individual mechanisms. The merging step in the simplified decomposition remains valid since the ratios remain constant within each group.

This completes the proof.
\end{proof}

\begin{example}
A classical example of parallel composition is computing the mean of an $m$-dimensional numerical vector, where $\vec{x} = (x^1, x^2, \dots, x^m) \in [0,1]^m$. Define $R_i(\vec{x}) =\big(i, x^i + \mathrm{Lap}(1/\varepsilon_0)\big)$ for $i = 1, 2, \dots, m$, each satisfying $\varepsilon_0$-LDP. When the probability vector is uniform, i.e., $\boldsymbol{p} = (1/m, 1/m, \dots, 1/m)$, the mechanism $\mathcal{R}^{\parallel,\boldsymbol{p}}_{[1,m]}$ corresponds to the setting where each user randomly selects one coordinate and reports it using the Laplace mechanism.
\end{example}

\begin{example}
An important theoretical advantage of parallel composition is that it subsumes Poisson subsampling, as discussed in Section~\ref{sec:contributions}. In the experimental section, we present the results of our method in computing upper bounds on the privacy amplification of the $k$-RR mechanism under Poisson subsampling followed by shuffling.
\end{example}

\section{Efficient Algorithm for Computing the Optimal Bounds}\label{sec_our_algo}
In this section, we present an efficient algorithm to compute the optimal privacy amplification bounds within the general clone framework. Our approach builds upon the previously overlooked concept of the \textit{privacy amplification random variable} (PARV), originally proposed in the privacy blanket framework~\cite{Balle2019}. While PARV offers an exact expression for the amplification bound, the original work did not provide a method for computing it precisely, relying instead on loose analytical approximations.

To overcome this limitation, we introduce the \textit{generalized privacy amplification random variable} (GPARV), which enables precise and efficient numerical computation of the optimal bounds via FFT.

In addition to computing upper bounds, our method can also be used to compute privacy amplification \emph{lower bounds}, which serve as a reference for evaluating the tightness of the computed upper bounds. Experimental results (see Section \ref{sec_experiment}) confirm that the gap between these upper and lower bounds is typically small, demonstrating the accuracy and effectiveness of our approach.

For clarity, we directly derive the GPARV within the general clone framework before describing its relationship to the original PARV.

\subsection{Generalized Privacy amplification random variable}\label{sec_gparv}

The following theorem is a remarkable result from the privacy blanket paper \cite{Balle2019}. We restate it here in a more general form.

\begin{theorem}\label{main_theorem}
Let $A_0$, $A_1$, and $A_2$ be three probability distributions such that $A_0$ and $A_1$ are absolutely continuous with respect to $A_2$, i.e., any measurable set $S$, if $A_2(S) = 0$, then $A_0(S) = A_1(S) = 0$.

Define two distributions $P_0,P_1$ of the histograms generated by sampling:
\begin{itemize}
    \item One sample from \( A_0 \) for \( P_0 \) (or \( A_1 \) for \( P_1 \));
    \item \( n-1 \) i.i.d. samples from \( A_2 \).
\end{itemize}

Define a random variable \( G = \frac{A_0(y) - e^\varepsilon A_1(y)}{A_2(y)} \), where \( y \sim A_2 \).
Suppose \( G_1, G_2, \dots \) are i.i.d. copies of \( G \). Then, we have the following:
\[
D_{e^\varepsilon} (P_0  \parallel  P_1) =\frac{1}{n} \mathbb{E} \left[\max \left\{0, \sum_{i=1}^{n} G_i \right\} \right]= \frac{1}{n} \mathbb{E} \left[ \sum_{i=1}^{n} G_i \right]_{+}.
\]
\end{theorem}
We provide the proof in Appendix~\ref{proof_main_theorem}. It is straightforward to observe that Theorem~\ref{main_theorem} can be directly applied to compute the general clone bound $D_{e^{\varepsilon}}(P_0^{\mathsf{GC}}, P_1^{\mathsf{GC}})$. Consider the general clone induced by the primary optimal decomposition; the corresponding random variable $G$ can be expressed as follows:

\begin{definition}[PARV and GPARV]\label{definition:gparv}
Suppose $W \sim \omega$ is a $\mathbb{Y}$-valued random variable sampled from the blanket. For any $\varepsilon > 0$ and $x, x' \in \mathbb{X}$, Define two random variable:
\begin{align*}
 L^{x,x'}_\varepsilon &= \frac{\mathcal{R}(x)(W) - e^\varepsilon \mathcal{R}(x')(W)}{\omega(W)}.\\
\text{G}_{\varepsilon}^{x,x'} &= \begin{cases}
\frac{1}{\gamma} L_{\varepsilon}^{x,x'}, & \text{w.p.} \ \gamma, \\
0, & \text{w.p.} \ 1-\gamma.
\end{cases}
\end{align*}
$\omega$ and $\gamma$ are defined in Section \ref{sec_blanket}.
\end{definition}
The quantity \( L_{\varepsilon}^{x,x'} \) was introduced by Balle et al.~\cite{Balle2019} and referred to as the \emph{Privacy Amplification Random Variable (PARV)}. We refer to \( G_{\varepsilon}^{x,x'} \) as the \emph{Generalized Privacy Amplification Random Variable (GPARV)}. The following theorem follows directly from Theorem~\ref{main_theorem} and the definition of the primary optimal decomposition in Section \ref{sec:blanket_in_gc}.

\begin{theorem}\label{our_main_theorem}
Let \( \mathcal{R} : \mathbb{X} \to \mathbb{Y} \) be a local randomizer and let \( \mathcal{M}_S = \mathcal{S} \circ \mathcal{R}^n \) be the shuffling of \( \mathcal{R} \). Fix \( \varepsilon \geq 0 \) and inputs \( X^0 \simeq X^1 \) with \( x_1^0 \neq x_1^1 \). Suppose \( G_1, G_2, \dots \) are i.i.d. copies of \( G_{\varepsilon}^{x_1^0, x_1^1} \) and \( \gamma \) is defined as in Section \ref{sec_blanket}. Then, we have:
\begin{align}
D_{e^\varepsilon} (\mathcal{M}_S(X^0) \| \mathcal{M}_S(X^1)) 
\leq D_{e^\varepsilon} ( P_0^{B} \| P_1^{B}) \nonumber 
= \frac{1}{n} \mathbb{E} \left[ \sum_{i=1}^{n} G_i \right]_{+}. 
\end{align}
\end{theorem}

We now briefly revisit how the original privacy blanket paper utilizes the PARV.
Using PARV, Balle et al. derived the precise expression for $D_{e^\varepsilon} ( P_0^{B} \| P_1^{B})$:

\begin{lemma}[Lemma 5.3 in \cite{Balle2019}]\label{lemma3}
Let \( \mathcal{R} : \mathbb{X} \to \mathbb{Y} \) be a local randomizer and let \( \mathcal{M}_S = \mathcal{S} \circ \mathcal{R}^n \) be the shuffling of \( \mathcal{R} \). Fix \( \varepsilon \geq 0 \) and inputs \( X^0 \simeq X^1 \) with \( x_1^0 \neq x_1^1 \). Suppose \( L_1, L_2, \dots \) are i.i.d. copies of \( L_{\varepsilon}^{x_1^0, x_1^1} \) and \( \gamma \) is defined as in Section \ref{sec_blanket}. Then, we have:
\begin{align}
D_{e^\varepsilon} (\mathcal{M}_S(X^0) \| \mathcal{M}_S(X^1))
\leq D_{e^\varepsilon} ( P_0^{B} \| P_1^{B}) 
&= \frac{1}{\gamma n} \sum_{m=1}^{n} \binom{n}{m} \gamma^m (1 - \gamma)^{n-m} 
\mathbb{E} \left[ \sum_{i=1}^{m} L_i \right]_{+} \nonumber \\
&=\frac{1}{\gamma n} \mathbb{E}_{M \sim \text{Bin}(n, \gamma)} \left[ \sum_{i=1}^{M} L_i \right]_{+},\label{f2}
\end{align}
where \( M \sim \text{Bin}(n, \gamma) \). We use the convention \( \sum_{i=1}^{m} L_i = 0 \) when \( m = 0 \).
\end{lemma}

Unfortunately, Balle et al. stopped at this point and did not pursue further simplification. In fact, the proof of Theorem~\ref{our_main_theorem} can be obtained by algebraically manipulating the expression in Lemma~\ref{lemma3}, using the definitions of PARV and GPARV. However, a more fundamental derivation arises directly from Theorem~\ref{main_theorem} and the definition of the general clone reduction with the primary optimal decomposition. By introducing the general clone framework, we provide a simpler and more principled proof of the specific bound established by the privacy blanket.

From the perspective of GPARV, the construction of the simplified optimal decomposition in Section~\ref{sec_simplified_od} becomes quite natural. The reason why components of the primary optimal decomposition that share the same ratio \( a_j : b_j : c_j \) can be merged is that the corresponding GPARV takes the same value on each of these components, namely \( \frac{a_j - e^{\varepsilon} b_j}{c_j} \). It can be verified that the GPARV defined by the decomposition remains unchanged after such merging.

The GPARV has the following properties:
\begin{property}\label{property2}
Let \( \mathcal{R} : \mathbb{X} \to \mathbb{Y} \) be an \( \varepsilon_0 \)-LDP local randomizer. For any \( \varepsilon \geq 0 \) and \( x, x' \in \mathbb{X} \), the generalized privacy amplification random variable \( G = G_{\varepsilon}^{x, x'} \) satisfies:
\begin{enumerate}
    \item \( \mathbb{E}[G] = 1 - e^{\varepsilon} \),
    \item \( 1 - e^{\varepsilon + \varepsilon_0} \leq G \leq e^{\varepsilon_0} - e^{\varepsilon} \).
\end{enumerate}
\end{property}

\begin{proof}
The first property follows from direct computation:
$$ \mathbb{E}[G] = \gamma \mathbb{E}\left[\frac{1}{\gamma} L\right] = \mathbb{E}[L] = \mathbb{E}_{W \sim \omega}\left[\frac{\mathcal{R}(x)(W) - e^\varepsilon \mathcal{R}(x')(W)}{\omega(W)}\right] = 1 - e^{\varepsilon}. $$

The second property is due to the \( \varepsilon_0 \)-DP property of \( \mathcal{R} \): \( \forall x \in \mathbb{X}, y \in \mathbb{Y}: 1 \leq \frac{\mathcal{R}(x)(y)}{\gamma \omega(y)} \leq e^{\varepsilon_0} \), so
$$ 1 - e^{\varepsilon + \varepsilon_0} \leq \frac{\mathcal{R}(x)(W) - e^\varepsilon \mathcal{R}(x')(W)}{\gamma \omega(W)} \leq e^{\varepsilon_0} - e^\varepsilon.$$ 
\end{proof}

\begin{remark}
In the original paper of the privacy blanket framework, a similar property of PARV was provided, but in a loose form \cite{Balle2019}. Specifically, they established that \( \gamma (e^{-\varepsilon_0} - e^{\varepsilon + \varepsilon_0}) \leq L \leq \gamma (e^{\varepsilon_0} - e^{\varepsilon - \varepsilon_0}) \). However, the \( \varepsilon_0 \)-DP property of \( \mathcal{R} \) actually guarantees a tighter bound: \( \gamma (1 - e^{\varepsilon + \varepsilon_0}) \leq L \leq \gamma (e^{\varepsilon_0} - e^{\varepsilon}) \).
\end{remark}

\subsection{Algorithm for computing privacy amplification upper bounds}
We present a new algorithm for computing the optimal privacy amplification upper bound under the general clone paradigm for any specific local randomizer, as described in Algorithm~\ref{alg:fft}.

\paragraph{Overview.}
First, the distribution of the generalized privacy amplification random variable \( G_{\varepsilon}^{x,x'} \) is computed for a given local randomizer. Since many local randomizers exhibit input symmetry, the distribution of \( G \) typically does not depend on the specific values of \( x \) and \( x' \). We thus denote it simply by \( G_{\varepsilon} \). For local randomizers lacking input symmetry, additional considerations are required; we discuss this case in detail in Section~\ref{sec_limitation}.

The distributions of \( G_{\varepsilon} \) for commonly used local randomizers are summarized in Table \ref{tab2}. Since the optimal decompositions and corresponding GPARVs of BLH, RAPPOR, OUE, and HR depend on the domain size \( D \) (see Table \ref{tab_decomposition} in Appendix~A), we present their GPARVs in the asymptotic regime (i.e., for \( D \gg 1 \)) for simplicity.
The distribution of \( G_{\varepsilon} \) for Laplace mechanism on $\{0,1\}$ is presented in Appendix~\ref{appendix_laplace}. (The Laplace mechanism over the interval \([0,1]\) does not satisfy input symmetry.) 

Given the distribution of \( G_\varepsilon \), our algorithm discretizes it to obtain \( \bar{G} \), by rounding each value in \( G \) up to the nearest larger multiple of a discretization interval length \( l \):
$$\text{Round}(x)=l\cdot \lceil \frac{x}{l} \rceil.$$

Next, the algorithm computes the \( n \)-fold convolution of \( \bar{G} \), denoted \( \bar{G}^{*n} \), using the classical Fast Fourier Transform method:
\[
\bar{G}^{*n} = \text{FFT}^{-1}\left((\text{FFT}(\bar{G}))^{\odot n}\right),
\]
where \( \odot n \) represents the element-wise exponentiation by \( n \).

Finally, the algorithm evaluates the integral
\[
I = \mathbb{E}[\bar{G}^{*n}]_{+} = \int_0^{+\infty} x \, \bar{G}^{*n}(x) \, dx,
\]
and outputs \( \frac{I}{n} \) as the upper bound on privacy amplification.

\paragraph{Correctness.}
The output of the algorithm is guaranteed to upper bound the true value, since the discretized distribution \( \bar{G} \) stochastically dominates \( G \):
\[
\frac{1}{n} \mathbb{E}[\bar{G}^{*n}]_{+} \geq \frac{1}{n} \mathbb{E}[G^{*n}]_{+}.
\]

\paragraph{Time Complexity and Error Analysis.}
Because \( G \) is supported on the interval \( [1 - e^{\varepsilon_0+\varepsilon}, e^{\varepsilon_0} - e^{\varepsilon}] \) (Property \ref{property2}), the discretization step requires \( O\left( \frac{e^{\varepsilon_0+\varepsilon}}{l} \right) \) operations. For most frequency oracles used on categorical data, \( G \) takes values on at most five points (see Table~\ref{tab2}), resulting in \( O(1) \) discretization time.

The discretization introduces a bounded error. An intuitive analysis is as follows:
\[
\mathbb{E}[\bar{G}^{*n}]_{+} - nl \leq \mathbb{E}[G^{*n}]_{+} \leq \mathbb{E}[\bar{G}^{*n}]_{+}.
\]
The FFT computation runs in \( O\left( \frac{n}{l} \log\left( \frac{n}{l} \right) \right)=\widetilde{O}(\frac{n}{l}) \) time. Choosing \( l = O\left( \frac{1}{n} \right) \) ensures an \( O(1) \) additive error in total \( \widetilde{O}(n^2) \) time.
More precise analysis can be done:
\[
\mathbb{E}[\bar{G}^{*n}]_{+} - nl \cdot \Pr\left[ [G^{*n}]_{+} > -nl \right] \leq \mathbb{E}[G^{*n}]_{+} \leq \mathbb{E}[\bar{G}^{*n}]_{+}.
\]
By Hoeffding’s inequality,
\[
\Pr\left[ [G^{*n}]_{+} > -nl \right] \leq \exp\left(-\left( \frac{2a^2}{b^2} - l \right) n \right),
\]
where \( a = -\mathbb{E}[G] = e^\varepsilon - 1 \) and \( b = (e^{\varepsilon_0} - 1)(e^\varepsilon + 1) \). This shows that the error decays exponentially in \( n \) as long as \( l \leq \frac{2a^2}{b^2} \). 

The above analysis concerns the error introduced when using the FFT-based algorithm to compute \( D_{e^\varepsilon}\big(P_0^{B}, P_1^{B} \big) \). From a theoretical perspective, since \( \delta \) is typically very small, achieving high-precision computation (i.e., with error \( o(\delta) \)) would require a very small discretization step size \( l \), which incurs a high computational cost.

However, our experimental results show that when discretization is used to infer the smallest value of \( \varepsilon \) such that \( D_{e^\varepsilon}\big(P_0^{B}, P_1^{B} \big) \le \delta \), the resulting error in \( \varepsilon \) is negligible in practice. This is because the distribution of the GPARV is highly sensitive to changes in \( \varepsilon \). In particular, once \( \varepsilon \) exceeds the true threshold \( \varepsilon^* \) (defined by \( D_{e^{\varepsilon^*}}\big(P_0^{B}, P_1^{B} \big) = \delta \)), the function \( \delta(\varepsilon) \) decreases rapidly. This rapid decay ensures that even a relatively coarse discretization yields an accurate estimate of \( \varepsilon \).
Empirical evaluations confirm that \( l = O(1) \) is sufficiently accurate in practice (see Section \ref{sec_experiment}), and the overall FFT runtime becomes \( O(n \log n) \).

\begin{algorithm}[t]
\caption{Calculate the optimal privacy amplification bound via FFT }
\label{alg:fft}
\begin{algorithmic}[0]
\Require the distritbuion density $G$ of $G_\varepsilon$, number of users $n$, discretisation interval length $l$
\State $\bar{G}=$ Discretize($G$,$l$)\quad \Comment{Round every point of $G$ to $nl,n\in \mathbb{Z}$}
\State $\bar{G}^{*}= \text{FFT}^{-1}((\text{FFT}(\bar{G}))^{\odot n})$ \quad \Comment{Compute the $n$-fold convolution}
\State $I \gets \int_0^{+\infty} x\bar{G}^{*}(x)dx$ \Comment{Compute the integral }
\State \Return $\frac{I}{n}$
\end{algorithmic}
\end{algorithm}

\begin{table*}[t]
    \centering
    \caption{The probability distribution of GPARV for common local randomizers} \label{tab2}
    \small{(The GPARVs for BLH, RAPPOR, OUE and HR are asymptotic as $|\mathbb{X}|\to \infty$.)}
    \begin{tabular}{cccccc}
    \toprule
         & $1-e^{\varepsilon_0+\varepsilon}$ & $e^{\varepsilon_0}-e^{\varepsilon_0+\varepsilon}$ & $1-e^{\varepsilon}$ & 0 & $e^{\varepsilon_0}-e^{\varepsilon}$\\
         \midrule
    $k$-RR \cite{11}     & $\frac{1}{e^{\varepsilon_0}+k-1}$ & 0 & $\frac{k-2}{e^{\varepsilon_0}+k-1}$ & $\frac{e^{\varepsilon_0}-1}{e^{\varepsilon_0}+k-1}$ & $\frac{1}{e^{\varepsilon_0}+k-1}$ \\
    BLH \cite{3} & $\frac{1}{2(e^{\varepsilon_0}+1)}$ & $\frac{1}{2(e^{\varepsilon_0}+1)}$ & $\frac{1}{2(e^{\varepsilon_0}+1)}$ & $\frac{e^{\varepsilon_0}-1}{e^{\varepsilon_0}+1}$ & $\frac{1}{2(e^{\varepsilon_0}+1)}$\\
    RAPPOR \cite{6} & $\frac{1}{(e^{\varepsilon_0/2}+1)^2}$ & $\frac{1}{e^{\varepsilon_0/2}(e^{\varepsilon_0/2}+1)^2}$ & $\frac{e^{\varepsilon_0/2}}{(e^{\varepsilon_0/2}+1)^2}$ & $1-e^{-\varepsilon_0/2}$ & $\frac{1}{(e^{\varepsilon_0/2}+1)^2}$ \\
    OUE \cite{3}   & $\frac{1}{2(e^{\varepsilon_0}+1)}$ & $\frac{1}{2e^{\varepsilon_0}(e^{\varepsilon_0}+1)}$ & $\frac{e^{\varepsilon_0}}{2(e^{\varepsilon_0}+1)}$ & $\frac{1}{2}(1-e^{-\varepsilon_0})$ & $\frac{1}{2(e^{\varepsilon_0}+1)}$ \\
    HR \cite{HR}    & $\frac{1}{2(e^{\varepsilon_0}+1)}$ & $\frac{1}{2(e^{\varepsilon_0}+1)}$ & $\frac{1}{2(e^{\varepsilon_0}+1)}$ & $\frac{e^{\varepsilon_0}-1}{e^{\varepsilon_0}+1}$ & $\frac{1}{2(e^{\varepsilon_0}+1)}$\\
    \bottomrule
    \end{tabular}
\end{table*}

\paragraph{Comparison with Existing Numerical Methods.}
The numerical computation of privacy amplification in the shuffle model has been studied since the introduction of the clone paradigm by Feldman et al.~\cite{Feldman2021}. Prior numerical algorithms can only handle three-point decompositions, such as computing \( D_{e^\varepsilon}(P_0^{C} \| P_1^{C}) \)~\cite{Feldman2021,tv_shuffle}. However, these techniques do not generalize to decompositions involving more than three points, which are essential for obtaining tight bounds from optimal decompositions.

In contrast, our FFT-based algorithm supports the optimal decomposition of \emph{any} local randomizer, enabling tighter and more accurate privacy amplification bounds. Furthermore, our method is not only more general but also simpler and faster than existing numerical algorithms (see Section~\ref{sec_experiment}).

\subsection{Algorithm for computing privacy amplification lower bounds}\label{sec_lowerbound}
An upper bound refers to the existence of a value \( \delta_u \) such that, for a given local randomizer with specified \( \varepsilon_0, \varepsilon, n \), and for any two neighboring input datasets \( X^0 \) and \( X^1 \), we have:
\[
D_{e^\varepsilon}\big( \mathcal{M}_S(X^0), \mathcal{M}_S(X^1) \big) \leq \delta_u.
\]
For a given local randomizer with specified \( \varepsilon_0, \varepsilon, n \), we can construct two neighboring datasets \( X^0 \) and \( X^1 \), and compute:
\[
\delta_l \le D_{e^\varepsilon}\big( \mathcal{M}_S(X^0), \mathcal{M}_S(X^1) \big),
\]
which serves as the lower bound for this amplification.

In previous studies, a common strategy for selecting neighboring datasets \( X^0 \) and \( X^1 \) is to set \( x_2 = x_3 = \dots = x_n \), such that \( x_1^0 \), \( x_1^1 \), and \( x_2 \) are mutually distinct~\cite{feldman2023soda,tv_shuffle,Ghazi20}. 
In the special case where \( |\mathbb{X}| = 2 \), the datasets \( X^0 \) and \( X^1 \) can be chosen as
\[
x_1^0 \ne x_1^1, \quad \text{and} \quad x_2 = x_3 = \dots = x_n = x_1^1.
\]
For the Laplace mechanism on \( \mathbb{X} = \{0,1\} \), we choose
\[
x_1^0 = 0, \quad x_1^1 = 1, \quad \text{and} \quad x_2 = x_3 = \dots = x_n = 1.
\]

The following theorem provides an efficient method for computing the hockey-stick divergence
\(
D_{e^\varepsilon}\big( \mathcal{M}_S(X^0), \mathcal{M}_S(X^1) \big)
\)
under this setting. It appeared in \cite{Ghazi20}.

\begin{theorem}\label{theorem_lowerbound}
Let \( \mathcal{R} : \mathbb{X} \to \mathbb{Y} \) be a local randomizer, and let \( \mathcal{M}_S = \mathcal{S} \circ \mathcal{R}^n \) be the shuffling of \( \mathcal{R} \). Fix \( \varepsilon \geq 0 \) and inputs \( X^0 \simeq X^1 \) with \( x_1^0 \neq x_1^1 \) and \( x_2 = x_3 = \dots = x_n \).
Define a random variable \( G = \frac{\mathcal{R}(x_1^0)(y) - e^\varepsilon \mathcal{R}(x_1^1)(y)}{\mathcal{R}(x_2)(y)} \), where \( y \sim \mathcal{R}(x_2) \).
Suppose \( G_1, G_2, \dots \) are i.i.d. copies of \( G \). Then, we have the following:
\[
D_{e^\varepsilon} (\mathcal{M}_S(X^0) \| \mathcal{M}_S(X^1)) = \frac{1}{n} \mathbb{E} \left[ \sum_{i=1}^{n} G_i \right]_{+}.
\]
\end{theorem}

\begin{proof}
It is a direct result from Theorem \ref{main_theorem} by setting $A_0=\mathcal{R}(x_1^0),A_1=\mathcal{R}(x_1^1)$ and $A_2=\mathcal{R}(x_2)$.
\end{proof}
Therefore, we can compute the lower bound using an algorithm similar to Algorithm~\ref{alg:fft}, with the only difference being that rounding up is replaced by rounding down. This modification ensures the correctness of the resulting bound. For convenience, we refer to the random variable \( G \) in Theorem~14 as the \emph{lower-bound GPARV}.

Theorem~\ref{theorem_lowerbound} also provides a method for computing the lower bound in the cases of joint composition and parallel composition. The procedure is analogous to that used for computing the corresponding upper bound: one first computes the decomposition of each sub-local randomizer at inputs \( x_1^0, x_1^1, x_2 \). Since \( x_2 = x_3 = \cdots = x_n \), the resulting decomposition does not involve any left-over distribution. Each component of the decomposition is defined over the set \( \{y \mid \mathcal{R}_i(x_1^0) : \mathcal{R}_i(x_1^1) : \mathcal{R}_i(x_2) = a : b : c \} \). 

By applying either the Cartesian product (for joint composition) or a weighted combination (for parallel composition) of these components, one obtains the decomposition corresponding to the composed local randomizer. This decomposition can then be used to compute the lower-bound GPARV.

It is important to note that, for joint composition, only the decompositions of the sub-local randomizers satisfy the compositional structure—this does not necessarily extend to their lower-bound GPARV. In contrast, for parallel composition, the lower-bound GPARV can be obtained as a weighted mixture of the individual lower-bound GPARVs for the sub-local randomizers.

\paragraph{Comparison with existing lower bound methods:} While \cite{Ghazi20} provided Theorem \ref{theorem_lowerbound}, they explicitly note that their computational method is applicable only when the induced random variable $G$ has support of size three.
For example, consider $k$-RR ($k\ge 4$) with datasets $X = (1,3,\dots,3)$ and $X' = (2,3,\dots,3)$. The random variable
$$
\frac{\mathcal{R}(x_1^0)(y) - e^{\varepsilon} \mathcal{R}(x_1^1)(y)}{\mathcal{R}(x_2)(y)},\quad y\sim \mathcal{R}(x_2)
$$
has support over four values $\{e^{\varepsilon_0}-e^\varepsilon,1-e^{\varepsilon_0+\varepsilon},1-e^{\varepsilon},\frac{1-e^\varepsilon}{e^{\varepsilon_0}} \}$. To avoid this, \cite{Ghazi20} assumes inputs of the form $X = (1,C,\dots,C)$ and $X' = (2,C,\dots,C)$, with $C$ uniformly distributed over $[k]$. This ensures the above expression has support over only three values: $\{e^{\varepsilon_0}-e^\varepsilon,\ 1-e^{\varepsilon_0+\varepsilon},\ 1-e^{\varepsilon} \}$.
Our algorithm, by contrast, imposes no such restrictions.

\section{Numerical Experiments}\label{sec_experiment}
In this section, we evaluate the performance of our FFT-based numerical algorithm for computing optimal privacy amplification bounds under the general clone framework.

\paragraph{Comparison with existing methods.}
We compare our computed bounds for several common local randomizers against existing bounds from prior work. The baselines include two bounds derived from the privacy blanket framework using Hoeffding's and Bennett's inequalities, respectively~\cite{Balle2019}, as well as numerical bounds from the standard clone paradigm~\cite{Feldman2021}. We use the publicly available implementations released by the respective authors. 

\paragraph{Experimental setup.}
We present results for five widely used local randomizers: \(k\)-ary randomized response (with \(k = 10\))~\cite{11}, Binary Local Hash (BLH)~\cite{3}, RAPPOR~\cite{6}, Optimized Unary Encoding (OUE)~\cite{3}, and the Laplace mechanism over \(\{0,1\}\)~\cite{Dwork2006}. For simplicity, we use the asymptotic forms of GPARV for BLH, RAPPOR, and OUE in the regime where the domain size \(D \gg 1\). As an illustrative example, we also report how the upper bound for RAPPOR changes with different values of \(D\). As shown in Fig.~\ref{varying_D}, when \(D \ge 5\), the upper bound curves become nearly indistinguishable. Furthermore, for \(D \ge 3\), where \(x_2 \in \mathbb{X} \setminus \{x_1^0, x_1^1\}\) can be selected, the lower-bound GPARV for the above local randomizers becomes independent of \(D\).

\paragraph{Results and observations.}
The main results are presented in Fig.~\ref{fig_exp1} and Fig.~\ref{fig_exp2}. Lower bounds are computed as described in Section~\ref{sec_lowerbound}. In the legends, ``\(k\)-joint'' denotes the joint composition of \(k\) local randomizers, each satisfying \(\frac{\varepsilon_0}{k}\)-LDP. The discretization interval \(l\) used in our FFT-based algorithm is set to \(\frac{e^{\varepsilon_0} - 1}{1200}\) when \(\varepsilon_0 = 0.1\), and \(\frac{e^{\varepsilon_0} - 1}{1000}\) when \(\varepsilon_0 = 4.0\).

Our results show that the computed upper bounds consistently outperform all baselines. Moreover, the gap between our upper and lower bounds is generally small, confirming the tightness and reliability of our analysis. These results also validate that the chosen discretization step \(l\) provides sufficient numerical precision in practice.
In addition to accuracy, our algorithm is highly efficient: a full amplification curve is generated in approximately 30 seconds, compared to roughly 5 minutes required by the standard clone’s numerical method.

\paragraph{Impact of joint composition.}
Our findings further highlight the advantage of computing specific bounds for joint compositions. The resulting amplification is significantly tighter than those derived from generic methods, illustrating the strength of our approach in multi-attribute settings. Moreover, under a fixed total privacy budget \(\varepsilon_0\), we observe that the amplification effect improves as the number of composed randomizers \(k\) increases.

\paragraph{Impact of parallel composition.}
We evaluate the parallel composition of 10-ary randomized response (10-RR) and BLH, where each mechanism is selected with equal probability \( \boldsymbol{p} = [0.5, 0.5] \). As shown in Fig.~\ref{fig_exp3}, the amplification curve of the parallel composition lies between those of the individual mechanisms, aligning well with theoretical expectations.

\paragraph{Poisson subsampling in the shuffle model.}
As an illustrative case, Fig.~\ref{fig:poisson} presents the upper bound computed by our method when applied to Poisson subsampling in the shuffle model, instantiated with 10-RR and analyzed via parallel composition. Notably, the corresponding lower bound closely matches the upper bound, demonstrating the tightness of our analysis in this setting as well.

\paragraph{Practical Application.}  
We validate the improvement brought by our bound on estimation error using the frequency estimation task as an example. Taking 10-RR with \( n = 1000 \) and \( \delta = 10^{-6} \), we compute, for a target privacy level \( \varepsilon \), the value of \( \varepsilon_0 \) required by our bound and by the previously best-known bound (Blanket with Bennett).  
In the simulation, user data follow a Zipf distribution with parameter \( 0.7 \). Repeating the experiment multiple times, we measure the \( \ell_2 \)-error between the estimated and true frequency distributions\footnote{For illustration, we adopt the classical unbiased estimation without clipping to $[0,1]$.}. The results are summarized in Table \ref{table20}.  
As shown, the \( \varepsilon_0 \) suggested by our bound is roughly 10\% larger than that given by prior work, leading to a substantial reduction in estimation error.

\section{Discussion: Beyond the General Clone}\label{sec_discussion}

In this work, we develop an efficient algorithm to compute the best-known privacy amplification bounds within the \emph{general clone} framework, which encompasses all possible decompositions. A natural and important question arises: \emph{Can we achieve tighter bounds than those provided by the general clone?} Addressing this question requires stepping beyond decomposition methods.

A promising direction is to identify the \emph{most vulnerable neighboring dataset pair} \( (X^0, X^1) \) such that 
\(
D_{e^\varepsilon}\left(\mathcal{M}_S(X^0)\parallel \mathcal{M}_S(X^1)\right)
\)
is maximized among all neighboring pairs of size \( n \). If, for a local randomizer \( \mathcal{R} \), one can prove that a specific pair \( (X^0_v, X^1_v) \) consistently maximizes the divergence for every \( \varepsilon \), then this would yield the exact privacy amplification bound for \( \mathcal{R} \) under shuffling.

For many local randomizers, a plausible candidate for the most vulnerable dataset pair is
\[
X^0 = (x_1^0, x_2, x_2, \dots, x_2), \quad X^1 = (x_1^1, x_2, x_2, \dots, x_2),
\]
where \( x_1^0, x_1^1, x_2 \in \mathbb{X} \) are mutually distinct. This construction is also used in Section~\ref{sec_lowerbound} for computing lower bounds of privacy amplification.

This conjecture is motivated by two observations. First, to maximize distinguishability, the set of inputs \( \{x_i \mid i = 2, 3, \dots, n\} \) should exclude both \( x_1^0 \) and \( x_1^1 \), ensuring that the outputs are not easily confounded. Second, having unified inputs among the remaining users simplifies the inference of their output contributions, thereby potentially increasing the overall distinguishability between the shuffled outputs of \( X^0 \) and \( X^1 \).

Despite its intuitive appeal, this conjecture currently lacks a formal proof. Developing tools to rigorously establish the most vulnerable neighboring pair remains an open problem and a valuable direction for future research.

\section{Limitation} \label{sec_limitation}
When a local randomizer satisfies \textit{input symmetry}, it is sufficient to evaluate \( G_\varepsilon^{x,x'} \) for any fixed neighboring pair \( (x, x') \); the resulting privacy bound then applies to all neighboring datasets. However, in the absence of input symmetry, one must, in principle, evaluate \( G_\varepsilon^{x,x'} \) for all possible input pairs and take the \textit{maximum} to ensure worst-case validity.

For the Laplace mechanism over \([0,1]\), we currently do not know how to prove or disprove whether the pair \( (x = 0, x' = 1) \) constitutes the worst case under the blanket decomposition. We leave this question as future work.

In the case of joint composition, a subtle issue arises regarding input symmetry. Even if each sub-local randomizer individually satisfies input symmetry, their joint composition introduces additional complications. Specifically, consider two neighboring inputs \( (x^1, x^2, \dots, x^m) \) and \( (\bar{x}^1, \bar{x}^2, \dots, \bar{x}^m) \) such that \( x^i \neq \bar{x}^i \) for all \( i \). In this case, the pair is symmetric, and it suffices to evaluate the GPARV on a single such instance.

However, when the adjacency relation is defined over the full product domain \( \mathbb{X}^1 \times \mathbb{X}^2 \times \dots \times \mathbb{X}^m \), one must also consider neighboring pairs that differ in only a subset of coordinates. If all \( \mathcal{R}_i \) are identical, we need to separately consider adjacency relations with Hamming distances \( 1, 2, \dots, m \).
For example, in the joint composition of three \( \varepsilon_0/4 \)-DP \( k \)-RR mechanisms, we must evaluate the GPARV for the following input pairs:
\[
\big( (0,0,0,0), (1,1,1,1) \big),\quad \big( (0,0,0,0), (0,1,1,1) \big),\quad \big( (0,0,0,0), (0,0,1,1) \big),\quad \big( (0,0,0,0), (0,0,0,1) \big).
\]
Among them, the pair \( \big( (0,0,0,0), (1,1,1,1) \big) \) is theoretically expected to yield the GPARV with the highest variance, and our experimental results confirm that it produces the largest upper bound (see Fig.~\ref{fig_joint_varying_d}). Whether this observation can be formally proven remains an open question for future work.

\section{Conclusion} \label{sec_conclusion}

In this work, we propose the general clone framework, which encompasses all decomposition
methods, and identify the optimal bounds within the general clone framework.
We also present an efficient algorithm for numerically computing these optimal bounds. With these results, we achieve the best-known bounds. Experiments demonstrate the tightness of our analysis. Additionally, we present methods for computing optimal amplification bounds for both joint composition and parallel composition in the shuffle model.
We hope that this work contributes to both the practical deployment and the theoretical advancement of the shuffle model in differential privacy.

\bibliographystyle{alpha}
\bibliography{shuffle}

\appendix
\section{Optimal Decompositions for Common Local Randomizers}\label{appenix_decomposition}

In this section, we derive the simplified optimal decomposition for several common local randomizers. Each of these randomizers satisfies the extreme property discussed in Section~\ref{sec_simplified_od}, and thus their simplified optimal decompositions consist of exactly five components, expressed as follows:
\begin{align*}
\mathcal{R}(x_1^0) &= e^{\varepsilon_0} p \mathcal{Q}_{e^{\varepsilon_0},1} + p \mathcal{Q}_{1,e^{\varepsilon_0}} + e^{\varepsilon_0} q \mathcal{Q}_{e^{\varepsilon_0},e^{\varepsilon_0}} + r \mathcal{Q}_{1,1},\\
\mathcal{R}(x_1^1) &= p \mathcal{Q}_{e^{\varepsilon_0},1} + e^{\varepsilon_0} p \mathcal{Q}_{1,e^{\varepsilon_0}} + e^{\varepsilon_0} q \mathcal{Q}_{e^{\varepsilon_0},e^{\varepsilon_0}} + r \mathcal{Q}_{1,1},\\
\forall i \in [2, n],\ \mathcal{R}(x_i) &= p \mathcal{Q}_{e^{\varepsilon_0},1} + p \mathcal{Q}_{1,e^{\varepsilon_0}} + q \mathcal{Q}_{e^{\varepsilon_0},e^{\varepsilon_0}} + r \mathcal{Q}_{1,1} + (1 - 2p - q - r) \mathcal{Q}_i.
\end{align*}
The values of the coefficients \( p, q, r \) for several widely-used local randomizers are summarized in Table~\ref{tab_decomposition} and \ref{tab_decomposition2}.
We provide the details for computing the corresponding coefficients.

\subsection{Binary local hash}
Let \( \mathcal{X} = [D] \). The Binary Local Hashing (BLH) mechanism is defined as \( \mathsf{BLH}(x) = (h, \mathsf{RR}(h(x))) \), where \( h: [D] \to [2] \) is a uniformly random function drawn from \( \mathcal{H} = \{ f : [D] \to [2] \} \), and \( \mathsf{RR} \) denotes the 2-ary randomized response mechanism (see Example~\ref{example1} for the definition) \cite{3}.

Define $h^*:[D]\to [2],\forall x,h^*(x)=1$ and $\bar{h}^*:[D]\to [2],\forall x,\bar{h}^*(x)=0$. When \( h \in \{ h^*, \bar{h}^*\} \), the output distribution of \( \mathsf{BLH}(x) \) is the same for all inputs \( x \). We have
\begin{align*}
\inf_{x} \Pr[\mathsf{BLH}(x) = (h^*,0)] =\frac{1}{|\mathcal{H}|(1 + e^{\varepsilon_0})},\quad\inf_{x} \Pr[\mathsf{BLH}(x) = (h^*,1)] =\frac{e^{\varepsilon_0}}{|\mathcal{H}|(1 + e^{\varepsilon_0})},\\
\inf_{x} \Pr[\mathsf{BLH}(x) = (\bar{h}^*,0)] =\frac{e^{\varepsilon_0}}{|\mathcal{H}|(1 + e^{\varepsilon_0})},\quad\inf_{x} \Pr[\mathsf{BLH}(x) = (\bar{h}^*,1)] =\frac{1}{|\mathcal{H}|(1 + e^{\varepsilon_0})}.
\end{align*}

When $h\notin \{h^*,\bar{h}^*\}$, there must exist two inputs \( a, b \in [D] \) such that \( h(a) = 0 \) and \( h(b) = 1 \). In this case, we have
\[
\inf_{x} \Pr[\mathsf{BLH}(x) = (h,0)] = \inf_{x} \Pr[\mathsf{BLH}(x) = (h,1)] = \frac{1}{|\mathcal{H}|(1 + e^{\varepsilon_0})}.
\]
In summary, the blanket distribution of BLH is ``nearly'' uniform distribution over \( \mathcal{H} \times [2] \):
\begin{align*}
\omega(h,y)=\inf_x \Pr[\mathsf{BLH}(x)=(h,y)]/\gamma=\begin{cases}
\frac{1}{\gamma(1+e^{\varepsilon_0})},& (h,y)\notin \{(h^*,1),(\bar{h}^*,0)\},\\
\frac{e^{\varepsilon_0}}{\gamma(1+e^{\varepsilon_0})},& (h,y)\in \{(h^*,1),(\bar{h}^*,0)\}.
\end{cases}
\end{align*}
where $\gamma=\sum_{(h,y)}\inf_x \Pr[\mathsf{BLH}(x)=(h,y)]=\frac{2}{1+e^{\varepsilon_0}}+\frac{e^{\varepsilon_0}-1}{2^{D-1}(1+e^{\varepsilon_0})}$.

Recall the definition in Section \ref{sec_simplified_od}: $$S_{r_0,r_1}=\{y| \frac{\mathcal{R}(x_1^0)(y)}{\inf_{x} \mathcal{R}(x)(y)}=r_0, \frac{\mathcal{R}(x_1^1)(y)}{\inf_{x} \mathcal{R}(x)(y)}=r_1\}.$$
It can be computed that
\begin{align*}
S_{e^{\varepsilon_0},1}&=\{(h,y)|h(x_1^0)=y\wedge h(x_1^1)\neq y \},\\
S_{1,e^{\varepsilon_0}}&=\{(h,y)|h(x_1^0)\neq y\wedge h(x_1^1)= y \},\\
S_{e^{\varepsilon_0},e^{\varepsilon_0}}&=\{(h,y)|h(x_1^0)=y\wedge h(x_1^1)= y \}- \{(h,y)|h=h^* \vee h=\bar{h}^* \},\\
S_{1,1}&=\{(h,y)|h(x_1^0)\neq y\wedge h(x_1^1)\neq y \} \cup \{(h,y)|h=h^* \vee h=\bar{h}^* \}.
\end{align*}
Consequently, it follows that:
\begin{alignat*}{3}
p &= \gamma \omega(S_{e^{\varepsilon_0},1}) 
  &\quad =\quad& \frac{1}{2(e^{\varepsilon_0}+1)}, \\[0.5ex]
q &= \gamma \omega(S_{e^{\varepsilon_0},e^{\varepsilon_0}}) 
  &\quad =\quad& \frac{1}{2(e^{\varepsilon_0}+1)} - \frac{2}{|\mathcal{H}|(1+e^{\varepsilon_0})} 
  &\quad =\quad& \frac{1}{2(e^{\varepsilon_0}+1)} - \frac{1}{2^{D-1}(1+e^{\varepsilon_0})}, \\[0.5ex]
r &= \gamma \omega(S_{1,1}) 
  &\quad =\quad& \frac{1}{2(e^{\varepsilon_0}+1)} + \frac{2e^{\varepsilon_0}}{|\mathcal{H}|(1+e^{\varepsilon_0})} 
  &\quad =\quad& \frac{1}{2(e^{\varepsilon_0}+1)} + \frac{e^{\varepsilon_0}}{2^{D-1}(1+e^{\varepsilon_0})}.
\end{alignat*}

\begin{table}
    \centering
    \caption{Optimal decompositions for common local randomizers}
 \small  {($D$ denotes the size of the domain.)}
    \begin{tabular}{cccc}
    \toprule
         & $p$  & $q$ & $r$\\
         \midrule
         $k$-RR \cite{11}& $\frac{1}{e^{\varepsilon_0}+k-1}$& $0$ & $\frac{k-2}{e^{\varepsilon_0}+k-1}$ \\
        BLH~\cite{3} & $\frac{1}{2(e^{\varepsilon_0}+1)}$ & $\frac{1}{2(e^{\varepsilon_0}+1)} - \frac{1}{2^{D-1}(e^{\varepsilon_0}+1)}$ & $\frac{1}{2(e^{\varepsilon_0}+1)} + \frac{e^{\varepsilon_0}}{2^{D-1}(e^{\varepsilon_0}+1)}$\\
        RAPPOR~\cite{6} & $\frac{1}{(e^{\varepsilon_0/2}+1)^2}$ &$\frac{e^{-\varepsilon_0/2}}{(e^{\varepsilon_0/2}+1)^2}-\frac{e^{-\varepsilon_0/2} }{(e^{\varepsilon_0/2}+1)^D}$& $\frac{e^{\varepsilon_0/2}}{(e^{\varepsilon_0/2}+1)^2}+\frac{e^{\varepsilon_0/2} }{(e^{\varepsilon_0/2}+1)^D}$  \\
        OUE~\cite{3}  & $\frac{1}{2(e^{\varepsilon_0}+1)}$ & $\frac{e^{-\varepsilon_0}}{2(e^{\varepsilon_0}+1)}-\frac{e^{-\varepsilon_0}}{2(1 + e^{\varepsilon_0})^{D-1}}$ & $\frac{e^{\varepsilon_0}}{2(e^{\varepsilon_0}+1)}+\frac{1}{2(e^{\varepsilon_0}+1)^{D-1}}$\\
        HR~\cite{HR} & $\frac{1}{2(e^{\varepsilon_0}+1)}$  & $\frac{1}{2(e^{\varepsilon_0}+1)}$ & $\frac{1}{2(e^{\varepsilon_0}+1)}+\frac{2(e^{\varepsilon_0}-1)}{D(e^{\varepsilon_0}+1)}$\\

         \bottomrule
    \end{tabular}

    \label{tab_decomposition}
\end{table}

\begin{table}
    \centering
    \caption{Optimal decompositions for common local randomizers (Asymptotic,$D\gg 1$)}
    \begin{tabular}{cccc}
    \toprule
         & $p$  & $q$ & $r$\\
         \midrule
        BLH ~\cite{3} & $\frac{1}{2(e^{\varepsilon_0}+1)}$ & $\frac{1}{2(e^{\varepsilon_0}+1)}$ & $\frac{1}{2(e^{\varepsilon_0}+1)}$\\
        RAPPOR ~\cite{6} & $\frac{1}{(e^{\varepsilon_0/2}+1)^2}$ &$\frac{1}{e^{\varepsilon_0/2}(e^{\varepsilon_0/2}+1)^2}$& $\frac{e^{\varepsilon_0/2}}{(e^{\varepsilon_0/2}+1)^2}$  \\
        OUE ~\cite{3}  & $\frac{1}{2(e^{\varepsilon_0}+1)}$ & $\frac{1}{2e^{\varepsilon_0}(e^{\varepsilon_0}+1)}$ & $\frac{e^{\varepsilon_0}}{2(e^{\varepsilon_0}+1)}$\\
        HR ~\cite{HR} & $\frac{1}{2(e^{\varepsilon_0}+1)}$  & $\frac{1}{2(e^{\varepsilon_0}+1)}$ & $\frac{1}{2(e^{\varepsilon_0}+1)}$\\

         \bottomrule
    \end{tabular}
    \label{tab_decomposition2}
\end{table}

\subsection{RAPPOR}
RAPPOR first encodes the input \( x \in [D] \) using unary encoding, yielding \( \mathrm{UE}(x) \in \{0,1\}^D \), where \( \forall i \in [D],\ \mathrm{UE}(x)[i] = \mathbbm{1}[x = i] \). It then applies \( \frac{\varepsilon_0}{2} \)-DP 2-ary randomized response independently to each bit of \( \mathrm{UE}(x) \). Specifically, the output distribution of RAPPOR satisfies:
\[
\Pr[\mathrm{RAPPOR}(x) = \boldsymbol{y}] = \frac{e^{(D - d_H(\mathrm{UE}(x), \boldsymbol{y})) \cdot \varepsilon_0/2 }}{(1 + e^{\varepsilon_0/2})^D},
\]
where \( d_H(\mathrm{UE}(x), \boldsymbol{y}) \) denotes the Hamming distance between two binary vectors in \( \{0,1\}^D \).

As in the case of BLH, RAPPOR also exhibits an extreme case, namely \( \boldsymbol{y}^* = (1,1,\dots,1) \). We have:
\begin{align*}
\omega(\boldsymbol{y})=\inf_x \Pr[\mathrm{RAPPOR}(x) = \boldsymbol{y}]/\gamma = \begin{cases}
\frac{e^{(D -1- d_H(\boldsymbol{y},\boldsymbol{0})) \cdot \varepsilon_0/2 }}{\gamma(1 + e^{\varepsilon_0/2})^D},& \boldsymbol{y}\neq \boldsymbol{y}^*,\\
\frac{e^{\varepsilon_0/2} }{\gamma(1 + e^{\varepsilon_0/2})^D},&\boldsymbol{y}=\boldsymbol{y}^*.
\end{cases}
\end{align*}
where
\begin{align*}
\gamma=\sum_{\boldsymbol{y}}\inf_x \Pr[\mathrm{RAPPOR}(x)=\boldsymbol{y}]&= \sum_{\boldsymbol{y}\in \{0,1\}^D}\frac{e^{(D -1- d_H(\boldsymbol{y},\boldsymbol{0})) \cdot \varepsilon_0/2 }}{(1 + e^{\varepsilon_0/2})^D} +\frac{e^{\varepsilon_0/2}-e^{-\varepsilon_0/2} }{(1 + e^{\varepsilon_0/2})^D}\\
&=\frac{1}{(1 + e^{\varepsilon_0/2})^D}\sum_{i=0}^D \tbinom{D}{i} e^{(i-1)\cdot \varepsilon_0/2}+\frac{e^{\varepsilon_0/2}-e^{-\varepsilon_0/2} }{(1 + e^{\varepsilon_0/2})^D}\\
&=e^{-\varepsilon_0/2}+\frac{e^{\varepsilon_0/2}-e^{-\varepsilon_0/2} }{(1 + e^{\varepsilon_0/2})^D}.
\end{align*}

According to the definition in Section \ref{sec_simplified_od}:
\begin{align*}
S_{e^{\varepsilon_0},1}&=\{\boldsymbol{y}|\boldsymbol{y}[x_1^0]=1 \wedge \boldsymbol{y}[x_1^1]=0 \},\\
S_{1,e^{\varepsilon_0}}&=\{\boldsymbol{y}|\boldsymbol{y}[x_1^0]=0 \wedge \boldsymbol{y}[x_1^1]=1 \},\\
S_{e^{\varepsilon_0},e^{\varepsilon_0}}&=\{\boldsymbol{y}|\boldsymbol{y}[x_1^0]=1 \wedge \boldsymbol{y}[x_1^1]=1 \}- \{\boldsymbol{y}^*\},\\
S_{1,1}&=\{\boldsymbol{y}|\boldsymbol{y}[x_1^0]=0 \wedge \boldsymbol{y}[x_1^1]=0 \}\cup \{\boldsymbol{y}^* \}.
\end{align*}
When $D>2$, it can be computed that
\begin{alignat*}{3}
p &= \gamma \omega(S_{e^{\varepsilon_0},1}) 
  &\quad =\quad& \frac{1}{(e^{\varepsilon_0/2}+1)^2},  \\[0.5ex]
q &= \gamma \omega(S_{e^{\varepsilon_0},e^{\varepsilon_0}}) 
  &\quad =\quad&  \frac{e^{-\varepsilon_0/2}}{(e^{\varepsilon_0/2}+1)^2} -\frac{e^{-\varepsilon_0/2} }{(1 + e^{\varepsilon_0/2})^D}, \\[0.5ex]
r &= \gamma \omega(S_{1,1}) 
  &\quad =\quad& \frac{e^{\varepsilon_0/2}}{(e^{\varepsilon_0/2}+1)^2} +\frac{e^{\varepsilon_0/2} }{(1 + e^{\varepsilon_0/2})^D}.
\end{alignat*}

\subsection{Optimal Unary Encode (OUE)}
OUE (Optimized Unary Encoding) follows the same initial step as RAPPOR by applying unary encoding to the input \( x \), resulting in \( \mathrm{UE}(x) \in \{0,1\}^D \). However, it differs in how randomized response is applied: for each bit in \( \mathrm{UE}(x) \), a 2-ary randomized response mechanism satisfying \( \varepsilon_0 \)-DP is applied to the bits with value 0, while the bits with value 1 are replaced with uniformly random bits \cite{3}.

The resulting probability distribution is given by:
\[
\Pr[\mathrm{OUE}(x) = \boldsymbol{y}] = \frac{e^{(D-1 - d_H(\boldsymbol{y}_{-x},\boldsymbol{0})) \cdot \varepsilon_0 }}{2(1 + e^{\varepsilon_0})^{D-1}},
\]
where \( d_H \) denotes the Hamming distance, and \( \boldsymbol{y}_{-x} \) denotes the vector \( \boldsymbol{y} \) with the \( x \)-th coordinate removed.

As in the case of RAPPOR, OUE also exhibits an extreme case, namely \( \boldsymbol{y}^* = (1,1,\dots,1) \). We have:
\begin{align*}
\omega(\boldsymbol{y})=\inf_x \Pr[\mathrm{OUE}(x) = \boldsymbol{y}]/\gamma = \begin{cases}
\frac{e^{(D-1 - d_H(\boldsymbol{y},\boldsymbol{0})) \cdot \varepsilon_0 }}{2\gamma(1 + e^{\varepsilon_0})^{D-1}},& \boldsymbol{y}\neq \boldsymbol{y}^*,\\
\frac{1}{2\gamma(1 + e^{\varepsilon_0})^{D-1}},&\boldsymbol{y}=\boldsymbol{y}^*.
\end{cases}
\end{align*}
where
\begin{align*}
\gamma=\sum_{\boldsymbol{y}}\inf_x \Pr[\mathrm{OUE}(x)=\boldsymbol{y}]&=\sum_{\boldsymbol{y}\in\{0,1\}^D} \frac{e^{(D-1 - d_H(\boldsymbol{y},\boldsymbol{0})) \cdot \varepsilon_0 }}{2(1 + e^{\varepsilon_0})^{D-1}} +\frac{1-e^{-\varepsilon_0}}{2(1 + e^{\varepsilon_0})^{D-1}}\\
&=\frac{1}{2(1 + e^{\varepsilon_0})^{D-1}}\sum_{i=0}^D \tbinom{D}{i} e^{(i-1)\cdot \varepsilon_0}+\frac{1-e^{-\varepsilon_0}}{2(1 + e^{\varepsilon_0})^{D-1}}\\
&=\frac{1+e^{\varepsilon_0}}{2e^{\varepsilon_0}}+\frac{1-e^{-\varepsilon_0}}{2(1 + e^{\varepsilon_0})^{D-1}}.
\end{align*}

According to the definition in Section \ref{sec_simplified_od}:
\begin{align*}
S_{e^{\varepsilon_0},1}&=\{\boldsymbol{y}|\boldsymbol{y}[x_1^0]=1 \wedge \boldsymbol{y}[x_1^1]=0 \},\\
S_{1,e^{\varepsilon_0}}&=\{\boldsymbol{y}|\boldsymbol{y}[x_1^0]=0 \wedge \boldsymbol{y}[x_1^1]=1 \},\\
S_{e^{\varepsilon_0},e^{\varepsilon_0}}&=\{\boldsymbol{y}|\boldsymbol{y}[x_1^0]=1 \wedge \boldsymbol{y}[x_1^1]=1 \}- \{\boldsymbol{y}^*\},\\
S_{1,1}&=\{\boldsymbol{y}|\boldsymbol{y}[x_1^0]=0 \wedge \boldsymbol{y}[x_1^1]=0 \}\cup \{\boldsymbol{y}^* \}.
\end{align*}
When $D>2$, it can be computed that
\begin{alignat*}{3}
p &= \gamma \omega(S_{e^{\varepsilon_0},1}) 
  &\quad =\quad& \frac{1}{2(1+e^{\varepsilon_0})},  \\[0.5ex]
q &= \gamma \omega(S_{e^{\varepsilon_0},e^{\varepsilon_0}}) 
  &\quad =\quad&  \frac{e^{-\varepsilon_0}}{2(1+e^{\varepsilon_0})}-\frac{e^{-\varepsilon_0}}{2(1 + e^{\varepsilon_0})^{D-1}}, \\[0.5ex]
r &= \gamma \omega(S_{1,1}) 
  &\quad =\quad& \frac{e^{\varepsilon_0}}{2(e^{\varepsilon_0}+1)}+\frac{1}{2(1 + e^{\varepsilon_0})^{D-1}}.
\end{alignat*}

\subsection{Hadamard Response}
The Hadamard Response~\cite{HR} assumes an input domain \( \mathbb{X} = [D] \),
where \( D = 2^r \) is a power of two\footnote{Because of properties of the Hadamard matrix, \(0\) is treated as a special symbol and is excluded from the input domain.}. The mechanism maps each input \( x \in [D] \) to a row of the normalized Hadamard matrix \( H \in \{-1,1\}^{D \times D} \), and samples an output index \( y \in [D] \) with probability proportional to \( \exp\left( \frac{\varepsilon_0}{2} H_{x,y} \right) \), where \( \varepsilon_0 \) is the privacy budget. Each entry of the Hadamard matrix is defined as \( H_{x,y} = (-1)^{x \odot y} \), where \( x \odot y \) denotes the bitwise inner product modulo 2 between the binary representations of \( x \) and \( y \); that is, the sum (mod 2) of the component-wise products of their bits.

Formally, the Hadamard Response mechanism is defined as:
\[
\Pr[\mathrm{HR}(x) = y] = \frac{1}{Z} \exp\left( \frac{\varepsilon_0}{2} H_{x,y} \right),
\]
where \( H \in \{-1,1\}^{D \times D} \) is the Hadamard matrix, and \( Z \) is the normalization constant ensuring that the distribution sums to one.

Similarly, HR also exhibits an extreme case:
\begin{align*}
\inf_{x\in \mathbb{X}}\Pr[\mathcal{R}(x)=0]&=\frac{2 e^{\varepsilon_0}}{(1+e^{\varepsilon_0})D},\\
\forall y\neq 0,\inf_{x\in \mathbb{X}}\Pr[\mathcal{R}(x)=y]&=\frac{2}{(1+e^{\varepsilon_0})D}.
\end{align*}

By direct computation from the definition, its corresponding parameter is \( \gamma = \frac{2}{1 + e^{\varepsilon_0}}+\frac{2 (e^{\varepsilon_0}-1)}{(1+e^{\varepsilon_0})D} \).

According to the definition in Section \ref{sec_simplified_od}:
\begin{align*}
S_{e^{\varepsilon_0},1}&=\{y|x_1^0\odot y=1 \wedge x_1^1\odot y=0 \},\\
S_{1,e^{\varepsilon_0}}&=\{y|x_1^0\odot y=0 \wedge x_1^1\odot y=1 \},\\
S_{e^{\varepsilon_0},e^{\varepsilon_0}}&=\{y|x_1^0\odot y=1 \wedge x_1^1\odot y=1\},\\
S_{1,1}&=\{y|x_1^0\odot y=0 \wedge x_1^1\odot y=0 \}.
\end{align*}
Due to the properties of the Hadamard matrix, we have $|S_{e^{\varepsilon_0},1}|=|S_{1,e^{\varepsilon_0}}|=|S_{e^{\varepsilon_0},e^{\varepsilon_0}}|=|S_{1,1}|$.
Therefore,
\begin{alignat*}{3}
p &= \gamma \omega(S_{e^{\varepsilon_0},1}) 
  &\quad =\quad& \frac{1}{2(1 + e^{\varepsilon_0})},  \\[0.5ex]
q &= \gamma \omega(S_{e^{\varepsilon_0}, e^{\varepsilon_0}}) 
  &\quad =\quad& \frac{1}{2(1 + e^{\varepsilon_0})}, \\[0.5ex]
r &= \gamma \omega(S_{1,1}) 
  &\quad =\quad& \frac{1}{2(1 + e^{\varepsilon_0})}+\frac{2 (e^{\varepsilon_0}-1)}{(1+e^{\varepsilon_0})D}.
\end{alignat*}

\section{GPARV for Laplace Mechanism} \label{appendix_laplace}

We consider the Laplace mechanism over the domain \(\mathbb{X} = \{0,1\} \), defined by
\[
\mathcal{R}(x) = x + \text{Lap}\left(\frac{1}{\varepsilon_0}\right),
\]
where \(\text{Lap}(\lambda)\) denotes the Laplace distribution with scale parameter \(\lambda\). The corresponding distributions \(\mathcal{R}(0)\), \(\mathcal{R}(1)\), and the blanket distribution \(\omega\) are given as follows:
\[
\mathcal{R}(0)(y) = \frac{\varepsilon_0}{2} \exp\left(-\varepsilon_0 |y|\right), \quad
\mathcal{R}(1)(y) = \frac{\varepsilon_0}{2} \exp\left(-\varepsilon_0 |y - 1|\right),
\]
\[
\omega(y) = \inf_{x \in [0,1]} \frac{\mathcal{R}(x)(y)}{\gamma} =
\begin{cases}
    \frac{\varepsilon_0}{2\gamma} \exp\left(\varepsilon_0 (y - 1)\right), & y \le 0.5, \\
    \frac{\varepsilon_0}{2\gamma} \exp\left(-\varepsilon_0 y\right), & y > 0.5,
\end{cases}
\]
where \( \gamma = \exp(-\varepsilon_0 / 2) \).

The cumulative distribution function (CDF) of the privacy amplification random variable \( L^{0,1}_\varepsilon \) is computed as:
\[
\Pr\left[L^{0,1}_\varepsilon \le t\right] =
\begin{cases}
    0, & t < \gamma(1 - e^{\varepsilon_0 + \varepsilon}), \\
    \frac{1}{2} \sqrt{\frac{e^\varepsilon}{1 - e^{\varepsilon_0/2} t}}, & \gamma(1 - e^{\varepsilon_0 + \varepsilon}) \le t < \gamma(1 - e^{\varepsilon}), \\
    1 - \frac{1}{2} (e^{\varepsilon_0/2} t + e^\varepsilon)^{-1/2}, & \gamma(1 - e^{\varepsilon}) \le t < \gamma(e^{\varepsilon_0} - e^\varepsilon), \\
    1, & t \ge \gamma(e^{\varepsilon_0} - e^\varepsilon).
\end{cases}
\]

The CDF of \( L^{0,1}_\varepsilon \) with parameters \( \varepsilon_0 = 1.0 \) and \( \varepsilon = 0.1 \) is shown in Fig.~\ref{fig_laplace_cdf}.

\section{Proof of Theorem \ref{main_theorem}}\label{proof_main_theorem}
\begin{proof}
The proof of this theorem follows similar lines as the proof of Lemma 12 in \cite{Balle2019}.

Let the support of \( A_2 \) be \( \mathbb{Y} \). Since \( A_0 \) and \( A_1 \) are absolutely continuous with respect to \( A_2 \), their supports are subsets of \( \mathbb{Y} \).
Define random variables $Y_1^b \sim A_b,b=0,1$ and $W_i \sim A_2,i=1,2,\dots,n-1$. $\mathbf{W}_{n-1}=\{ W_1,W_2,\dots,W_{n-1}\}$.
Let $\vec{y}\in \mathbb{Y}^n$ be a tuple of elements from $\mathbb{Y}$ and $Y \in \mathbb{N}_{n}^{\mathbb{Y}}$ be the corresponding multiset of entries. Then we have
\[
\mathbb{P}[\{ Y_1^b \} \cup \mathbf{W}_{n-1} = Y] = \frac{1}{n!} \sum_{\sigma} \mathbb{P}\left[ (Y_1^b,W_1, \ldots, W_{n-1}) = \vec{y}_{\sigma} \right],
\]
where $\sigma$ ranges over all permutations of $[n]$ and $\vec{y}_{\sigma} = (y_{\sigma(1)}, \ldots, y_{\sigma(n)})$. We also have
\[\mathbb{P}\left[ (Y_1^b,W_1, \ldots, W_{n-1}) = \vec{y}_{\sigma} \right] = A_b(y_{\sigma(1)})\prod_{i=2}^n A_2(y_{\sigma(i)}),\quad b=0,1.\]

Summing this expression over all permutations $\sigma$ and factoring out the product of the $A_2$'s yields:
\begin{align*}
\frac{1}{n!} \sum_{\sigma} A_b(y_{\sigma(1)}) A_2(y_{\sigma(2)}) \cdots A_2(y_{\sigma(n)})  &= \left( \prod_{i=1}^n A_2(y_i) \right) \left( \frac{1}{n} \sum_{i=1}^n \frac{A_b(y_i)}{A_2(y_i)} \right)\\
&= \mathbb{P}[\mathbf{W}_n = Y] \cdot \frac{1}{n} \sum_{i=1}^n \frac{A_b(y_i)}{A_2(y_i)}.
\end{align*}
Now we can plug these observation into the definition of $\mathcal{D}_{e^\varepsilon}$ and complete the proof as
\begin{align*}
D_{e^\varepsilon} \big(P_0 \parallel P_1\big)&=D_{e^\varepsilon} (\{ Y_1^0 \} \cup \mathbf{W}_{n-1} \parallel \{ Y_1^1 \} \cup \mathbf{W}_{n-1})\\
&= \int_{\mathbb{N}_n^{\mathbb{Y}}} \left[ \mathbb{P}[\{ Y_1^0 \} \cup \mathbf{W}_{n-1} = Y] - e^{\varepsilon} \mathbb{P}[\{ Y_1^1 \} \cup \mathbf{W}_{n-1} = Y] \right]_+\\
&= \mathbb{E} \left[ \frac{1}{n} \sum_{i=1}^n \frac{A_0(y_i) - e^{\varepsilon} A_1(y_i)}{A_2(y_i)} \right]_+\\
&= \mathbb{E} \left[ \frac{1}{n} \sum_{i=1}^n G_i \right]_+.
\end{align*}
\end{proof}

\begin{figure*}[ht]
    \centering
    
    \begin{subfigure}[b]{0.48\textwidth}
        \centering
        \includegraphics[width=\linewidth]{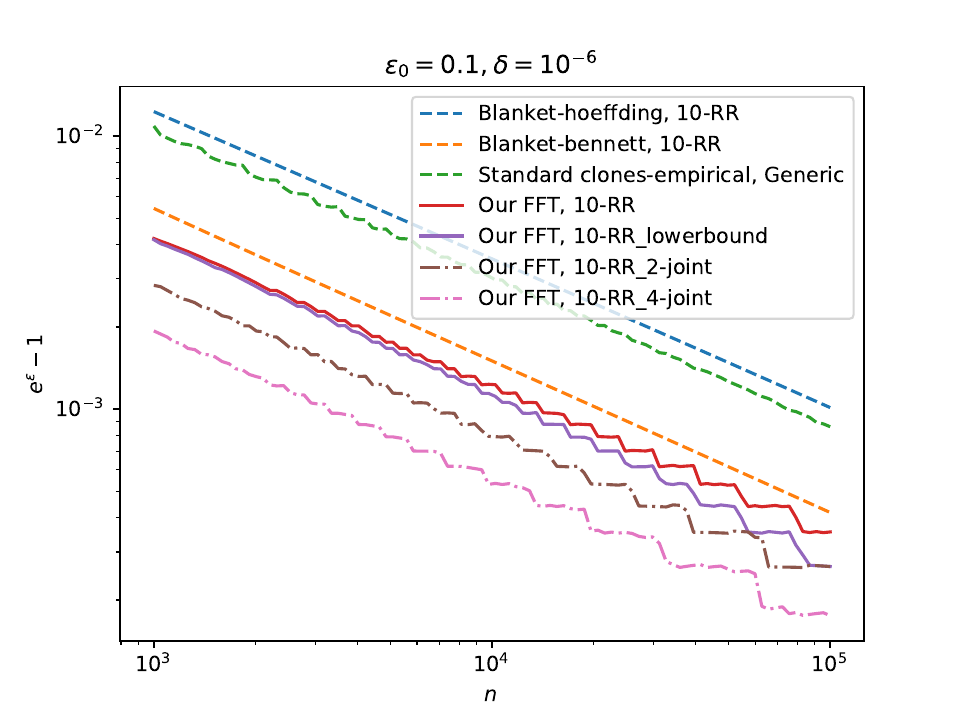} 
        \caption*{10-RR, low $\varepsilon_0$}
    \end{subfigure}
    \hfill
    \begin{subfigure}[b]{0.48\textwidth}
        \centering
        \includegraphics[width=\linewidth]{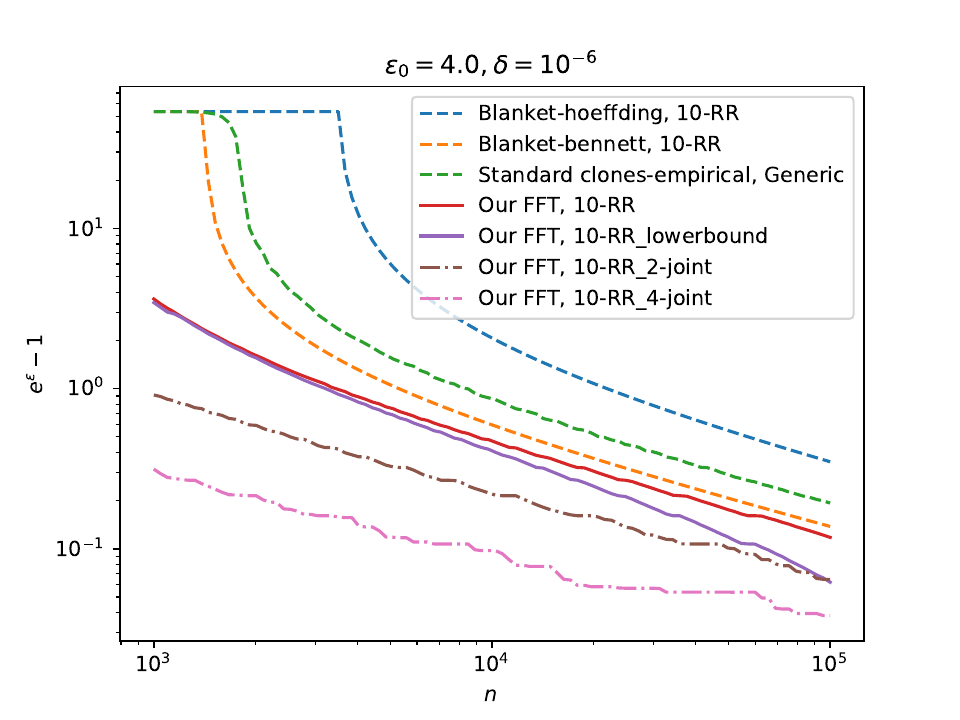} 
        \caption*{10-RR, large $\varepsilon_0$}
    \end{subfigure}
    
    \vspace{0.2cm} 
    
    \begin{subfigure}[b]{0.48\textwidth}
        \centering
        \includegraphics[width=\linewidth]{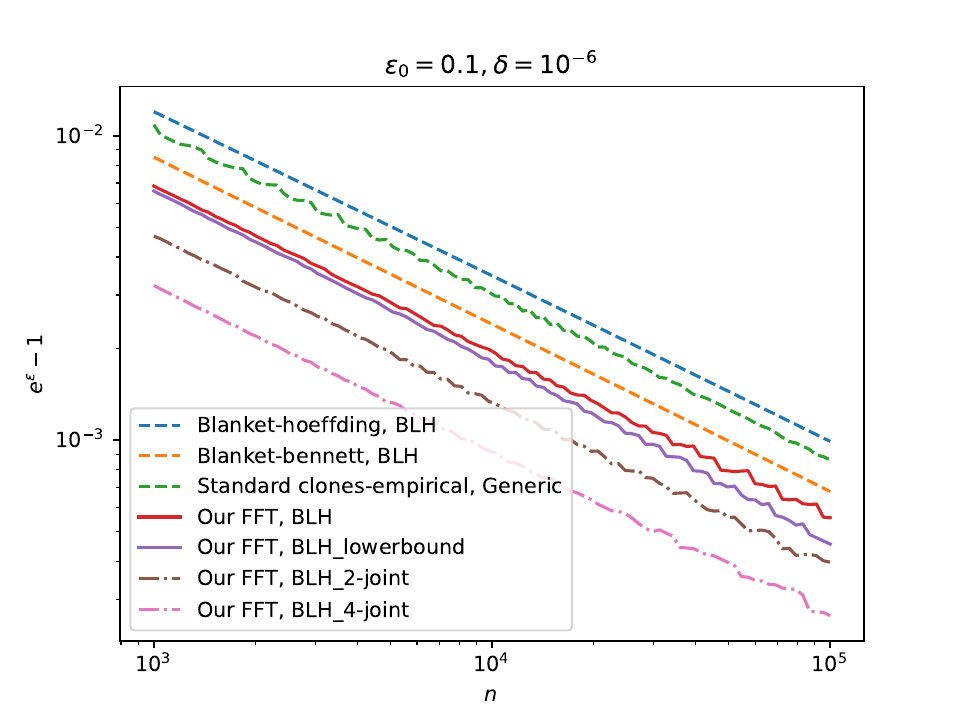} 
        \caption*{BLH, low $\varepsilon_0$}
        \label{fig:group3a}
    \end{subfigure}
    \hfill
    \begin{subfigure}[b]{0.48\textwidth}
        \centering
        \includegraphics[width=\linewidth]{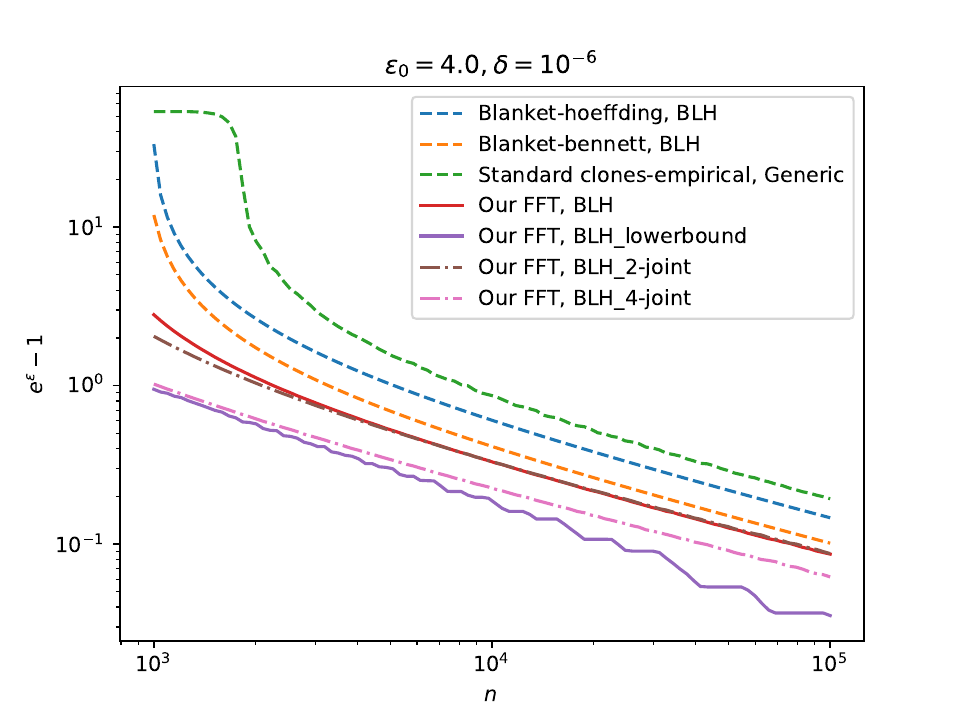} 
        \caption*{BLH, large $\varepsilon_0$}
        \label{fig:group3b}
    \end{subfigure}
    
    \vspace{0.2cm} 

    \begin{subfigure}[b]{0.48\textwidth}
    \centering
    \includegraphics[width=\linewidth]{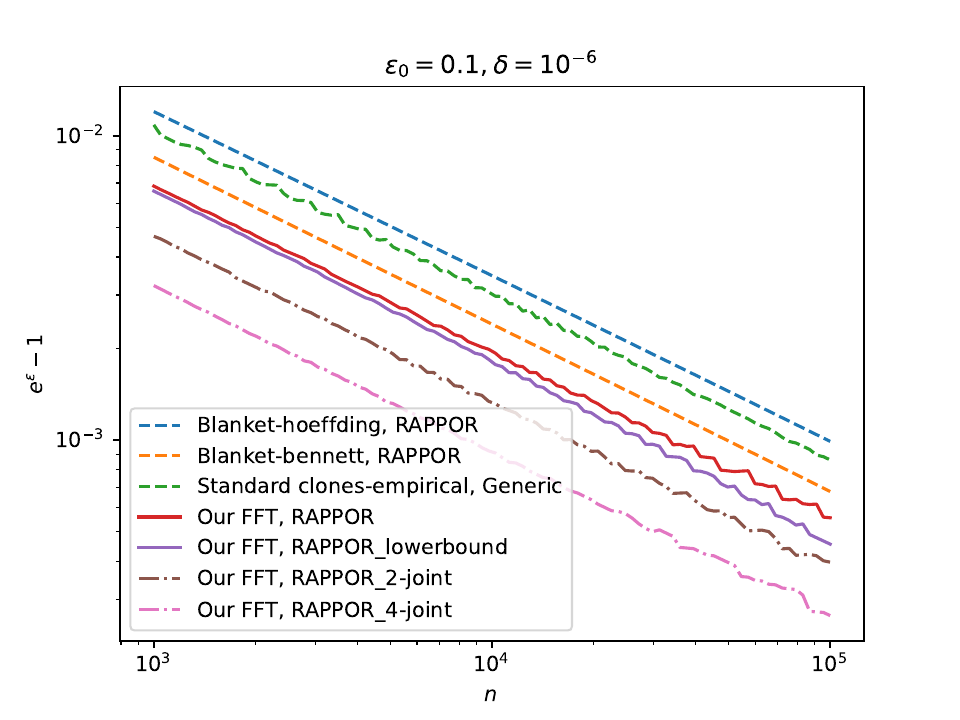} 
    \caption*{RAPPOR, low $\varepsilon_0$}
    \end{subfigure}
    \hfill
    \begin{subfigure}[b]{0.48\textwidth}
        \centering
        \includegraphics[width=\linewidth]{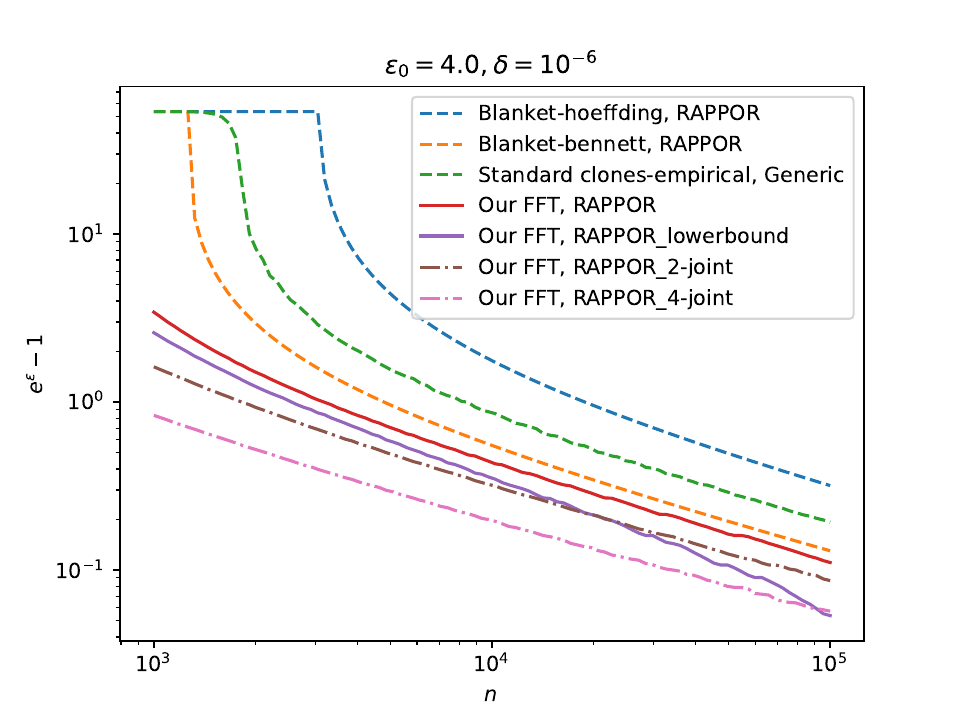} 
        \caption*{RAPPOR, large $\varepsilon_0$}
    \end{subfigure}
    
    \caption{Experimental Results: RR, BLH and RAPPOR}
    \label{fig_exp1}
\end{figure*}

\begin{figure*}[ht]
    \centering
    \begin{subfigure}[b]{0.48\textwidth}
        \centering
        \includegraphics[width=\linewidth]{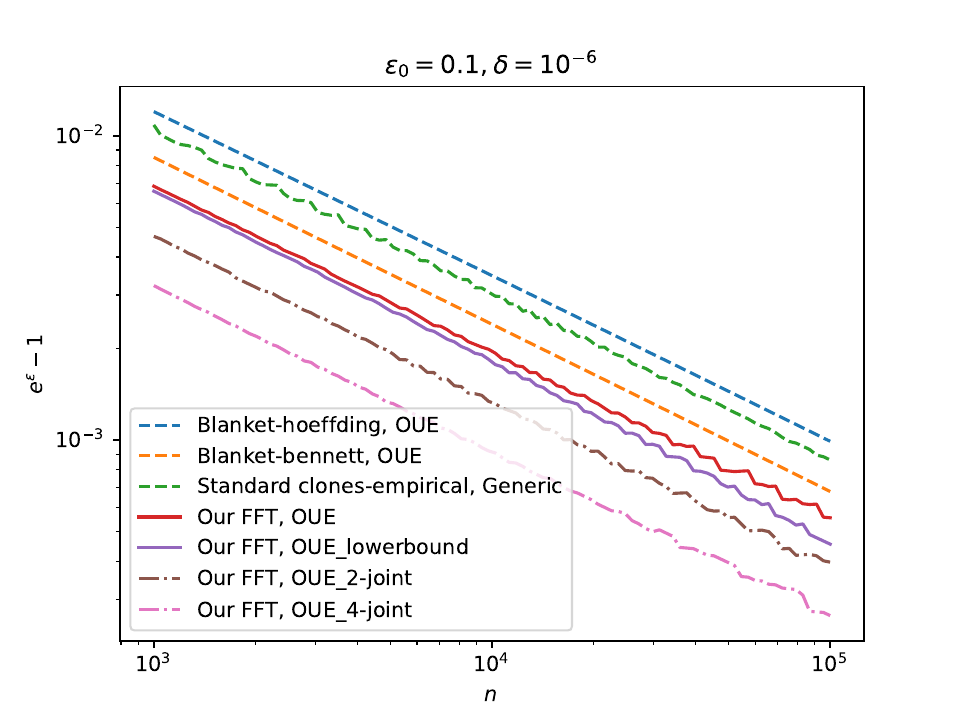} 
        \caption*{OUE, low $\varepsilon_0$}
    \end{subfigure}
    \hfill
    \begin{subfigure}[b]{0.48\textwidth}
        \centering
        \includegraphics[width=\linewidth]{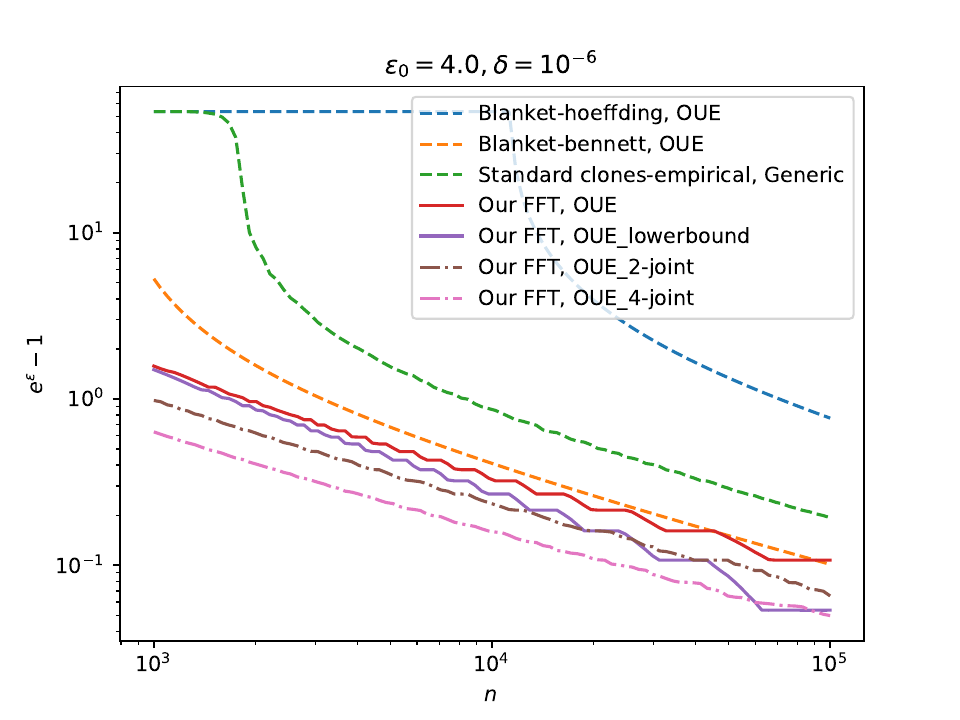} 
        \caption*{OUE, large $\varepsilon_0$}
    \end{subfigure}

    \begin{subfigure}[b]{0.48\textwidth}
        \centering
        \includegraphics[width=\linewidth]{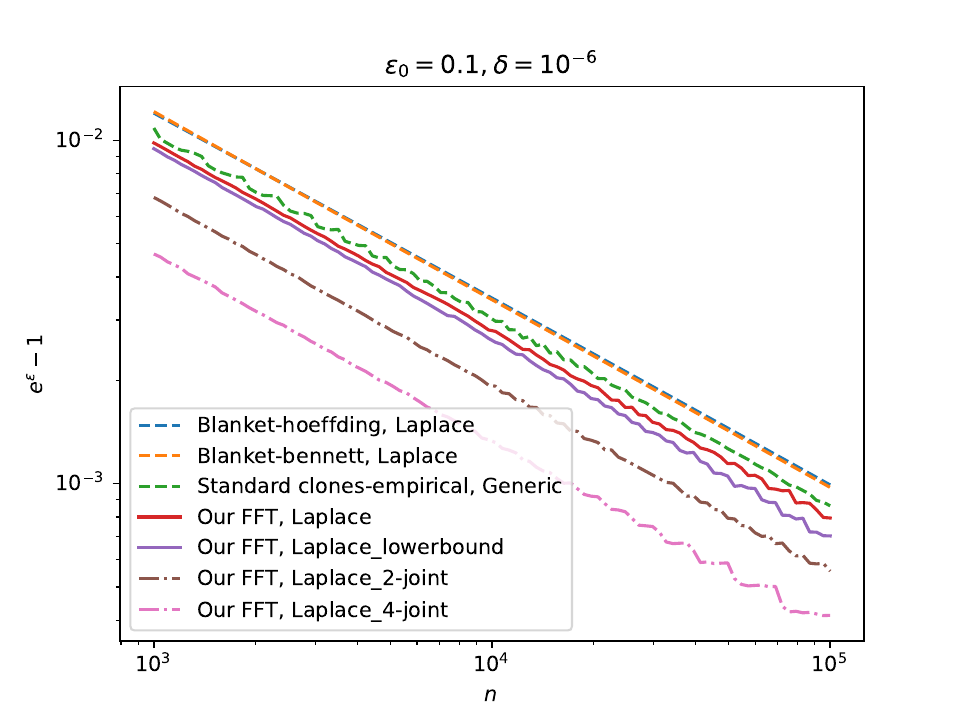} 
        \caption*{Laplace, low $\varepsilon_0$}
    \end{subfigure}
    \hfill
    \begin{subfigure}[b]{0.48\textwidth}
        \centering
        \includegraphics[width=\linewidth]{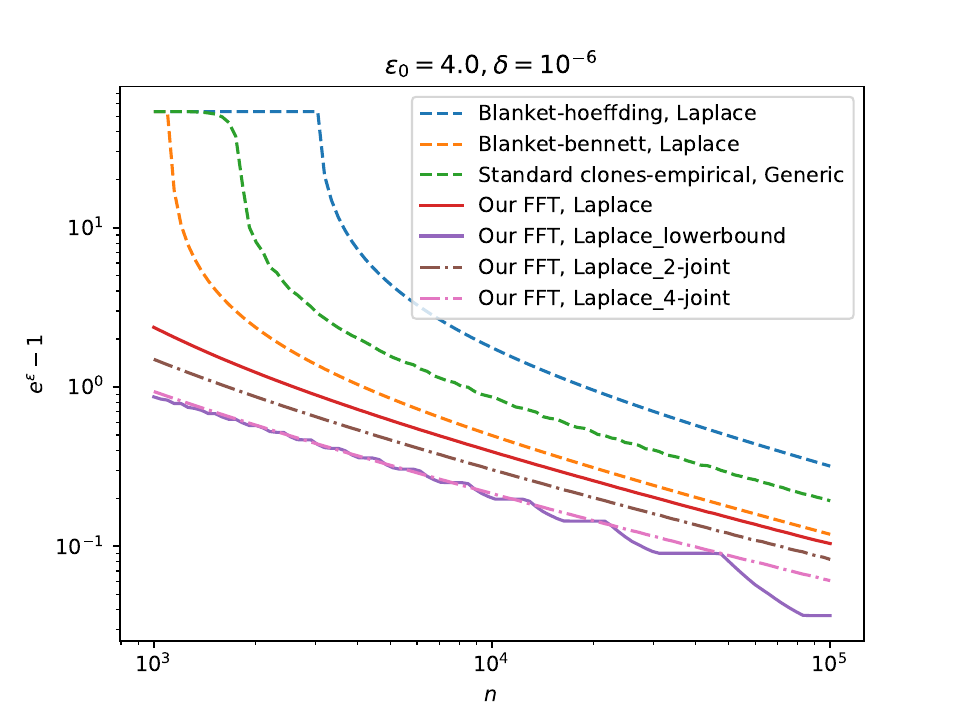} 
        \caption*{Laplace, large $\varepsilon_0$}
    \end{subfigure}
    
    \caption{Experimental Results: OUE and Laplace}
    \label{fig_exp2}
\end{figure*}

\begin{figure*}[ht]
    \centering

    \begin{subfigure}[b]{0.48\textwidth}
            \centering
    \includegraphics[width=\linewidth]{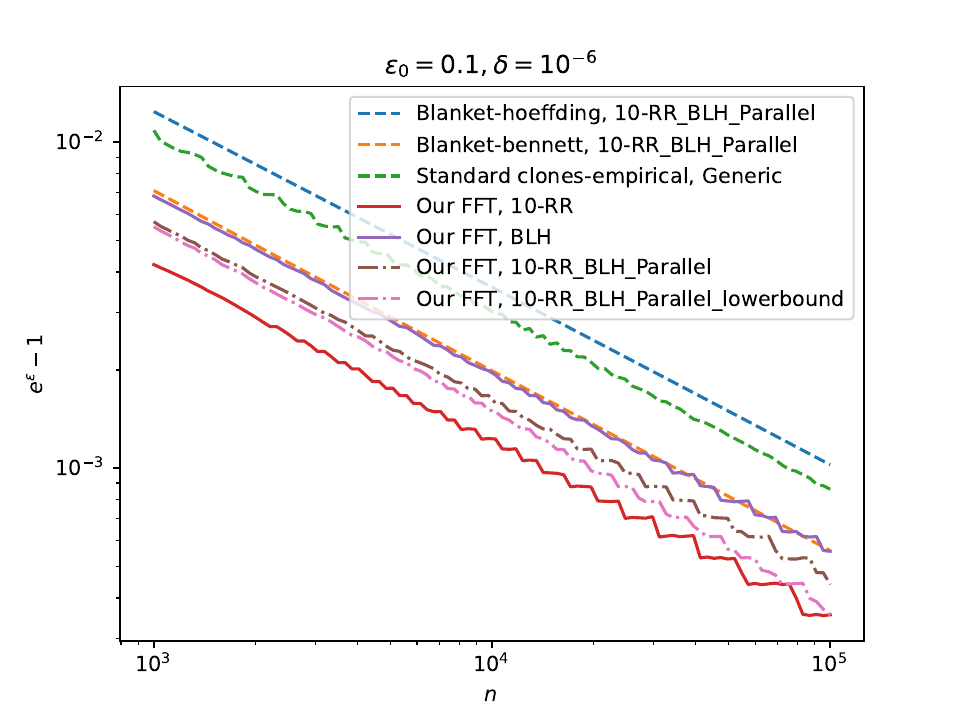}
    \caption*{$\frac{1}{2}$RR+$\frac{1}{2}$BLH, low $\varepsilon_0$}
    \end{subfigure}
    \hfill
    \begin{subfigure}[b]{0.48\textwidth}
            \centering
    \includegraphics[width=\linewidth]{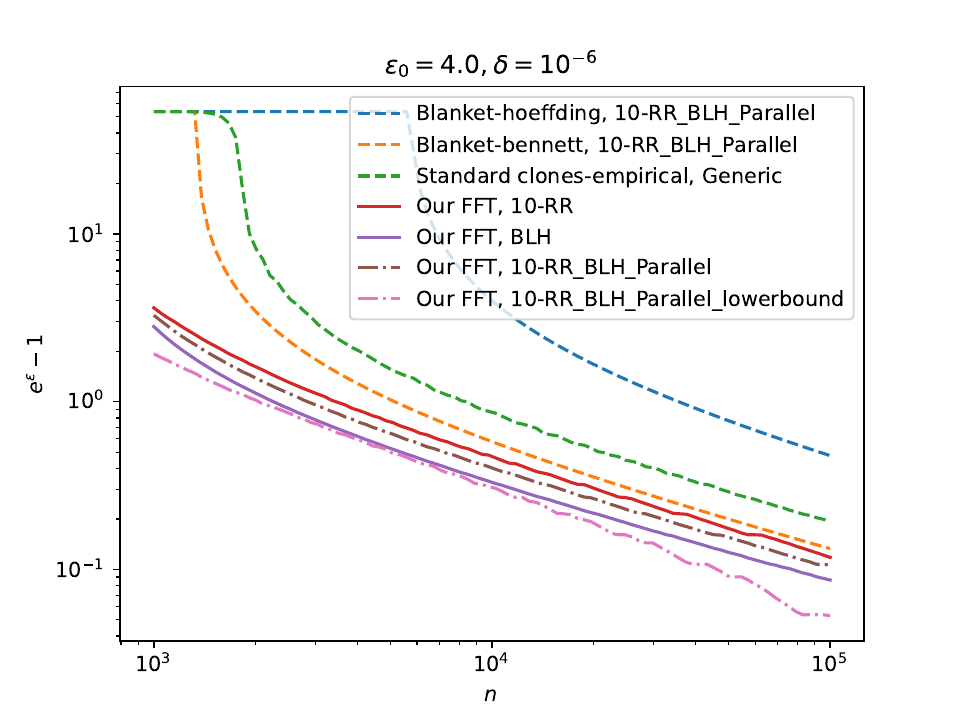}
        \caption*{$\frac{1}{2}$RR+$\frac{1}{2}$BLH, large $\varepsilon_0$}
    \end{subfigure}
    
    \caption{Experimental Results: Parallel Composition}
    \label{fig_exp3}
\end{figure*}

\clearpage

\begin{table}[t]
\centering
\caption{Frequency estimation under 10-RR ($n=1000,\delta=10^{-6}$): each entry reports $(\varepsilon_0,\ \ell_2\text{ error})$ required to achieve target privacy level $\varepsilon$.}
\label{table20}
\footnotesize
\begin{tabular}{c|cccccc}
\toprule
$\varepsilon$ 
& 0.01 
& 0.05 
& 0.10 
& 0.20 
& 0.50 
& 1.00 
\\
\midrule
\makecell{Blanket\\ (Bennett) }
& $(0.170,\ 1.598)$ 
& $(0.620,\ 0.366)$
& $(1.000,\ 0.199)$
& $(1.510,\ 0.108)$
& $(2.370,\ 0.051)$
& $(3.060,\ 0.033)$
\\
Our FFT 
& $(0.210,\ 1.285)$
& $(0.730,\ 0.295)$
& $(1.150,\ 0.163)$
& $(1.700,\ 0.089)$
& $(2.650,\ 0.042)$
& $(3.510,\ 0.025)$
\\
\bottomrule
\end{tabular}
\label{tab:exp_freq}
\end{table}

\begin{minipage}[t]{0.45\textwidth}
  \begin{figure}[H]
    \centering
    \includegraphics[width=\linewidth]{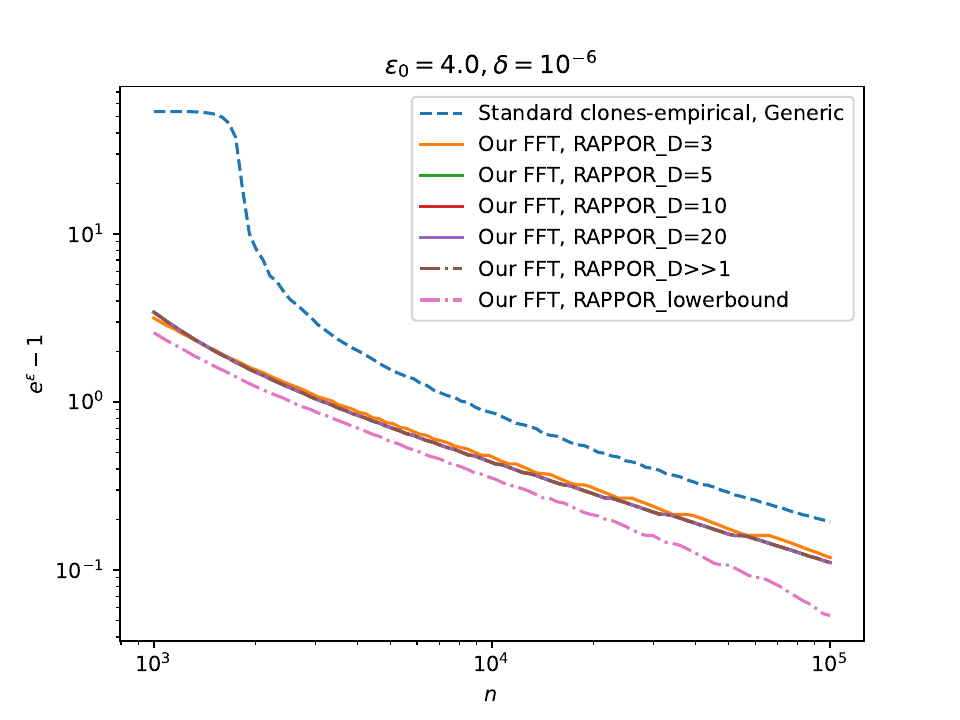}
    \caption{Upper Bounds for RAPPOR with Varying Domain Size $D$}
    \label{varying_D}
  \end{figure}
\end{minipage}
\hfill
\begin{minipage}[t]{0.45\textwidth}
  \begin{figure}[H]
    \centering
    \includegraphics[width=\linewidth]{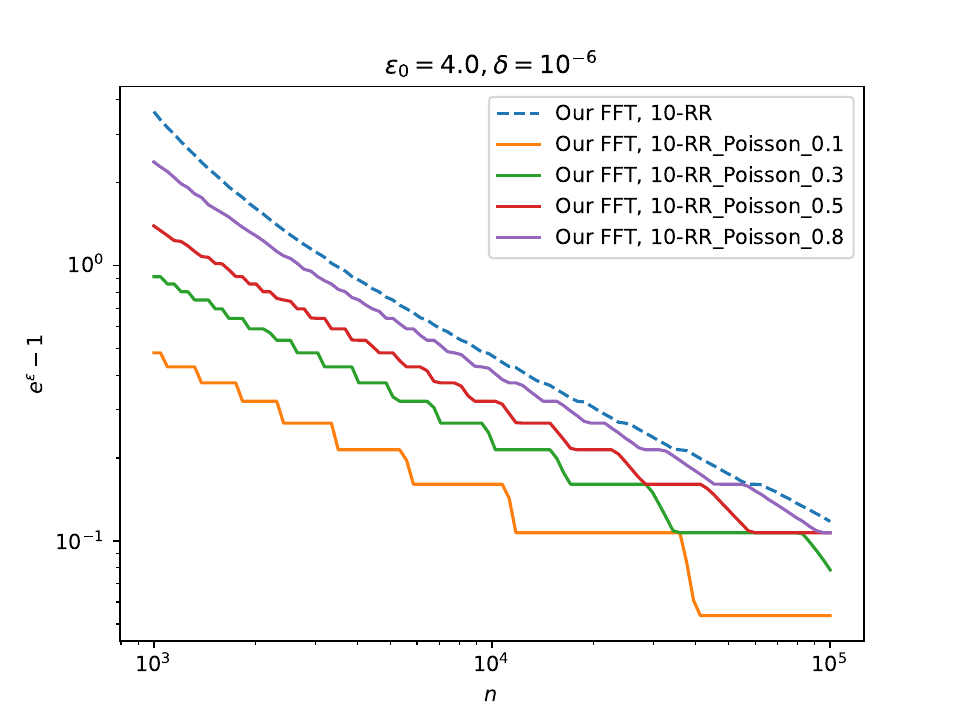}
    \caption{Upper Bounds for 10-RR under Poisson Subsampling and Shuffling}
    \label{fig:poisson}
  \end{figure}
\end{minipage}

\vspace{1em}

\begin{minipage}[t]{0.45\textwidth}
  \begin{figure}[H]
    \centering
    \includegraphics[width=\linewidth]{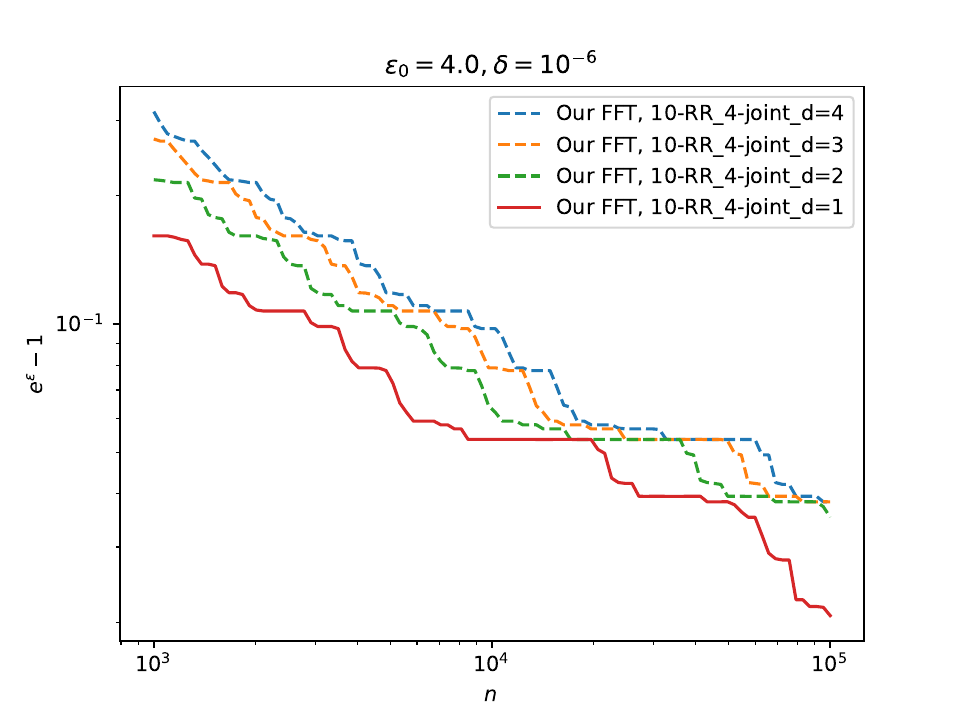}
    \caption{Upper Bounds for Joint Composition under the adjacency relation based on Hamming distance $d$}
    \label{fig_joint_varying_d}
  \end{figure}
\end{minipage}
\hfill
\begin{minipage}[t]{0.4\textwidth}
  \begin{figure}[H]
    \centering
    \includegraphics[width=\linewidth]{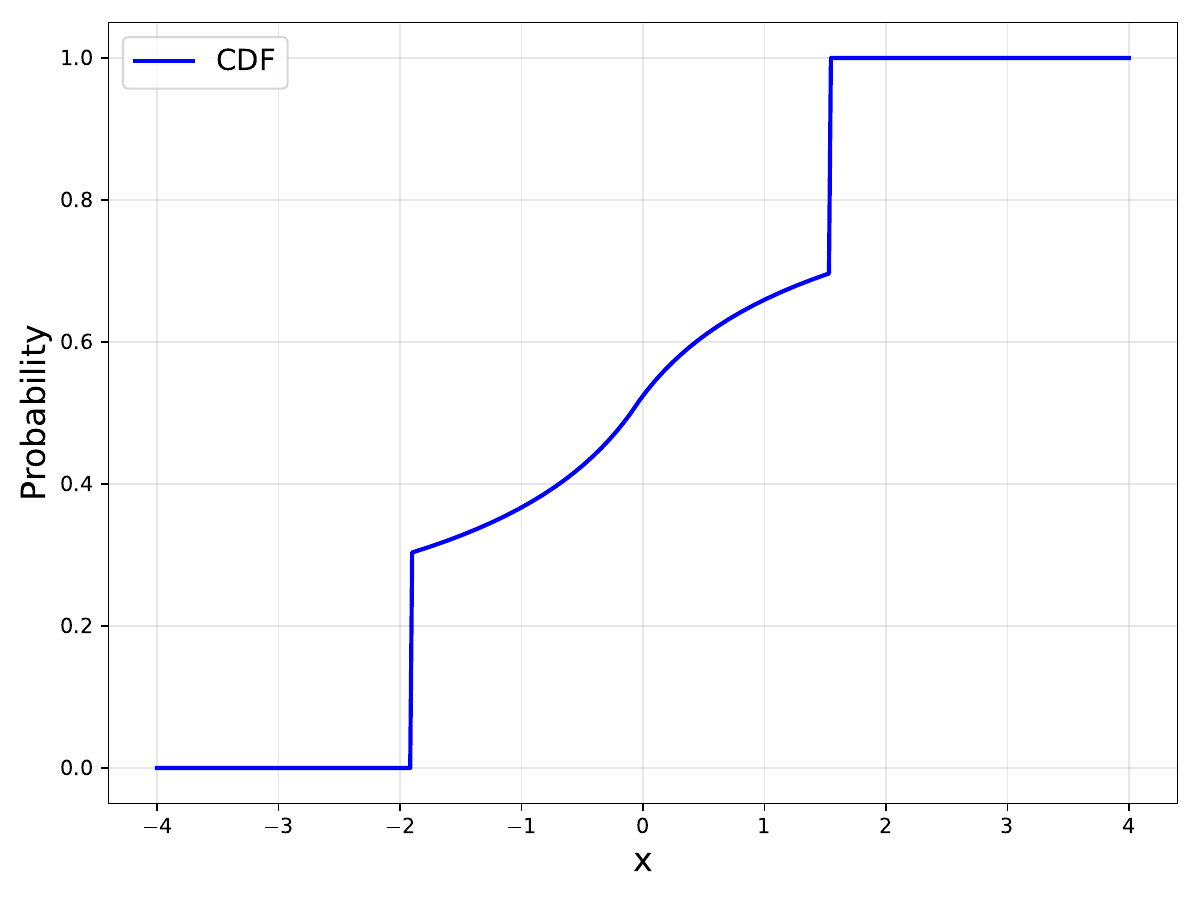}
    \caption{Cumulative distribution of \( L^{0,1}_\varepsilon \) for the Laplace mechanism with \( \varepsilon_0 = 1.0 \) and \( \varepsilon = 0.1 \)}
    \label{fig_laplace_cdf}
  \end{figure}
\end{minipage}

\end{document}